\documentclass[a4paper,11pt]{article}
\usepackage{amsfonts,amsthm,amssymb}
\usepackage{authblk}
\newcommand\COMP{\hbox{C\kern -.58em {\raise .54ex \hbox{$\scriptscriptstyle |$}}
\kern-.55em {\raise .53ex \hbox{$\scriptscriptstyle |$}} }}
\newcommand\NN{\hbox{I\kern-.2em\hbox{N}}}
\newcommand\RR{\hbox{I\kern-.2em\hbox{R}}}
\newcommand\sRR{{\it \hbox{I\kern-.2em\hbox{R}}}}
\newcommand\QQ{\hbox{I\kern-.53em\hbox{Q}}}
\newcommand\PP{\hbox{I\kern-.53em\hbox{P}}}
\newcommand\EE{\hbox{I\kern-.53em\hbox{E}}}
\newcommand\ZZ{{{\rm Z}\kern-.28em{\rm Z}}}
\newcommand\be{\begin{equation}}
\newcommand\ee{\end{equation}}
\newtheorem{theorem}{Theorem}[section]

\newtheorem{proposition}[theorem]{Proposition}
\newtheorem{remark}[theorem]{Remark}

\newtheorem{lemma}[theorem]{Lemma}

\newtheorem{definition}[theorem]{Definitions}
\newtheorem{corollary}[theorem]{Corollary}

\newcommand{\mm}{m}

\newcommand{\mg}{\widehat M}
\newcommand{\sh}{\widehat S}

\newcommand{\E}{\mathbb E}
\def\comf#1{\left ( #1\right )\!^{p,\mathbb F}}
\def\comg#1{\left ( #1\right )\!^{p,\mathbb G}}
\def\prof#1{ \phantom{l}^{p,\mathbb F}\!\left ( #1\right )}
\def\prog#1{ \phantom{l}^{p,\mathbb G}\!\left ( #1\right )}

\def \Rbrack {[\![}
\def \Lbrack {]\!]}

\def\crof#1{\langle #1\rangle^{\mathbb F}}

\def\dro#1{[ #1 ]}
\def \Rbrack {[\![}
\def \Lbrack {]\!]}

\newcommand{\Lopt}{L}
\newcommand{\is}{\centerdot}

\begin{document}
\title{Non-Arbitrage up to Random Horizon  for Semimartingale Models\thanks{The research of Tahir Choulli  and Jun Deng is supported financially by the
Natural Sciences and Engineering Research Council of Canada,
through Grant G121210818. The research of Anna Aksamit and Monique Jeanblanc is supported
by Chaire Markets in transition, French Banking Federation.}}

\author[2]{Anna  Aksamit}
\author[1]{Tahir Choulli\thanks{corresponding author, {Email: tchoulli@ualberta.ca} }}
\author[1]{Jun Deng}
\author[2]{Monique Jeanblanc}
\affil[1]{Mathematical and Statistical Sciences Depart., University of Alberta, Edmonton, Canada}
\affil[2]{Laboratoire Analyse et Probabilit\'es,
Universit\'e d'Evry  Val d'Essonne,  Evry, France \newline  \newline
\textbf{This version develops in details the ideas and the results of the previous version and fixes a glitch in the quasi-left-continuous case.}}



\maketitle

\begin{abstract}
This paper addresses the question of how an arbitrage-free
semimartingale model is affected when stopped at a random horizon.
We focus on No-Unbounded-Profit-with-Bounded-Risk (called NUPBR
hereafter) concept, which is also known in the literature as the
first kind of non-arbitrage. For this non-arbitrage notion, we
obtain two principal results. The first result lies in describing
the pairs of market model and random time for which the resulting
stopped model fulfills NUPBR condition. The second main result
characterises the random time models that preserve the NUPBR
property after stopping for any market model. These results are
elaborated in a very general market model, and we also pay
attention to some particular and practical models. The analysis
that drives these results is based on new stochastic developments
in semimartingale theory with progressive enlargement.
Furthermore, we construct explicit martingale densities
(deflators) for some classes of local martingales when stopped at
random time.
\end{abstract}


\section{Introduction}

Since the seminal work of Arrow-Debreu for discrete markets (see
\cite{arrow/debreu}, \cite{duffiebook2010}, \cite{debreu1959}  and
\cite{dybvigross1989}) where the authors stated what is known in
the financial literature as the fundamental theorem of asset
pricing (FTAP hereafter), arbitrage becomes one of the fundamental
concepts in modern finance. This importance was strengthen by
establishing the link between  arbitrage and asset pricing
theories through the Arrow-Debreu model (see \cite{duffiebook2010}
Chapter 7), the Black and Scholes formula (see
\cite{blackscholes1973}), and the Cox and Ross linear pricing
model (see \cite{coxross76}). These works have been formalized in
general framework by many researchers such as Duffie,  Harrison,
Huang, Kreps and Pliska (see \cite{harrisonkreps1979},
\cite{harrisonpliska1981} and \cite{duffiehuang86}). Other links
between arbitrage theory and other financial/economical concepts
have been discovered, elaborated, and explored further such as
stochastic dominance, market's equilibrium, market's viability,
portfolio analysis, num\'eraire portfolio,and so on. Arbitrage
plays crucial role in the analysis of securities markets, because
its effect is to bring prices to
fundamental values and to keep markets efficient.\\

\noindent In both mathematical finance and financial economics, arbitrage
(or equivalently absence of arbitrage) has numerous definitions such as no-free-lunch
(NFL), no-free-lunch-with-vanishing-risk (NFLVR), cheap thrill, free snack,
No-Unbounded-Profit-with-Bounded-Risk (NUPBR), and so on (see \cite{debreu1959},
\cite{delbaenschameyer1994}, \cite{delbaenschameyer1998}, \cite{kabanov}, \cite{kreps1981},
 \cite{loewensteinwillard00} and the references therein to cite a few). It is worthwhile to
 mention that the majority of these arbitrage concepts coincides for discrete markets
 (i.e, markets with finite number of scenarios and finite number of trading times),
 and at some extent in discrete-time markets with finite horizon. Philosophically, an
 arbitrage opportunity is a transaction with no cash outlay that results in a sure profit.
 The  mathematical formulation of this philosophy varies substantially from  a model to another.\\

\noindent Recently there has been an upsurge interest in studying
the effect of additional information on the NFLVR concept as well
as on the utility maximization problem. For these topics we refer
the reader to  \cite{amendingerimkellerschweizer98},
\cite{kohatsusulem06},  \cite{pikovskykaratzas96} among others.
Herein, we focus only on the NUPBR concept  for two main reasons.
Firstly, the NUPBR property is the non-arbitrage concept that is
intimately related to the weakest forms of markets' viability (see
\cite{choulli/deng/ma} and \cite{kardaras12} for details about
this issue, and \cite{aksamit/choulli/deng/jeanblanc} for many
examples of market models violating NFLVR and fulfilling NUPBR).
It has been clear recently that for a model violating the NUPBR,
the optimal portfolio will not exist even locally, and the pricing
rules fail as well. Secondly, due to \cite{Takaoka} and again
\cite{choulli/deng/ma}, the NUPBR property  is mathematically very
attractive and possesses the 'dynamic/localization' feature that
the NFLVR and other arbitrage concepts lack to possess. By
localization feature, we mean that if the property holds locally
(i.e. the property holds for the stopped models with a sequence of
stopping times that increases to infinity), then it holds
globally.\\

\noindent  In this paper we consider a general  semimartingale
model $S$ satisfying the NUPBR property under the ``public
information" and an
arbitrary random time $\tau$ and  we answer   to the following questions:\\

\centerline{ \hspace*{0.5cm} For which pairs   $(S,\tau)$, does
the NUPBR property hold for $S^{\tau}$?  \hfill   {{\bf (P1) } } }
\vspace{0.15cm}

\noindent In Theorem \ref{main6} we characterize  pairs of initial
market $S$ and   of random time $\tau$, for which {\bf (P1) } has
a positive answer. Our second
main question consists of\\

\centerline{ \hspace*{0.5cm} For which $\tau$, is   NUPBR
preserved for any $S$ after stopping at $\tau$?\hfill   {{\bf (P2) }
} }

\noindent To deepen our understanding of the precise interplay
between the initial market model and the random time model, we
address these two principal questions separately  in the case of
quasi-left-continuous models, and then in the case  of  thin
processes with predictable jumps. Afterwards, we combine the two
cases and state the results for the most general framework. The
results for the quasi-left-continuous models are Theorem
\ref{main2} and Proposition \ref{main1}, where the questions {\bf
{(P1)}} and {\bf {(P2)}}  are fully answered respectively. For the
case of thin processes with predictable jumps, our main result is
Theorem \ref{main4}. Then, the general case follows by splitting
the process $S$ into a quasi
continuous process and a  thin  process with predictable jumps.\\

\noindent This problem was studied  in the literature, in the
particular case of continuous filtration, under the hypothesis
that $\tau$ avoids $\mathbb F$-stopping times in \cite{fjs} and
that the market is complete. Many explicit examples can be found
in \cite{aksamit/choulli/deng/jeanblanc}. It is also possible to
derive some of the main results using a new optional decomposition
formula, established in \cite{aksamit/choulli/jeanblanc}. After
that  a  first version  of our results has been available some
months ago, Acciaio et al. \cite{afk}  proved some of our results
and ideas, using a different method.
\\

\noindent This paper is organized as follows. The next section
(Section \ref{mainresults}) presents our  main results in
different contexts, and discusses their meaning and/or  their
economical interpretations and their consequences as well. Section
\ref{newstochastic} develops new stochastic results, -that are the
key mathematical ideas behind the answers to {\bf {(P1)}}-{\bf
{(P2)}}.  Section \ref{SectionQLCprocesses} gives an explicit form
for the deflator in the case where $S$ is quasi-left continuous.
Section \ref{proofs} contains the proofs of the main theorems
announced, without proofs, in Section \ref{mainresults}. The paper
concludes with an Appendix, where some classical results on the
predictable characteristics of a semimartingales and other related
results are recalled. Some technical proofs are also postponed to
the Appendix, for the ease of the reader.


%
%

\section{Main Results and their Interpretations}\label{mainresults}

This section is devoted to the presentation of our main results
and their immediate consequences. To this end, we start specifying
our mathematical setting and the economical concepts that we will
address.
\subsection{Notations and Preliminary Results on NUPBR}
 \noindent We consider a stochastic basis $(\Omega, {\cal G},  {\mathbb  F }=({\cal F}_t)_{t\geq 0},  P)$,  where
${\mathbb  F }$ is a filtration
satisfying the usual hypotheses (i.e., right continuity and completeness), and $ {\cal F}_\infty
 \subseteq {\cal {G}}$. Financially speaking, the filtration $\mathbb F$ represents the flow of public information through time. On this basis, we consider an arbitrary but fixed $d$-dimensional c\`adl\`ag semimartingale $S$. This represents the discounted price processes of $d$-stocks, while the riskless asset's price is assumed to be constant.\\
 Beside  the initial  model
 $\left(\Omega,{\cal G}, \mathbb F, P, S\right)$, we  consider a random time $\tau$, i.e., a non-negative ${\cal G}$-measurable random variable. To  this random time, we associate  the process $D$ and the
filtration $\mathbb G$ given by
$$\label{AandfiltrationG}
D:= I_{\Rbrack\tau,+\infty\Rbrack},\ \ \ \mathbb G=\left({\cal G}_t\right)_{t\geq 0},\ \ \ {\cal G}_t =
\displaystyle\bigcap_{s>t}\Bigl({\cal F}_s\vee \sigma(D_u, u\leq s)\Bigr).$$
The filtration $\mathbb G$ is the smallest
right-continuous filtration which contains ${\mathbb  F }$ and makes $\tau$ a stopping time. In the probabilistic literature, $\mathbb G$ is called the progressive enlargement of $\mathbb F$ with $\tau$. In addition to $\mathbb G$ and $D$, we associate to $\tau$ two important
$\mathbb F$-supermartingales given by
\begin{equation}\label{ZandZtilde}
Z_t := P\left(\tau >t\ \big|\ {\cal F}_t\right)\   \mbox{and}\ \ \ \widetilde Z_t:=P\left(\tau\geq t\ \Big|\ {\cal F}_t\right).
\end{equation}
The supermartingale $Z$ is right-continuous with left limits and
coincides with the $\mathbb F$-optional projection of $I_{\Lbrack
0, \tau\Rbrack}$, while $\widetilde Z$ admits right limits and
left limits only and is the $\mathbb F$-optional projection of
$I_{\Lbrack 0, \tau\Lbrack}$.  The decomposition of $Z$ leads to
an important $\mathbb F$-martingale  $m$,   given by
 \begin{equation}\label{processmZ}
m := Z+D^{o,\mathbb F},\end{equation}
 where $D^{o,\mathbb F}$ is the $\mathbb F$-dual optional projection of $D$
 (See  \cite{Jeu} for more details).\\

\noindent In what follows, $\mathbb H$ is a filtration satisfying
the usual hypotheses and $Q$ a probability measure  on the
filtered probability space $(\Omega, \mathbb H)$. The set of
martingales for the filtration $\mathbb H$ under $Q$ is denoted by
${\cal M}(\mathbb H, Q)$. When $Q=P$, we simply denote ${\cal
M}(\mathbb H)$. As usual, ${\cal A}^+(\mathbb H)$ denotes the set
of increasing,
right-continuous, $\mathbb H$-adapted and integrable processes.\\
\noindent If ${\cal C}(\mathbb H)$ is a class of $\mathbb H$
adapted processes,
 we denote by ${\cal C}_0(\mathbb H)$ the set of processes $X\in {\cal C}(\mathbb H)$ with $X_0=0$, and
 by ${\cal C}_{loc}$
 the set  of processes $X$
 such that there exists a sequence $(T_n)_{n\geq 1}$ of $\mathbb H$-stopping times  that increases to $+\infty$ and the
 stopped processes $X^{T_n}$ belong
to ${\cal C}(\mathbb H)$. We put $ {\cal C}_{0,loc}(\mathbb H)={\cal
C}_0(\mathbb H)\cap{{\cal
C}}_{loc}(\mathbb H)$.\\
\noindent For a process $K$ with $\mathbb H$-locally integrable variation, we denote by $K^{o,\mathbb H}$  its dual optional projection. The dual predictable projection of $K$ (also called the $\mathbb H$-predictable dual projection) is denoted   $K^{p,\mathbb H}$. For a process $X$, we denote $^{o, \mathbb H \!} X$ (resp.$\,^{p, \mathbb H\!} X$ ) its optional (resp. predictable) projection with respect to $\mathbb H$.\\
\noindent For an $\mathbb H$- semi-martingale $Y$, the set $L(Y, \mathbb H)$ is the set of $\mathbb H$ predictable processes
integrable w.r.t. $Y$ and  for $H\in L(Y, \mathbb H )$, we denote $H\centerdot  Y_t:= \int_0^t H_sdY_s$.\\
As usual, for a process $X$ and a random time $\vartheta$, we
denote by $X^\vartheta$ the stopped process. To distinguish the
effect of filtration, we will denote $\langle ., .\rangle^{\mathbb
F},$ or $\langle ., . \rangle^{\mathbb G}$
 the sharp bracket (predictable covariation process) calculated in the filtration ${\mathbb F}$ or  ${\mathbb G},$ if confusion may rise. We recall that, for general semi-martingales $X$ and $Y$, the
sharp bracket is (if it exists) the dual predictable projection of
the covariation process $[X,Y]$.\\

\noindent We introduce the non-arbitrage notion that will be addressed in this paper.

\begin{definition}\label{DefinitionofNUPBR} An $\mathbb H$-semimartingale
$X$ satisfies the {\it No-Unbounded-Profit-with-Bounded-Risk}
condition under $(\mathbb H,Q)$ (hereafter called NUPBR($\mathbb
H, Q)$) if for any $T\in (0,+\infty)$ the set
$$\label{boundedset}
{\cal K}_{T} (X,\mathbb H):=\displaystyle \Bigl\{(H\centerdot
 S)_T \ \big|\ \
\ H \in L(X,\mathbb H),\ \mbox{and}\ H\is X\geq -1\ \Bigr\}$$ is
bounded in probability under $Q$. When $Q\sim P$, we simply write,
with an abuse of language, $X$ satisfies NUPBR$(\mathbb H)$.
\end{definition}

 \begin{remark}\label{remark1}
 (i) It is important to notice that this definition for NUPBR condition appeared first in \cite{karadarasHonest2010} (up to our knowledge), and it {  differs when the time horizon is
 infinite} from that of the literature given in Delbaen and Schachermayer \cite{DelbaenSchachermayer1994},
 Kabanov \cite{kabanov} and Karatzas and Kardaras\cite{KaratzasKardaras2007}. It is  obvious that, when the horizon is deterministic and finite, the current
  NUPBR  condition coincides with that of the literature. We could name the current NUPBR as NUPBR$_{loc}$, but for the sake of simplifying notation, we opted for the usual terminology.\\
 (ii) In general, when the horizon is infinite, the NUPBR condition
 of the literature implies the NUPBR  condition defined above. However, the reverse implication may not hold in general. In fact if we consider $S_t=\exp(W_t+t),\ t\geq 0$, then it is clear that $S$ satisfies our
  NUPBR$(\mathbb H)$, while the NUPBR$(\mathbb H)$ of the literature is violated. To see this last claim, it is enough to remark that $$\lim_{t\longrightarrow +\infty} (S_t-1)=+\infty\ \  \ P-a.s.\ \ \  S^t-1=H\is S\geq -1\ \ \ H:=I_{\Lbrack 0,t\Lbrack} .$$
 \end{remark}

\noindent The following proposition slightly generalizes Takaoka's
results obtained for a finite horizon (see Theorem 2.6 in
\cite{Takaoka}) to our NUPBR context.
\begin{proposition}\label{charaterisationofNUPBRloc} Let $X$ be an $\mathbb H$-semimartingale. Then the
 following assertions are equivalent.\\
{\rm{(a)}}  $X$ satisfies NUPBR$(\mathbb H)$.\\
  {\rm{(b)}} There exist a positive $\mathbb H$-local martingale, $Y$ and an
  $\mathbb H$-predictable process $\theta$ satisfying $0<\theta\leq 1$ and $Y(\theta\is X)$
   is a local martingale.
\end{proposition}
\begin{proof} The proof of the implication (b)$\Rightarrow$ (a)
is based on \cite{Takaoka} and is  omitted. Thus, we focus on
proving the reverse implication and suppose that assertion (a)
holds. Therefore, a direct application of Theorem 2.6 in
\cite{Takaoka} to each $(S_{t\wedge n})_{t\geq 0}$, we obtain the
existence of a positive $\mathbb H$-local martingale $Y^{(n)}$ and
an $\mathbb H$-predictable process $\theta_n$ such that
$0<\theta_n\leq 1$ and $Y^{(n)}(\theta_n\is S^n)$ is a local
martingale. Then, it is obvious that the process
$$
N:=\sum_{n=1}^{+\infty} I_{\Lbrack n-1,
n\Lbrack}(Y^{(n)}_{-})^{-1}\is Y ^{(n)}$$ is  a local martingale
and $Y:={\cal E}(N)>0$. On the other hand, the $\mathbb
H$-predictable process $\theta:=\sum_{n\geq 1} I_{\Lbrack n-1,
n\Lbrack}\theta_n$ satisfies $0<\theta\leq 1$ and $Y(\theta\is S)$
is a local martingale. This ends the proof of the
proposition.\end{proof}

\noindent For any $\mathbb H$-semimartingale $X$, the local martingales fulfilling the assertion (b) of Proposition \ref{charaterisationofNUPBRloc} are
 called   $\sigma$-martingale densities for $X$. The set of these $\sigma$-martingale densities will be denoted throughout the paper by
\begin{equation}\label{LFS}
{\cal L}({\mathbb H},X):=\left\{ Y\in {\cal M}_{loc}(\mathbb
H)\big|\ Y>0,\  \exists \theta \in {\cal {P}}(\mathbb H),\,
0<\theta\leq 1,\
 Y(\theta \centerdot X)\in {\cal M}_{loc}(\mathbb
H) \right\}\end{equation} where, as usual, ${\cal {P}}(\mathbb H)$
stands for predictable processes.  We state, without proof, an
obvious lemma.
\begin{lemma}\label{LY} For any $\mathbb H$-semimartingale $X$ and any $Y\in {\cal L}({\mathbb H},X)$, one has
$ ^{p,\mathbb H} (Y  \vert \Delta X \vert )<\infty$ and  $
^{p,\mathbb H}(Y \Delta X
  )=0$
\end{lemma}

\begin{remark} Proposition \ref{charaterisationofNUPBRloc} implies that
for any process $X$ and any {\bf finite} stopping time $\sigma$,
the two concepts of NUPBR$(\mathbb H)$ (the current concept and
the one of the literature) coincide for $X^{\sigma}$.
\end{remark}

\noindent Below, we  prove that,  in the  case of infinite
horizon, the current NUPBR  condition  is stable under
localization, while this is not the case  for  the NUPBR condition
 defined in the literature.

\begin{proposition}\label{NUPBRLocalization}
Let $X$ be an $\mathbb H$ adapted   process. Then, the following assertions are equivalent.\\
{\rm{(a)}}  There exists a sequence   $(T_n)_{n\geq 1}$ of
$\mathbb H$-stopping times that increases to $+\infty$, such that
for each $n\geq 1$, there exists a probability $Q_n$ on $(\Omega,
{\mathbb H}_{T_n})$ such that $Q_n\sim P$ and $X^{T_n}$ satisfies
 NUPBR$(\mathbb H)$ under $Q_n$.\\
{\rm{(b)}}  $X$ satisfies   NUPBR$(\mathbb H)$.\\
{\rm{(c)}} There exists an $\mathbb H$-predictable process $\phi$,
such that $0<\phi\leq 1$ and  $(\phi\is X)$ satisfies
NUPBR$(\mathbb H)$.
\end{proposition}
\begin{proof}
The proof for (a)$\Longleftrightarrow$(b) follows from the
stability of NUPBR condition for a finite horizon under
localization which is due to \cite{Takaoka} (see also
\cite{Choulli2012} for further discussion about this issue), and
the fact that the NUPBR condition is stable under any equivalent
probability change.\\The proof of (b)$\Rightarrow$(c) is trivial
and is omitted. To prove the reverse, we assume that (c) holds.
Then   Proposition \ref{charaterisationofNUPBRloc}  implies the
existence of an
  $\mathbb H$-predictable process $\psi$ such that
$0<\psi\leq 1$ and a positive $\mathbb H$-local martingale
$Z={\cal E}(N)$ such that $Z(\psi\phi\is X)$ is a local
martingale. Since $\psi\phi$ is predictable and $0<\psi\phi\leq
1$, we deduce that $S$ satisfies   NUPBR$(\mathbb H)$. This ends
the proof of the proposition.
\end{proof}

We end this section with a simple, but useful result for
predictable process with finite variation.
\begin{lemma}\label{NUPBRforPredictableProcesses}
Let $X$ be an $\mathbb H$-predictable process with finite variation. Then $X$ satisfies NUPBR$(\mathbb H)$ if and only if $X\equiv X_0$ (i.e. the process $X$ is constant).\end{lemma}
\begin{proof}
It is obvious that if $X\equiv X_0$, then $X$ satisfies
NUPBR$(\mathbb H)$. Suppose that $X$ satisfies NUPBR$(\mathbb H)$.
Consider  a positive $\mathbb H$-local martingale $Y$, and an
$\mathbb H$-predictable process $\theta $ such that $0<\theta\leq
1$ and $Y(\theta\is X)$ is a local martingale. Let $(T_n)_{n\geq
1}$ be a sequence of $\mathbb H$-stopping times that increases to
$+\infty$ such that $Y^{T_n}$ and $Y^{T_n}(\theta\is X)^{T_n}$ are
true martingales. Then, for each $n\geq 1$,  define
$Q_n:=\left(Y_{T_n}/Y_0\right)\cdot P.$ Since $X$ is predictable,
then  $(\theta\is X)^{T_n}$ is also predictable with finite
variation and is a $Q_n$-martingale. Thus, we deduce that
$(\theta\is X)^{T_n}\equiv 0$ for each $n\geq 1$. Therefore, we
deduce that $X$ is constant (since $X^{T_n}-X_0=\theta^{-1}\is
(\theta\is X)^{T_n}\equiv 0$). This ends the proof of the
lemma.\end{proof}



\subsection{The Quasi-Left-Continuous Processes}
In this subsection, we present our two main results on the NUPBR condition
under stopping at $\tau$ for quasi-left-continuous processes.
The first result  consists of characterizing the pairs
$\left(S,\tau\right)$ of market   and random time models,  for
which $S^{\tau}$   fulfills the NUPBR condition. The second  result
focuses on determining the models of random times $\tau$  such that,  for any semi-martingale $S$ enjoying NUPBR$(\mathbb F)$, the stopped process $S^\tau$ enjoys NUPBR$(\mathbb G)$ .
\\

\noindent We start by recalling some general notation. For any filtration
$\mathbb H$,   we denote
\begin{equation}\label{sigmaFields}
\widetilde {\cal O}(\mathbb H):={\cal O}(\mathbb H)\otimes {\cal
B}({\mathbb R}^d),\ \ \ \ \ \widetilde{\cal P}(\mathbb H):= {\cal
P}(\mathbb H)\otimes {\cal B}({\mathbb R}^d),
\end{equation}
where ${\cal B}({\mathbb R}^d)$ is the Borel $\sigma$-field on
${\mathbb R}^d$. The jump measure of $S$ is denoted by $\mu$, and
given by
$$
\mu(dt,dx)=\sum_{u>0} I_{\{\Delta S_u \neq 0\}}\delta_{(u,\Delta
S_u)}(dt,dx),$$ For a product-measurable functional $W\geq 0$ on
$\Omega\times[0,+\infty[\times{\mathbb R}^d$, we denote
$W\star\mu$ (or sometimes, with abuse of notation $W(x)\star\mu$)
the process
\begin{equation}\label{Wstarmu}
(W\star\mu)_t:=\int_0^t \int W(u,x)\mu(du,dx)=\sum_{0<u\leq t}
W(u,\Delta S_u) I_{\{ \Delta S_u\not=0\}}.\end{equation} Also on
$\Omega\times[0,+\infty[\times{\mathbb R}^d$, we define the
measure $M^P_{\mu}:=P\otimes\mu$ by $\int W
dM^P_{\mu}:=E\left[W\star\mu(\infty)\right]$ (when the integrals
are well defined). The conditional ``expectation" given $
\widetilde{\cal P}(\mathbb H)$ of a product-measurable functional
$W$, is the unique $ \widetilde{\cal P}(\mathbb H)$-measurable
functional $\widetilde W$ satisfying
$$
E\left[W I_{\Sigma}\star\mu(\infty)\right]=E\left[{\widetilde W}
I_{\Sigma}\star\mu(\infty)\right],\ \ \ \mbox{for all}\
\Sigma\in\widetilde{\cal P}(\mathbb H).$$

 The following theorem gives a characterization  of $\mathbb F$-quasi-left
continuous processes that satisfy NUPBR$(\mathbb G)$ after stopping with $\tau$. The proof
of this theorem will be given in Subsection \ref{main2section}, while its statement is based on the following $\mathbb F$-semimartingale
\begin{equation}\label{Szero}
  S^{(0)} := xI_{\left\{\psi = 0<Z_{-}\right\}} \star\mu, \ \ \ \ \ \ \ \mbox{where}\ \ \ \  \
  \psi:= M_\mu^P\left( I_{\{\widetilde Z>0\}} \Big| \widetilde{\cal P}\left( \mathbb{F} \right)\right).
 \end{equation}

\begin{theorem}\label{main2}Suppose that $S$ is $\mathbb F$-quasi-left-continuous. Then, the following assertions are equivalent.\\
{\rm(a)} $S^{\tau}$ satisfies NUPBR$(\mathbb G)$.\\
{\rm(b)} For any $\delta>0$, the process
\begin{equation}\label{main2/1} I_{\{ Z_{-}\geq
\delta\}}\is\left(S-S^{(0)}\right)\ \ \ \ \mbox{satisfies
NUPBR($\mathbb F$).}\end{equation}
{\rm(c)} For any $ n\geq 1 $, the process $  (S-S^{(0)} )^{\sigma_n}$ satisfies
NUPBR$(\mathbb F)$, where
 $$
 \sigma_n:=\inf\{t\geq 0\ :\ Z_{t }<1/n\}.$$
\end{theorem}

\begin{remark}\label{remakrformain1}
 1) From assertion (c) one can understand that  the NUPBR$(\mathbb G)$ property for $S^{\tau}$ can be verified by checking
$S-S^{(0)}$ satisfies  NUPBR$(\mathbb F)$   up to $\sigma_{\infty}:=\sup_n\sigma_n$. This is also equivalent to   NUPBR$(\mathbb F)$ of the same process on the predictable sets $\{ Z_{-}\geq \delta\}$, $\delta>0$. \\
 2)  The functionals $\psi$ and $Z_{-}+f_m:=M^P_{\mu}(\widetilde Z\big|\widetilde{\cal P}(\mathbb F))$ satisfy
 \begin{equation}\label{psifm}
\{\psi=0\}=\{Z_{-}+f_m=0\}\subset\{ \widetilde Z=0\},\ \ \
M^P_{\mu}-a.e.\end{equation} Indeed, due to $\widetilde Z\leq
I_{\{ \widetilde Z>0\}}$, we have
$$0\leq Z_{-}+f_m=M^P_{\mu}\left(\widetilde Z\big|\ \widetilde{\cal P}(\mathbb F)\right)\leq \psi.$$
Thus, we get $\{\psi=0\}\subset\{ Z_{-}+f_m=0\}\subset\{\widetilde Z=0\}\ \ M^P_{\mu}-a.e.$ on the one hand. On the other hand, the reverse inclusion follows from
$$
0=M^P_{\mu}\left(I_{\{ Z_{-}+f_m=0\}}I_{\{\widetilde
Z=0\}}\right)=M^P_{\mu}\left(I_{\{ Z_{-}+f_m=0\}}\psi\right).$$ 3)
As a  result of   remark 2) above and $\{ \widetilde
Z=0<Z_{-}\}\subset \Rbrack\sigma_{\infty}\Lbrack$, we deduce that
$S^{(0)}$ is a c\`adl\`ag $\mathbb F$-adapted process with finite
variation with $var(S^{(0)})_{\infty}\leq \vert \Delta
S_{\sigma_{\infty}}\vert I_{\{\sigma_{\infty}<+\infty\}}.$
Furthermore, it can be written as $$S^{(0)}:=\Delta
S_{\sigma_{\infty}}I_{\{ \widetilde
Z_{\sigma_{\infty}}=0=\psi(\sigma_{\infty}, \Delta
S_{\sigma_{\infty}})\ \&\
Z_{\sigma_{\infty}-}>0\}}I_{\Rbrack\sigma_{\infty},+\infty\Rbrack}.$$
This proves the claim stated before Theorem \ref{main2} about the
process $S^{(0)}$.
\end{remark}

\noindent   The following corollary is useful for studying the problem ({\bf {P2}}), and it describes examples of $\mathbb F$-quasi-left-continuous model $S$ that fulfill
(\ref{Szero}) as well.

\begin{corollary}\label{corollary1} Suppose that $S$ is $\mathbb F$-quasi-left-continuous and satisfies NUPBR$(\mathbb F)$.  Then, the following assertions hold.\\
{\rm{(a)}} If $\left(S, S^{(0}\right)$ satisfies NUPBR$(\mathbb F)$, then $S^{\tau}$ satisfies NUPBR$(\mathbb G)$.\\
{\rm{(b)}} If $S^{(0)}\equiv  0$, then the process $S^{\tau}$ satisfies NUPBR$(\mathbb G)$.\\
{\rm{(c)}} If $\{\Delta S\not=  0\}\cap\{\widetilde Z=0<Z_{-}\}=\emptyset$, then $S^{\tau}$ satisfies NUPBR$(\mathbb G)$.\\{\rm{(d)}} If $ \widetilde Z>0 $ (equivalently $Z>0$ or $Z_{-}>0$), then $S^{\tau}$ satisfies NUPBR$(\mathbb G)$.
\end{corollary}

\begin{proof} (a) Suppose that $\left(S, S^{(0}\right)$ satisfies NUPBR$(\mathbb F)$. Then, it is obvious that $S-S^{(0)}$ satisfies NUPBR$(\mathbb F)$, and assertion (a) follows from Theorem \ref{main2}. \\
(b) Since $S$ satisfies NUPBR$(\mathbb F)$ and $S^{(0)}\equiv  0$, then $\left(S, S^{(0}\right)\equiv\left(S, 0\right)$ satisfies NUPBR$(\mathbb F)$, and assertion (b) follows from assertion (a).\\
(c) It is easy to see that $\{\Delta S\not=  0\}\cap\{\widetilde Z=0<Z_{-}\}=\emptyset$ implies that $S^{(0)}\equiv  0$ (due to (\ref{psifm})). Hence, assertions (c) and (d) follow from assertion (b), and the proof of the corollary is completed.\end{proof}

\begin{remark}
It is worth mentioning that $X-Y$ may satisfy NUPBR$(\mathbb H)$,
while $(X,Y)$ may not satisfy  NUPBR$(\mathbb H)$. For a  non
trivial example,    consider $X_t=B_t+\lambda t$ and $Y_t=N_t$
where $B$ is a standard Brownian motion and $N$ is the Poisson
process with  intensity $\lambda$.\end{remark}

\noindent We now  give an answer to the second problem {\bf
{(P2)}} for the quasi-left-continuous semimartingales. Later on (in
Theorem \ref{main5}) we will generalize this result.
\begin{proposition}\label{main1}
The following assertions are equivalent:\\
{\rm{(a)}} The thin set $\left\{ \widetilde Z=0\ \&\ Z_{-}>0\right\}$ is accessible.\\
{\rm{(b)}} For any (bounded) $S$ that is $\mathbb
F$-quasi-left-continuous and satisfies NUPBR$(\mathbb F )$, the
process $ S^{\tau}$ satisfies NUPBR$(\mathbb G)$.
\end{proposition}

\begin{proof} The implication  (a)$\Rightarrow$(b) follows from Corollary \ref{corollary1}--(c), since we have $$\{\Delta S\not=0\}\cap\{\widetilde Z=0<Z_{-}\}=\emptyset.$$

We now  focus on proving the reverse implication. To this end, we
suppose that assertion (b) holds, and we consider an $\mathbb
F$-stopping time $\sigma$ such that $\Rbrack\sigma\Lbrack\subset
\{\widetilde Z=0<Z_{-}\}$. It is known that $\sigma$ can be
decomposed into a totally inaccessible part ${\sigma}^i$ and an
accessible part ${\sigma}^{a}$ such that $\sigma={\sigma}^i\wedge
{\sigma}^a$. Consider the quasi-left-continuous $\mathbb
F$-martingale
$$M=V-\widetilde V\in {\cal M}_{0,loc}(\mathbb F)$$ where $  V:=I_{\Rbrack {\sigma}^i,+\infty\Rbrack}\ \mbox{and} \ \ \widetilde V:=(V)^{p,\mathbb F}.$
It is known from  \cite[paragraph 14, Chapter XX]{DMM},  that
\begin{equation}\label{disjoint}\{\widetilde Z=0\} \mbox{~and~}
\{ Z_{-}=0\} \mbox{~are~ disjoint~ from~}  \Lbrack 0,\tau\Lbrack
\,.
\end{equation}  This implies that $\tau<\sigma\leq \sigma^i\ \
P-a.s.$. Hence, we get
\begin{equation}\label{predictableMtau}
M^{\tau}=- \widetilde V ^{\tau}\ \ \ \mbox{is $\mathbb G$-predictable}.\end{equation}
Since $M^{\tau}$ satisfies NUPBR$(\mathbb G)$, then we conclude that this process is null (i.e. $ \widetilde V ^{\tau}=0$) due to Lemma \ref{NUPBRforPredictableProcesses}. Thus, we get
$$
0=E\left(\widetilde V_{\tau}\right)=E\left(\int_0^{+\infty} Z_{s-} d\widetilde V_s\right)=E\left(Z_{{\sigma}^i-} I_{\{ \sigma^i<+\infty\}}\right),$$ or equivalently $Z_{{\sigma}^i-} I_{\{ \sigma^i<+\infty\}}=0\ \ \ P-a.s.$ This is possible only if $\sigma^i=+\infty\ \ P-a.s.$ since on $\{\sigma^i<+\infty\}\subset \{ \sigma=\sigma^i<+\infty\}$ we have $Z_{\sigma^i-}=Z_{\sigma-}>0$. This proves that $\sigma$ is an accessible stopping time. Since $\{\widetilde Z=0<Z_{-}\}$ is an optional thin set, assertion (a) follows immediately. This ends the proof of the proposition.

\end{proof}

\subsection{Thin Processes with Predictable Jump Times}
In this subsection, we  outline the main results on the NUPBR
condition for the stopped accessible parts of $\mathbb
F$-semimartingales with a random time. This boils down to consider
thin semimartingales with predictable jump times only. We start by
addressing question {\bf{(P1)}} in the case of single jump process
with predictable jump time.

\begin{theorem}\label{main3} Consider an $\mathbb F$-predictable stopping time $T$
and an ${\cal F}_T$-measurable random variable $\xi$
 such that $E(\vert \xi \vert\big|\ {\cal F}_{T-})<+\infty\ P-a.s.$.\\ If $S:= \xi  I_{\{ Z_{T-}>0\}}I_{\Rbrack T,+\infty\Rbrack}$, then the following assertions are equivalent:\\
{\rm{(a)}} $S^{\tau}$ satisfies NUPBR$(\mathbb G)$,\\
{\rm{(b)}} The process $\widetilde S:=\xi I_{\{ \widetilde Z_T>0\}}I_{\Rbrack T,+\infty\Rbrack}=
I_{\{ \widetilde Z>0\}}\centerdot S$ satisfies NUPBR$(\mathbb F)$.\\
{\rm{(c)}} There exists a probability measure on $(\Omega, {\cal F}_T)$, denoted by $Q_T$,
such that $Q_T$ is absolutely continuous with respect to $P$, and $S$ satisfies   NUPBR$(\mathbb F, Q_T)$.\\
\end{theorem}

\noindent The proof of this theorem is long and requires intermediary results that are
interesting in themselves. Thus, this proof will be given later in Section  \ref{proofs}.\\

\begin{remark} \label{remark214}1) The importance of Theorem \ref{main3} goes beyond its vital role,
as a building block for  the more general result. In fact,  Theorem \ref{main3} provides two different characterizations for   NUPBR$\left(\mathbb G\right)$ of $S^{\tau}$. The first characterization is expressed in term of   NUPBR$\left(\mathbb F\right)$ of $S$ under absolute continuous change of measure, while the second characterization uses transformation of $S$ without any change of measure. Furthermore, Theorem \ref{main3} can be easily extended to the case of countably many ordered predictable jump times $T_0=0\leq T_1\leq T_2\leq...$ with $\sup_n T_n=+\infty\ P-a.s.$.\\
2) In Theorem \ref{main3}, the choice of $S$ having the form
$S:= \xi I_{\{ Z_{T-}>0\}}I_{\Rbrack T,+\infty\Rbrack}$ is not
restrictive. This can be understood from the fact that any single
jump process $S$ can be decomposed as follows
$$S:= \xi  I_{\Rbrack T,+\infty\Rbrack}= \xi  I_{\{ Z_{T-}>0\}}I_{\Rbrack T,+\infty\Rbrack}+ \xi  I_{\{ Z_{T-}=0\}}I_{\Rbrack T,+\infty\Rbrack}
=:{\overline S}+{\widehat S}.$$ Thanks to $\{T\leq\tau\}\subset\{
Z_{T-}>0\}$, we have ${\widehat S}^\tau=\xi  I_{\{
Z_{T-}=0\}}I_{\{T\leq\tau\}}I_{\Rbrack T,+\infty\Rbrack}\equiv 0$
is (obviously) a $\mathbb G$-martingale. Thus, the only part of
$S$ that requires careful attention is ${\overline S}:=\xi I_{\{
Z_{T-}>0\}}I_{\Rbrack T,+\infty\Rbrack}$.
\end{remark}

\noindent The following result is a complete answer to {\bf
{(P2)}} in the case of predictable single jump processes.

\begin{proposition}\label{corollaryofmain3}
Let $T$ be an $\mathbb F$-predictable stopping time. Then, the following assertions are equivalent:\\
{\rm{(a)}} On $\left\{ T<+\infty\right\}$, we have
\begin{equation}\label{equation1111}
 \left\{ \widetilde Z_T=0\right\}\subset\displaystyle\Bigl\{ Z_{T-}=0\Bigr\}.\end{equation}
{\rm{(b)}} For any $M:=\xi I_{\Rbrack T,+\infty\Rbrack}$ where $\xi\in L^{\infty}({\cal F}_T)$ such that $E(\xi|{\cal F}_{T-})=0$, $M^{\tau}$ satisfies NUPBR$(\mathbb G)$.
\end{proposition}

\begin{proof} We start by  proving ${\rm{(a)}} \Rightarrow {\rm{(b)}} $.
Suppose that (\ref{equation1111}) holds; due to the above remark
\ref{remark214}2), we can restrict our attention to the case
  $M:=\xi I_{\{ Z_{T-}>0\}}I_{\Rbrack T,+\infty\Rbrack}$ where $\xi\in L^{\infty}({\cal F}_T)$ such that $E(\xi|{\cal F}_{T-})=0$. Since assertion (a) is equivalent to $\Rbrack T\Lbrack\cap\{\widetilde Z=0\ \&\ Z_{-}>0\}=\emptyset$, we deduce that
  $$\widetilde M:=\xi I_{\{ {\widetilde Z}_{T}>0\}}I_{\{ Z_{T-}>0\}}I_{\Rbrack T,+\infty\Rbrack}=M\ \ \ \ \ \mbox{is an $\mathbb F$-martingale}.$$
  Therefore, a direct application of Theorem \ref{main3} (to $M$)  allows us to conclude that $M^{\tau}$ satisfies the NUPBR$(\mathbb G)$.
  This ends the proof of (a)$\Rightarrow $ (b). To prove the reverse implication, we suppose that assertion (b) holds and consider
$$
M:=\xi I_{\Rbrack T,+\infty\Rbrack},\ \ \ \ \mbox{where}\ \ \xi:=I_{\{ \widetilde Z_T=0\}}-P( \widetilde Z_T=0|{\cal F}_{T-}).$$
From (\ref{disjoint}), we obtain $\{T\leq \tau\}\subset\{\widetilde Z_T>0\}\subset\{Z_{T-}>0\}$ which implies that
$$
M^{\tau}=-P( \widetilde Z_T=0|{\cal F}_{T-})I_{\{ T\leq\tau\}}I_{\Rbrack T,+\infty\Rbrack}\ \ \mbox{is $\mathbb G$-predictable}.$$
Therefore, $M^{\tau}$ satisfies NUPBR$(\mathbb G)$ if and only if it is a constant process equal to $M_0=0$ (see Lemma \ref{NUPBRforPredictableProcesses}). This is equivalent to
$$
0=E\Bigl[P( \widetilde Z_T=0|{\cal F}_{T-})I_{\{ T\leq\tau\}}I_{\Rbrack T,+\infty\Rbrack}\Bigr]=E\left(Z_{T-}I_{\{ \widetilde Z_T=0\ \&\ T<+\infty\}}\right).$$ It is obvious that this equality is equivalent to (\ref{equation1111}), and assertion (a) follows. This ends the proof of the theorem.
\end{proof}

\noindent We now state the following version of Theorem
\ref{main3}, which provides, as already said, an answer to
{\bf{(P1)}} in  the case where there are countable many arbitrary
predictable jumps. The proof of this theorem will be given in
Subsection \ref{proofmain3}.

\begin{theorem}\label{main4}
Let $S$ be a thin process with predictable jump times only and satisfying NUPBR$(\mathbb F)$.
Then, the following assertions are equivalent.\\
{\rm{(a)}}  The process $S^{\tau}$ satisfies NUPBR$(\mathbb G)$.\\
{\rm{(b)}}  For any $\delta>0$, there exists a positive $\mathbb
F$-local martingale, $Y$, such that $\ ^{p,\mathbb F}\left(Y\vert \Delta S\vert\right)<+\infty$ and
\begin{equation}\label{mainassumptionbeforetau}
 ^{p,\mathbb F}\left(Y{\Delta S} I_{\{ \widetilde Z>0\ \&\ Z_{-}\geq\delta\}}\right)=0.\end{equation}
\end{theorem}


\begin{remark}\label{remarkformmain4}
1) Suppose that $S$ is a thin process with predictable jumps only, satisfying NUPBR$(\mathbb F)$,
 and that $\{\widetilde Z=0\ \&\ Z_{-}>0\}\cap\{\Delta S\not=0\}=\emptyset$ holds. Then, S$^{\tau}$ satisfies NUPBR$(\mathbb G)$. This follows immediately from Theorem \ref{main4} by using $Y\in {\cal L}(S,\mathbb F)$ and Lemma \ref{LY}.\\
2) Similarly to Proposition \ref{main1}, we can easily prove that
the thin set $\{\widetilde Z=0\ \&\ Z_{-}>0\}$ is totally
inaccessible if and only if $X^{\tau}$ satisfies NUPBR$(\mathbb
G)$ for any thin process $X$ with predictable jumps only
satisfying NUPBR$(\mathbb F)$.
\end{remark}


\subsection{The General Framework}
Throughout the paper, with any $\mathbb H$-semimartingale, $X$, we
associate a sequence of $(\mathbb H)$-predictable stopping times
$(T^X_n)_{n\geq 1}$  that exhaust the accessible jump times of
$X$. Furthermore, we can decompose $X$ as follows.

\begin{equation}\label{Sdecomposition}
X=X^{(qc)}+X^{(a)},\ X^{(a)}:=I_{\Gamma_X}\is X,\
X^{(qc)}:=X-X^{(a)},\ \Gamma_X:=\bigcup_{n=1}^{\infty} \Rbrack
T_n^X\Lbrack.\end{equation} The process $X^{(a)}$ (the accessible
part of $X$) is a thin process with  predictable jumps only,
 while $X^{(qc)}$ is a $\mathbb H$-quasi-left-continuous process (the quasi-left-continuous part
 of $X$).

\begin{lemma}\label{NUPBRdecomposed}
 Let $X$ be an $\mathbb{H}$-semimartingale. Then $X$ satisfies NUPBR$(\mathbb{H})$ if and only
 if $X^{(a)}$ and $X^{(qc)}$ satisfy NUPBR$(\mathbb{H})$.
\end{lemma}

\begin{proof}
  Thanks to Proposition \ref{charaterisationofNUPBRloc}, $X$ satisfies NUPBR$(\mathbb{H})$ if
  and only if there exist an $\mathbb{H}$-predictable real-valued process $\phi>0$ and a
  positive $\mathbb{H}$-local martingale $Y$ such that $Y(\phi\centerdot X)$ is an $\mathbb{H}$-
  local martingale. Then, it is obvious that  $Y(\phi I_{\Gamma_X}\centerdot X)$ and
  $Y(\phi I_{{\Gamma_X}^c}\centerdot X)$ are both $\mathbb{H}$-local martingales.
  This proves that $X^{(a)}$ and $X^{(qc)}$ both satisfy NUPNR$(\mathbb{H})$.\\
 \noindent Conversely, if $X^{(a)}$ and $X^{(qc)}$ satisfy NUPNR$(\mathbb{H})$,
 then there exist two $\mathbb{H}$-predictable real-valued processes $\phi_1, \phi_2>0$ and
 two positive $\mathbb{H}$-local  martingales $D_1={\cal E}(N_1), D_2={\cal E}(N_2)$  such that $D_1(\phi_1\centerdot (I_{\Gamma_X}\centerdot S))$ and
$D_2(\phi_2\centerdot (I_{{\Gamma_X}^c}\centerdot X))$ are both
$\mathbb{H}$-local martingales. Remark that there is no loss of
generality in assuming $N_1=I_{\Gamma_X}\is N_1$ and
$N_2=I_{{\Gamma_X}^c}\is N_2$. Put
 \begin{eqnarray*}
  N := I_{{\Gamma_X}}\centerdot N_1  + I_{{\Gamma_X}^c}\centerdot N_2 \ \ \ \ \mbox{and} \ \ \ \psi := \phi_1 I_{{\Gamma_X}} + \phi_2I_{{\Gamma_X}^c}.
 \end{eqnarray*}
 Obviously, ${\cal E}( N)>0$, ${\cal E}( N)$ and ${\cal E}( N)(\psi \centerdot S)$ are $\mathbb H$-local martingales, $\psi$ is $\mathbb H$-predictable and $0< \psi\leq 1$. This ends the proof of the lemma.
\end{proof}

 \noindent  Below, we state  the answer to question {\bf{(P1)}} in this general framework,
  which, using Lemma \ref{NUPBRdecomposed} will be a
  consequence of Theorems \ref{main2} and \ref{main3}.

\begin{theorem}\label{main6}
Suppose that $S$ satisfies NUPBR$(\mathbb F)$. Then, the following assertions are equivalent.\\
{\rm{(a)}}  The process $S^{\tau}$ satisfies NUPBR$(\mathbb G)$.\\
{\rm{(b)}}  For any $\delta>0$, the process $$I_{\{ Z_{-}\geq \delta\}}\is (S^{(qc)}-S^{(qc,0)}):=I_{\{ Z_{-}\geq \delta\}}\is (S^{(qc)}-I_{\Gamma^c}\is S^{(0)})$$ satisfies NUPBR$(\mathbb F)$, and there exists a positive $\mathbb
F$-local martingale, $Y$, such that $\ ^{p,\mathbb F}\left(Y\vert \Delta S\vert\right)<+\infty$ and
$$
 ^{p,\mathbb F}\left(Y{\Delta S} I_{\{ \widetilde Z>0\ \&\
 Z_{-}\geq\delta\}}\right)=0.$$
\end{theorem}

\begin{proof} Due to Lemma \ref{NUPBRdecomposed}, it is obvious that $S^{\tau}$ satisfies NUPBR$(\mathbb G)$ if and only if both
$(S^{(qc)})^{\tau}$ and $(S^{(a)})^{\tau}$ satisfy NUPBR$(\mathbb
G)$. Thus, using both Theorems \ref{main2} and \ref{main4}, we
deduce that this last fact is true if and only if for any
$\delta>0$, the process $I_{\{ Z_{-}\geq
\delta\}}\is(S^{(qc)}-I_{\Gamma^c}\is S^{(0)})$ satisfies
NUPBR$(\mathbb F)$ and there exists a positive $\mathbb F$-local
martingale $Y$ such that
$$\begin{array}{lll}
^{p,\mathbb F}\left(Y\vert \Delta S\vert\right)=\ ^{p,\mathbb F}\left(Y\vert \Delta S^{(a)}\vert\right)<+\infty\ \ \ \mbox{and}\\
\\
^{p,\mathbb F}\left(Y\Delta S I_{\{ \widetilde Z>0,\ Z_{-}\geq\delta\}}\right)=\,^{p,\mathbb F}\left(Y\Delta S^{(a)} I_{\{ \widetilde Z>0,\ Z_{-}\geq\delta\}}\right)=0.\end{array}
$$
This ends the proof of the theorem.
\end{proof}

\begin{corollary}
\label{continuous} The following assertions hold.\\
{\rm{(a)}} If either $m$ is continuous or $Z$ is positive(equivalently $\widetilde Z>0$ or $Z_{-}>0$),  $S^{\tau}$ satisfies NUPBR$(\mathbb G)$ whenever $S$ satisfies NUPBR$(\mathbb  F)$.\\
{\rm{(b)}} If $S$   satisfies NUPBR$(\mathbb  F) $ and $\{\Delta
S\not=  0\}\cap\{\widetilde Z=0<Z_{-}\}=\emptyset$, then
$S^{\tau}$ satisfies NUPBR$(\mathbb G)$.\\{\rm{(c)}} If $S$ is
continuous and satisfies NUPBR$(\mathbb  F)$, then for any random
time $\tau$, $S^{\tau}$  satisfies NUPBR$(\mathbb G)$.
 \end{corollary}

\begin{proof} 1) The proof of the assertion (a) of the corollary follows easily from
Theorem \ref{main6}. Indeed, in the two cases, one has $\{
\widetilde Z=0<Z_-\}=\emptyset$ which implies that $\{ \widetilde
Z=0,\delta \leq Z_-\} =\emptyset$ and $S^{(qc,0)}\equiv 0$ (due to (\ref{psifm})). Then, due to Lemma \ref{LY}, It suffices to take   $Y \in {\cal {L}}(S,\mathbb F)$ ---since this set is non-empty--- and apply Theorem \ref{main6}. \\
2) it is obvious that assertion (c) follows from assertion (b). To
prove this latter, it is enough to remark that $\{\Delta
S\not=0\}\cap\{ \widetilde Z=0,\delta \leq Z_-\} =\emptyset$
implies that
$$
I_{\{Z_{-}\geq \delta\}}\is S^{(qc,0)}\equiv 0\ \ \mbox{and}\ \ \ \Delta S I_{\{ \widetilde Z=0,\delta \leq Z_-\}}=\Delta S I_{\{Z_-\geq\delta\}}.$$
Thus, again, it is enough to take $Y \in {\cal {L}}(S,\mathbb F)$  and apply Theorem \ref{main6}. This ends the proof of the corollary.
\end{proof}

\begin{remark} Any  of the two assertions of the above corollary generalizes the main
result of Fontana et al.\cite{fjs}, obtained under some
restrictive assumptions on the random time $\tau$ and the market
model as well.\end{remark}

\noindent Below, we provide a general answer to question
{\bf{(P2)}} , as a consequence of Theorems \ref{main2},
\ref{main1} and \ref{main4}.

\begin{theorem}\label{main5}
The following assertions are equivalent:\\
{\rm{(a)}} The thin set $\left\{ \widetilde Z=0\ \&\ Z_{-}>0\right\}$ is evanescent.\\
{\rm{(b)}} For any (bounded) $X$ satisfying NUPBR$(\mathbb F )$, $ X^{\tau}$ satisfies NUPBR$(\mathbb G)$.
\end{theorem}

\begin{proof} Suppose that assertion (a) holds, and consider a process $X$ satisfying
NUPBR$(\mathbb F )$. Then, $X^{(qc,0)}:=I_{\Gamma_X^c}\is
X^{(0)}\equiv 0$, where $X^{(0)}$ is defined   as in
(\ref{Szero}). Hence $I_{\{ Z_{-}\geq \delta\}}\is \left(
X^{(qc)}-I_{\Gamma_X^c}\is X^{(0)}\right)$ satisfies
NUPBR$(\mathbb F)$ for any $\delta>0$, and  NUPBR$(\mathbb G)$
property of $(X^{(qc)})^{\tau}$ follows immediately from
Theorem \ref{main2} on the one hand. On the other hand, it is easy to see that $X^{(a)}$ fulfills the condition (\ref{mainassumptionbeforetau}) with $Y\equiv 1$. Thus, thanks to Theorem \ref{main4} (applied to the thin process $X^{(a)}$ satisfying NUPBR$(\mathbb F)$), we conclude that $(X^{(a)})^{\tau}$ satisfies   NUPBR$(\mathbb G)$. Thus, due to Lemma \ref{NUPBRdecomposed}, the proof of (a)$\Rightarrow$(b) is completed.\\
We now suppose that assertion (b) holds. On the one hand, from
Proposition \ref{main1}, we deduce that $\{\widetilde Z=0<Z_{-}\}$ is
accessible and can be covered with the graphs of $\mathbb
F$-predictable stopping times $(T_n)_{n\geq 1}$. On the other
hand, a direct application of Proposition \ref{corollaryofmain3}
to all single predictable jump $\mathbb F$-martingales, we obtain
$\{\widetilde Z=0<Z_{-}\}\cap \Rbrack T\Lbrack=\emptyset$ for any
$\mathbb F$-predictable stopping time $T$. Therefore, we get
$$
\{\widetilde Z=0<Z_{-}\}=\bigcup_{n=1}^{\infty} \left(\{\widetilde Z=0<Z_{-}\}\cap \Rbrack T_n\Lbrack\right)=\emptyset.$$
This proves assertion (a), and the proof of the theorem is completed.
\end{proof}


\section{Stochastics from--and--for Informational Non-Arbitrage}\label{newstochastic}

In this section, we  develop new stochastic results that will play
a key role in the proofs and/or the statements of the main results
outlined in the previous section.  The first subsection compares
the $\mathbb G$-compensators and the $\mathbb F$-compensators,
while the second subsection studies a   $\mathbb G$-martingale
that  is vital in the explicit construction of deflators. We
recall that    $
 Z_{-}+\Delta m=\widetilde Z $ (see \cite{Jeuyor}).

\begin{lemma}\label{proposition1}
Let $Z$ and $\widetilde Z$  be the two supermartingales given by (\ref{ZandZtilde}). \\
{\rm{(a)}} The  three sets $\{\widetilde Z=0\}$, $\{ Z=0\}$ and $\{ Z_{-}=0\}$ have the same d\'ebut which is an $\mathbb F$-stopping time that we denote by
\begin{equation}
\label{stoppingtimesfotau}
\widehat R:=\inf\{ t\geq 0\big|\  Z_{t-}=0\}.
\end{equation}

\noindent{\rm{(b)}}   The following $\mathbb F$-stopping times $$\label{RhatzeroRtildezero}
\widehat R_0:=\left\{\begin{array}{lll} \widehat R\ \ \mbox{on}\ \{ Z_{\widehat R-}=0\}\\
\\
+\infty\ \ \ \mbox{otherwise}\end{array}\right.\ \ \ \ \mbox{and}\ \ \ \ \ \  {\widetilde R}_0:=\left\{\begin{array}{lll} {\widehat R}\ \ \ \mbox{on}\ \{ \widetilde Z_{\widehat R}=0\}\\
 \\
 +\infty\ \ \mbox{otherwise}\end{array}\right.$$
are such that $\widehat R_0$ is a $\mathbb F$-predictable stopping
time, and
\begin{equation}\label{tauandStoppingtimes}
\tau\leq \widehat R,\ \ \ \ \ \tau<\widetilde R_0,\ \ \ \ \ \  \ \ P -a.s.\end{equation}
{\rm{(c)}}  The $\mathbb G$-predictable process
\begin{equation}\label{ProcessH}
H_t  := \left(Z_{t-}\right)^{-1}I_{\Rbrack 0,\tau\Lbrack}(t), \end{equation}
 is $\mathbb G$-locally bounded.\\

\end{lemma}

\begin{proof} From \cite[paragraph 14, Chapter XX]{DMM}, for any random time $\tau$, the sets $\{\widetilde Z=0\}$ and $\{Z_-=0\}$ are disjoint from $\Lbrack 0,\tau\Lbrack$ and have the same lower bound $\widehat R$, the smallest $\mathbb F$-stopping time greater than $\tau$. Thus, we also conclude that $\{Z=0\}$ is disjoint from $\Lbrack 0,\tau\Rbrack$.  This leads to assertion (a). The process $X:=Z^{-1} I_{\Lbrack 0,\tau\Rbrack}$ being  a c\`adl\`ag  $\mathbb G$-supermartingale \cite{Yor},
its left limit is locally bounded. Then, due to $$(Z_{-})^{-1}I_{\Lbrack 0,\tau\Lbrack}=X_{-},$$
the local boundedness of $H$ follows.
This ends the proof of the lemma. \end{proof}

\subsection{Exact Relationship between Dual Predictable Projections under $\mathbb G$ and $\mathbb F$ }
The main results of this subsection are summarized in Lemmas
\ref{lemmecrucial} and \ref{lem:beforetauforoptional1}, where we
address the question of how to compute     $\mathbb G$-dual
predictable projections in term of   $\mathbb F$-dual predictable
projections and vice versa. These results are based essentially on
the following standard  result on progressive enlargement of
filtration (we refer the reader to \cite{DMM,Jeu} for
proofs).

\begin{proposition}\label{semiGF} Let $M$  be an ${\mathbb  F}$-local martingale. Then,
for any random time $\tau$, the process
\begin{eqnarray}
\label{before} \widehat{M}_t  &:=& M_t^\tau - \int_{0}^{t\wedge
\tau} { \frac{d \langle M, \mm\rangle^{\mathbb F}_s }{Z_{s-}} }
\end{eqnarray} is a $\mathbb G$-local martingale, where $m$ is defined in (\ref{processmZ}).
\end{proposition}


\noindent In the following lemma, we   express the $\mathbb
G$--dual predictable projection of an $\mathbb F$-locally
integrable variation process in terms of an $\mathbb F$--dual
predictable projection, and $\mathbb G$-predictable projection in
terms of $\mathbb F$-predictable projection.


\begin{lemma}\label{lemmecrucial}
The following assertions hold.\\
{\rm{(a)}}  For any $\mathbb F$-adapted process $V$  with locally integrable
variation, we have
\begin{equation}\label{GcompensatorofVbeforetaugeneral}
\comg{V^{\tau}} =(Z_{-})^{-1}I_{\Lbrack 0,\tau\Lbrack}\is\bigl(\widetilde Z \is V\bigr)^{p,\mathbb F}  . \end{equation}
{\rm{(b)}}  For any $\mathbb F$-local martingale  $M$, we have, on $\Rbrack 0,\tau\Lbrack$
\begin{equation}\label{ZZtildeinaccessiblejumps}
\prog{\frac{\Delta M}{\widetilde Z}}={{ \prof{{\Delta M}I_{\{\widetilde Z>0\}}}}\over{Z_{-}}},\ \ \mbox{and}\ \ \prog{\frac{1}{\widetilde Z}}={{\prof{I_{\{\widetilde Z>0\}}}}\over{Z_{-}}}.
\end{equation}

\noindent {\rm{(c)}}  For any quasi-left-continuous $\mathbb F$-local martingale  $M$, we have,  on $\Rbrack 0,\tau\Lbrack$
\begin{equation}\label{Usefulprojectionbefortauequa}
\prog{\frac{\Delta M}{\widetilde Z}}= 0,\ \ \mbox{and}\ \
\prog{\frac{1}{Z_{-}+\Delta
m^{(qc)}}}={1\over{Z_{-}}},\end{equation} where  $m^{(qc)}$ is the
quasi-left-continuous $\mathbb F$-martingale defined in
(\ref{Sdecomposition}).
\end{lemma}

\begin{proof}
(a)  Using the notation (\ref{ProcessH}), the equality
(\ref{before}) takes the form
$$M^{\tau}=\mg  + HI_{\Lbrack 0,\tau\Lbrack}\is{\crof{M,m}}\,.$$
By taking $M=V- V ^{p,\mathbb F}   $, we obtain
$$V^{\tau}=I_{\Lbrack 0,\tau\Lbrack}\is V^{p,\mathbb F}+\mg
 +HI_{\Lbrack 0,\tau\Lbrack}\is\crof{V,m}=\mg +I_{\Lbrack 0,\tau\Lbrack}\is V^{p,\mathbb F}+{1\over{Z_{-}}}I_{\Lbrack 0,\tau\Lbrack}\is\comf{\Delta m\is V},$$
which proves assertion (a).\\
(b) Let $M$ be an $\mathbb F$-local martingale, then, for any positive
integers $(n,k)$ the process $V^{(n,k)}:=\sum {{\Delta
M}\over{\widetilde Z}} I_{\{\vert \Delta M\vert\geq k^{-1},\ \widetilde
Z\geq n^{-1}\}}$ has a locally integrable variation. Then, by using the known equality $^{p,\mathbb G}(\Delta V)=\Delta(V^{p,\mathbb G})$ (see Theorem 76 in pages 149--150 of \cite{dm2} or Theorem 5.27 in page 150 of \cite{Yanbook}), and applying assertion (a) to the process $V^{(n,k)}$, we get,   on
$\Lbrack 0,\tau\Lbrack$
\begin{eqnarray*}
\prog{{{\Delta M}\over{\widetilde Z}}I_{\{\vert\Delta M\vert\geq
k^{-1},\ \widetilde Z\geq n^{-1}\}}}&=&{1\over{Z_{-}}}\prof{\Delta
MI_{\{\vert\Delta M\vert\geq k^{-1},\ \widetilde Z\geq n^{-1}\}}}.
\end{eqnarray*}
Since $M$ is a local martingale,   by stopping we can exchange
limits with projections in both sides. Then by letting $n$ and
$k$ go to infinity, and using the fact that $\widetilde Z>0$ on
$\Lbrack 0,\tau\Lbrack$, we deduce that
\begin{eqnarray*}
\prog{{{\Delta M}\over{\widetilde Z}}}&=&{1\over{Z_{-}}}\prof{\Delta MI_{\{\widetilde Z>0\}}}.
\end{eqnarray*}
This proves the first equality in (\ref{ZZtildeinaccessiblejumps}), while the second equality follows from $\widetilde Z=\Delta m+Z_{-}$:
$$\begin{array}{lll}
Z_{-}\prog{{\widetilde Z}^{-1}}=\prog{(\widetilde Z-\Delta m)/\widetilde Z}=1-\prog{{\Delta m/\widetilde Z}}\\
\\
\hskip 2cm =1-(Z_{-})^{-1}\prof{\Delta m I_{\{\widetilde Z>0\}}}=1-\prof{ I_{\{\widetilde Z=0\}}}=\prof{ I_{\{\widetilde Z>0\}}}.\end{array}$$
In the above string of equalities, the third equality follows from the first equality in (\ref{ZZtildeinaccessiblejumps}), while the fourth equality is due to  $^{p,\mathbb F}(\Delta m)=0$ and  $\Delta m I_{\{ \widetilde Z=0\}}=-Z_{-}I_{\{ \widetilde Z=0\}}$. This ends the proof of assertion (b).\\
 \\
\noindent(c) If $M$ is a quasi-left-continuous $\mathbb F$-local martingale, then $\prof{\Delta MI_{\{\widetilde Z>0\}}}=0$, and the first property of the assertion (c) follows. Applying the first property   to $M=m^{(qc)}$  and using that,   on $\Lbrack 0,\tau\Lbrack$, one has
$\Delta m^{(qc)}\left({Z_{-}+\Delta m}\right)^{-1}=\Delta m^{(qc)}\left(Z_{-}+\Delta m^{(qc)}\right)^{-1}$, we obtain
$$
{1\over{Z_{-}}}\prog{\frac{Z_{-}}{Z_{-}+\Delta m^{(qc)}}}={1\over{Z_{-}}}\left(1-\prog{\frac{\Delta m^{(qc)}}{Z_{-}+\Delta m^{(qc)}}}\right)={1\over{Z_{-}}}.$$
This proves assertion (c), and the proof of the lemma is achieved.\end{proof}


\noindent The next lemma proves that $\widetilde Z^{-1}I_{\Lbrack0,\tau\Lbrack}$ is Lebesgue-Stieljes-integrable with respect to any process that is $\mathbb F$-adapted with $\mathbb F$-locally integrable variation. Using this fact, the lemma addresses the question of how an $\mathbb F$-compensator stopped at $\tau$ can be written in terms of a $\mathbb G$-compensator, and constitutes a {\it sort of} converse result to Lemma \ref{lemmecrucial}--(a).

 \begin{lemma}\label{lem:beforetauforoptional1}
Let $V$ be an $\mathbb F$-adapted c\`adl\`ag  process. Then the following properties hold.\\
{\rm{(a)}} If $V$ belongs to ${\cal A}^+_{loc}(\mathbb F)$ (respectively $V\in{\cal A}^+(\mathbb F)$), then the process
\begin{equation}\label{ZcdotVbeforetau}
U:= \widetilde Z ^{-1}I_{\Lbrack 0, \tau\Lbrack}\is V,
\end{equation}
belongs to ${\cal A}^+_{loc}(\mathbb G)$ (respectively to ${\cal A}^+(\mathbb G)$).\\
{\rm{(b)}} If $V$ has   $\mathbb F$-locally integrable variation, then the process $U$ is well defined, its variation is $\mathbb G$-locally integrable, and its $\mathbb G$-dual predictable projection is given by
\begin{equation}\label{Gcompensatorforqlc}
 U ^{p,\mathbb G}=\left({1\over{\widetilde Z}}I_{\Lbrack 0, \tau\Lbrack}\is V\right)^{p,\mathbb G}={1\over{Z_{-}}}I_{\Lbrack 0, \tau\Lbrack}\is \left(I_{\{ \widetilde Z>0\}}\is V\right)^{p,\mathbb F}.
\end{equation}
In particular, if $\mbox{supp} V \subset \{ \widetilde Z>0\}$,
then, on ${\Lbrack 0, \tau\Lbrack}$, one has $ V^{p,\mathbb F} =
Z_-\centerdot U ^{p,\mathbb G}$.
\end{lemma}

\begin{proof} (a) Suppose that $V\in{\cal A}^+_{loc}(\mathbb F)$. First, remark that, due to the fact that  $\widetilde Z$ is positive on $\Lbrack 0,
\tau\Lbrack$,
 $U$ is well defined.  Let $(\vartheta_n)_{n\geq 1}$ be a sequence of $\mathbb
F$-stopping times that increases to $+\infty$ such that $
E\left(V_{\vartheta_n}\right)<+\infty.$ Then, if
$E\left(U_{\vartheta_n}\right)\leq E\left(V_{\vartheta_n}\right)
$,   assertion (a) follows.  Thus, we calculate
 \begin{eqnarray*}
E\left(U_{\vartheta_n}\right )&=&\displaystyle E\left(\int_0^{\vartheta_n} I_{\{0<t\leq\tau\}}{1\over{\widetilde{Z}_t}}d V_t\right)=E\left(\int_0^{\vartheta_n} {{P(\tau\geq t|{\cal F}_t)}\over{\widetilde{Z}_t}}I_{\{ \widetilde{Z}_t>0\}}d V_t\right)\\
&\leq& E\left(V_{\vartheta_n}\right).\end{eqnarray*}
The last inequality is obtained due to $\widetilde{Z}_t:= P(\tau\geq t|{\cal F}_t)$. This ends the proof of assertion (a) of the lemma.\\
(b) Suppose that $V\in {\cal A}_{loc}(\mathbb F)$, and denote by
$W:= V^++V^-$ its variation. Then $W\in {\cal A}^+_{loc}(\mathbb
F)$, and a direct application of the first assertion implies that
$$ \left({\widetilde Z}\right)^{-1}I_{\Lbrack 0, \tau\Lbrack}\is W\in {\cal A}^+_{loc}(\mathbb G).$$ As a result, we deduce that $U$ given by (\ref{ZcdotVbeforetau}) for the case of $V=V^+-V^-$ is well defined and has variation equal to  $\left({\widetilde Z}\right)^{-1}I_{\Lbrack 0, \tau\Lbrack}\is W $ which is $\mathbb G$-locally integrable.
By setting $U_n:= I_{\Lbrack 0, \tau\Lbrack}\is  \left({\widetilde Z}^{-1}I_{\{ \widetilde Z\geq 1/n\}}\is
 V\right)$, we derive,
due to (\ref{GcompensatorofVbeforetaugeneral}),
$$
\left(U_n\right)^{p,\mathbb G}={1\over{ Z_{-}}}I_{\Lbrack 0,
\tau\Lbrack}\is\left(I_{\{ \widetilde Z\geq 1/n\}}\is
V\right)^{p,\mathbb F}.$$    Hence, since $
 U  ^{p,\mathbb G}=\lim_{n\longrightarrow+\infty}\left(U_n\right)^{p,\mathbb G}, $
by taking the limit in the above equality, (\ref{Gcompensatorforqlc}) follows
immediately, and the lemma is proved. 
\end{proof}

\subsection{An Important $\mathbb G$- local martingale}
\noindent In this subsection, we   introduce a $\mathbb
G$- local martingale that will be  crucial for the construction of the deflator.
\begin{lemma}\label{lem:vbfinite}
The following nondecreasing process
\begin{eqnarray}\label{V(b)process}
V^{\mathbb G}_t := \sum_{0\leq u\leq t} \ ^{p,\mathbb
F}\left(I_{\{\widetilde{Z}=0\}}\right)_u I_{\{u\leq \tau\}}
\end{eqnarray}
is $\mathbb G$-predictable, c\`adl\`ag, and locally bounded.
\end{lemma}

\begin{proof} The $\mathbb G$-predictability of $V^{\mathbb G}$ being  obvious, it remains to   prove that this process is $\mathbb G$-locally bounded.
   Since $Z_{-}^{-1}I_{\Lbrack 0, \tau \Lbrack}$  is $\mathbb{G}$-locally bounded, then there exists a sequence of
 $ \mathbb{G}$-stopping times $(\tau^{\mathbb{G}}_n)_{n\geq 1}$ increasing to infinity such that
\begin{eqnarray*}
\left(\frac{1}{Z_{-}}I_{\Lbrack 0, \tau \Lbrack}\right)^{\tau^{\mathbb{G}}_n} \leq n+1.
\end{eqnarray*}
Consider a sequence of $\mathbb{F}$-stopping times
$(\sigma_n)_{n\geq 1}$ that increases to infinity such that
$\left\langle m,m\right\rangle_{\sigma_n} \leq n+1.$ Then, for any nonnegative $\mathbb F$-predictable process $H$ which is bounded by $C>0$, we
calculate that
  \begin{eqnarray*}
    (H\is V^{\mathbb G})_{\sigma_n\wedge \tau^{\mathbb{G}}_n} &=& \sum_{0\leq u\leq \sigma_n\wedge \tau_n^{\mathbb{G}}} H_u\ ^{p,\mathbb F}\left(I_{\{\widetilde{Z}=0\}}\right)_u I_{\{u\leq \tau\}} I_{\{Z_{u-} \geq \frac{1}{n+1}\}}\\
    &\leq& \sum_{0\leq u\leq \sigma_n} H_u\ ^{p,\mathbb F}\left(I_{\{\Delta m \leq - \frac{1}{n+1}\}} \right)_u\\
     &\leq& (n+1)^2H\is\left\langle m,m\right\rangle_{\sigma_n} \leq C (n+1)^3.
  \end{eqnarray*}
   This ends the proof of the proposition.
\end{proof}

\noindent  The important
$\mathbb G$-local martingale  will  result  from an optional
integral. For the notion of compensated stochastic integral (or
optional stochastic integral), we refer the reader to \cite{Jacod}
(Chapter III.4.b p. 106-109) and \cite{dm2} (Chapter VIII.2
sections 32-35 p. 356-361 ). Below, for the sake of completeness,
we give the definition of this integration.

\begin{definition}\label{SIOptional} (see \cite{Jacod}, Definition (3.80))
Let $N$ be an $\mathbb H$-local martingale with continuous martingale part $N^c$, and let $H$  be an $\mathbb H$-optional process.\\
i) The process $H$ is said to be integrable with respect to $N$ if $^{p, \mathbb H}H$ is $N^c$ integrable and the process
$$\left ( \sum_{s\leq t}\left (H_s\Delta N_s -\, ^{p, \mathbb H}(H\Delta N)_s \right )^2 \right ) ^{{1/2}}$$
is locally integrable.
The set of integrable processes with respect to $N$ is denoted by $^{o}L^1_{loc}(N, \mathbb H)$.\\
ii) For $H\in \,^{o}L^1_{loc}(N, \mathbb H)$, the compensated stochastic integral of $H$ with respect to $N$, denoted by  $H\odot N$, is the unique local martingale $M$ which satisfies
\begin{equation}\label{continuousPartandJumpsofOI}
M^c=\,^{p, \mathbb H}H \centerdot N^c \quad \textrm{   and   }\quad \Delta M=H\Delta N -\, ^{p, \mathbb H}(H\Delta N).\end{equation}
\end{definition}

\noindent Among the most useful results of the literature involving this integral is the following

 \begin{proposition}\label{DM2} (see \cite{dm2})
{\rm{(a)}} The compensated stochastic integral $M=H\odot N$ is  the unique $\mathbb H$-local martingale such that, for any $\mathbb H$-local martingale $Y$,
$$\E\left ([M,Y]_\infty\right )= \E \left (\int_0^\infty H_s d[N, Y]_s\right ).$$
{\rm{(b)}} The process $[M,Y]- H\centerdot[N,Y]$ is an $\mathbb H$-local
martingale. As a result $[M,Y]\in {\cal A}_{loc}(\mathbb H)$ if and only if $H\centerdot[N,Y]\in {\cal A}_{loc}(\mathbb H)$ and in this case we have
$$
\langle M, Y\rangle^{\mathbb H}=\left(H\centerdot[N,Y]\right)^{p,\mathbb H}.$$
\end{proposition}

\noindent Now, we are in the stage of defining the $\mathbb G$-local martingale which will play the role of deflator for a class of processes.

\begin{proposition}\label{lemmforNtilde} Consider the following $\mathbb G$-local martingale
  \begin{eqnarray}\label{hatm}
  \widehat m :=    I_{\Lbrack 0,\tau\Lbrack}\is m - \frac{1}{Z_{-} }I_{\Lbrack 0,\tau\Lbrack}\is\langle m \rangle^{\mathbb F}\,,\end{eqnarray}
and the process
 \begin{equation}\label{integrabilitypofK}
   K:={{Z_{-}^2}\over{Z_{-}^2+\Delta \langle m\rangle^{\mathbb F}}}\,{1\over{\widetilde Z}} \,I_{\Lbrack 0,\tau\Lbrack}.\end{equation}
 Then,  $K$ belongs to    the space $^{o}L^1_{loc}(\widehat m, \mathbb G)$  defined in \ref{SIOptional}. Furthermore, the  $\mathbb G$-local martingale
  \begin{equation}\label{Ntilde}
  {\Lopt} := -K \odot\widehat{m},\end{equation}
 satisfies the following\\
{\rm{(a)}} ${\cal E}\left(\Lopt\right)>0$ (or equivalently $1+\Delta \Lopt>0$).\\
{\rm{(b)}}  For any $M\in {\cal M}_{0,loc}(\mathbb F)$,  setting $\widehat M  := M^{\tau} - Z_{-}^{-1} I_{\Rbrack 0,\tau\Lbrack}\is\langle M, m\rangle^{\mathbb F} $, we have
\begin{equation}\label{localintegrablebracket}
 [\Lopt, \widehat M ]\in{\cal A}_{loc}(\mathbb G)\ \ \left(\mbox{i.e.}\ \langle \Lopt, \widehat M \rangle^{\mathbb G}\ \mbox{ exists}\right)\,.\end{equation}
\end{proposition}

\begin{proof} We shall prove  that $K\in {^{o}L^1_{loc}(\widehat m, \mathbb G)}$ in the appendix \ref{proofk}.  For the sake of simplicity in notations, throughout this proof, we will use $\kappa:= Z_{-}^2+ \Delta
\langle m \rangle^{\mathbb F}$. \\
We now prove assertions (a) and (b). Due to  (\ref{differencede Sauts}), we have, on $\Lbrack 0,\tau\Lbrack$,
 $$
 -\Delta\Lopt=K\Delta{\widehat m }-\,{^{p,\mathbb G\!}\left(K\Delta{\widehat m}\right)}=1-Z_{-}\Bigl(\widetilde Z\Bigr)^{-1}-\,^{p,\mathbb F\!}\Bigl(I_{\{ \widetilde Z=0\}}\Bigr).$$ Thus, we deduce that $1+{\Delta\Lopt}>0$, and assertion (a) is proved. In the rest of this proof, we will prove (\ref{localintegrablebracket}). To this end, let $M\in {\cal M}_{0,loc}(\mathbb F)$.
Thanks to Proposition \ref{DM2}, (\ref{localintegrablebracket}) is equivalent to
$$
 K \is [{\widehat m}, \widehat{M}] \in {\cal A}_{loc}(\mathbb G)\ (\mbox{or equivalently}\ V_4 :=\frac{1}{\widetilde{Z}}I_{\Lbrack 0, \tau\Lbrack} \is [{\widehat m}, \widehat{M}]\in  {\cal A}_{loc}(\mathbb G)),
$$
for any $M\in {\cal M}_{0,loc}(\mathbb F)$. Then, it is easy to check that
\begin{eqnarray*}\label{eq:crulembeforetau1}
V_4 &=& \frac{Z_{-}}{\widetilde{Z}}I_{\Lbrack 0, \tau\Lbrack} \is [{\widehat m}, \widehat{M}]=  \frac{1}{\widetilde{Z}}I_{\Lbrack 0, \tau\Lbrack} \is [m, \widehat{M}] - \frac{1}{Z_{-}\,{\widetilde{Z}}}I_{\Lbrack 0, \tau\Lbrack} \is [\langle m \rangle^{\mathbb F}, \widehat{M}]\nonumber\\
\\
&=&  \frac{1}{\widetilde{Z}}I_{\Lbrack 0, \tau\Lbrack} \is [m, M] - \frac{1}{Z_{-}\,\widetilde{Z}}I_{\Lbrack 0, \tau\Lbrack} \is [m, \langle M, m\rangle^{\mathbb F}] \nonumber\\
\\
&&- \frac{1}{Z_{-}\widetilde{Z}}I_{\Lbrack 0, \tau\Lbrack} \is [\langle m \rangle^{\mathbb F}, M] + \frac{1}{Z_{-}^2\,\widetilde{Z}}I_{\Lbrack 0, \tau\Lbrack} \is [\langle m \rangle^{\mathbb F}, \langle M, m\rangle^{\mathbb F}].
\end{eqnarray*}
 Since $m$ is an $\mathbb F$-locally bounded local martingale, all the processes $$ [{m}, {M}],\ [m, \langle M, m\rangle^{\mathbb F}],\  [\langle m \rangle^{\mathbb F}, M],\ \mbox{and}\ [\langle m \rangle^{\mathbb F}, \langle M, m\rangle^{\mathbb F}]$$ belong to ${\cal A}_{loc}(\mathbb F)$. Thus, by combining this fact with Lemma \ref{lem:beforetauforoptional1} and the $\mathbb G$-local boundedness of $Z_{-}^{-p}I_{\Lbrack 0, \tau\Lbrack}$ for any $p>0$, it follows that $V_4\in{\cal A}_{loc}(\mathbb  G)$. This ends the proof of the proposition.
  \end{proof}



\section{Explicit Deflators}\label{SectionQLCprocesses}

This  section describes some classes of $\mathbb
F$-quasi-left-continuous local martingales for which the NUPBR is
preserved after stopping with $\tau$. For these stopped processes,
we describe explicitly their local martingale densities in Theorems
\ref{individualSbeforetau}--\ref{theoremMgDensity2} with
an increasing degree of generality. We recall that $m^{(qc)}$ was
defined in (\ref{Sdecomposition}) and  $\Lopt$  was defined in
Proposition \ref{lemmforNtilde}.

\begin{theorem}\label{individualSbeforetau}
 Suppose that $S$ is a quasi-left-continuous $\mathbb F$-local martingale. If $S$ and $\tau$ satisfy
 \begin{equation}\label{condition111}
 \left\{\Delta S\not=0\right\}\cap\{Z_{-}>0\}\cap\{\widetilde Z=0\}=\emptyset,\end{equation}
 then the following equivalent assertions hold
\\
{\rm{(a)}} ${\cal E}\left({\Lopt}\right)S^{\tau}$ is a $\mathbb G$-local martingale.\\
{\rm{(b)}}  ${\cal E}\left(I_{\{ \widetilde
Z=0<Z_{-}\}}\odot{m}^{(qc)}\right)S$ is an $\mathbb F$-local
martingale.
\end{theorem}
\begin{proof}
We start by giving some useful observations. Since $S$ is $\mathbb F$-quasi-left-continuous, on the one hand we deduce that ($\Gamma_m$ is  defined in (\ref{Sdecomposition}))
\begin{equation}\label{qlcequalities}
\langle S,m\rangle^{\mathbb F}=\langle S,m^{(qc)}\rangle^{\mathbb
F}=\langle S, I_{\Gamma^c_m}\is m\rangle^{\mathbb F}.
\end{equation}
On the other hand, we note that assertion  (a) is equivalent to ${\cal E}(\Lopt^{(qc)})S^{\tau}$ is a $\mathbb G$-local martingale, where $\Lopt  ^{(qc)}$ is the quasi-left-continuous local martingale part of $\Lopt$ given by
$
  {\Lopt}^{(qc)}:= I_{\Gamma^c_m}\is  {\Lopt}=-K\odot {\widehat m}^{(qc)}.$ Here $K$   is  given in Proposition \ref{lemmforNtilde}  and
  $$
 {\widehat m}^{(qc)}:=  I_{\Lbrack 0, \tau\Lbrack}\is m^{(qc)}-(Z_{-}) I_{\Lbrack 0, \tau\Lbrack}\is\langle m^{(qc)}\rangle^{\mathbb F}.
 $$
It is  easy to check that (\ref{condition111}) is equivalent  to
\begin{equation}\label{condition111bis}
I_{\{ Z_{-}>0\ \&\ \widetilde Z=0\}}\is [S,m]=0.
 \end{equation}

\noindent We now compute $ -\langle{\Lopt}^{(qc)}, {\widehat S}
 \rangle^{\mathbb G}$, where $\widehat S $ is the $\mathbb
G$-local martingale given by
$${\widehat S}:=S^\tau- (Z_{-})^{-1} I_{\Lbrack 0, \tau\Lbrack}\is \crof{S, m}.$$ Due to the quasi-left continuity of $S$ and that of $m^{(qc)}$, the two processes $\crof{S,m}$ and $\langle m^{(qc)}\rangle^{\mathbb F}$ are continuous and $\dro{S, m^{(qc)}}=\dro{S, m}$. Hence, we obtain
 \begin{eqnarray*}
K\is\dro{\sh,{\widehat{m}}^{(qc)}}&= &K\is \dro{S, {\widehat{m}^{(qc)}}}-K{\Delta{\widehat{m}^{(qc)}}}(Z_{-})^{-1}\is \crof{S,m}\\
&=&\displaystyle (\widetilde Z)^{-1}I_{\Lbrack 0, \tau\Lbrack}\is\dro{S, m^{(qc)}}=\displaystyle (\widetilde Z)^{-1}I_{\Lbrack 0, \tau\Lbrack}\is\dro{S, m}.
\end{eqnarray*}
It follows that
\begin{eqnarray}
&&-\langle\Lopt^{(qc)}, {\widehat S}\rangle^{\mathbb
G}=\comg{K\is\dro{{\widehat S},
{{\widehat m}}^{(qc)}}}=\comg{(\widetilde Z)^{-1}I_{\Lbrack 0, \tau\Lbrack}\is\dro{S, m}} \nonumber\\
& =&(Z_{-})^{-1}I_{\Lbrack 0, \tau\Lbrack}\is \left(I_{\{\widetilde Z>0\}}\is [ S,m]\right)^{p,\mathbb F}\nonumber\\
&=&\displaystyle (Z_{-})^{-1}I_{\Lbrack 0, \tau\Lbrack}\is \langle S,m\rangle^{\mathbb F}-
 (Z_{-})^{-1}I_{\Lbrack 0, \tau\Lbrack}\is \Bigl(I_{\{\widetilde Z=0<Z_{-}\}}\is[ S,m]\Bigr)^{p,\mathbb F}\nonumber\\
& =&\displaystyle (Z_{-})^{-1}I_{\Lbrack 0, \tau\Lbrack}\is
\langle S,m\rangle^{\mathbb F} +(Z_{-})^{-1}I_{\Lbrack 0,
\tau\Lbrack}\is\langle S,-I_{\{\widetilde Z=0<Z_{-}\}}\odot
m^{(qc)}\rangle^{\mathbb F}. \label{equa2december3}
\end{eqnarray}
The first and the last equality follow from  Proposition \ref{DM2} applied to $\Lopt^{(qc)} $ and $-I_{\{\widetilde Z=0<Z_{-}\}}\odot m^{(qc)}$ respectively. The second and the third equalities are due to (\ref{qlcequalities}) and (\ref{Gcompensatorforqlc}) respectively.\\

\noindent Now, we   prove the theorem. Thanks to
(\ref{equa2december3}), it is obvious that assertion (a) is
equivalent to $\langle S,-I_{\{\widetilde Z=0<Z_{-}\}}\odot
m^{(qc)}\rangle^{\mathbb F}\equiv 0$ which  in turn is equivalent
to assertion (b). This ends the proof of the equivalence between
(a) and (b).

\noindent It is also clear that the condition (\ref{condition111}) or equivalently (\ref{condition111bis}) implies that assertion (b), due to  $I_{\{ \widetilde Z=0<Z_{-}\}}\odot m^{(qc)}\equiv 0$. \\

\end{proof}
 \begin{remark}

  Suppose that    $S$ is a quasi-left-continuous $\mathbb F$-local
 martingale and let  $\widetilde R_0$ be defined in Lemma \ref{proposition1}--(b).
  Then, ${\cal E}\left({\Lopt } \right){\overline S}^{\tau}$ is a
 $\mathbb G$-local martingale, where
 \begin{equation}\label{Sbar}
 \overline{S}:=S^{\widetilde R_0-}+
 \left(\Delta S_{\widetilde R_0}I_{\Rbrack\widetilde R_0,+\infty\Rbrack}\right)^{p,\mathbb F}.\end{equation}
Indeed, writing $$\overline{S}:=S^{\widetilde R_0 }- \Delta
S_{\widetilde R_0} I_{\Rbrack\widetilde R_0,+\infty\Rbrack}+
 \left(\Delta S_{\widetilde R_0}I_{\Rbrack\widetilde R_0,+\infty\Rbrack}\right)^{p,\mathbb F}
$$ it is easy to see that the condition (\ref{condition111}) is
satisfied for $\overline S$.
 \end{remark}

 \begin{corollary}  If $S$ is quasi-left continuous and  satisfies NUPBR$(\mathbb F)$
   and $
 \left\{\Delta S\not=0\right\}\cap\{Z_{-}>0\}\cap\{\widetilde
 Z=0\}=\emptyset$, then $S$   satisfies NUPBR$(\mathbb G)$
 \end{corollary} \begin{proof} This follows from Proposition \ref{NUPBRLocalization},
  Theorem \ref {individualSbeforetau} and the fact that, if $Q$
 is equivalent to $P$, then we have  $$\{Z_{-}>0\}\cap\{\widetilde
 Z=0\}=  \{Z^Q_{-}>0\}\cap\{\widetilde
 Z^Q=0\}.$$
 Here $Z^Q_t= Q(\tau >t  \vert  {\cal F}_t) $ and
 $\widetilde Z^Q_t= Q(\tau  \geq t  \vert  {\cal F}_t)$. This last claim is a direct application of the optional and predictable selection measurable theorems, see Theorems 84 and  85 (or apply Theorem 86 directly) in \cite{dm2}.\end{proof}

\noindent In order to generalize the previous result, we need to
introduce more notations and recall others notations and some
results that are delegated in the Appendix. For the random measure
$\mu$ defined in (\ref{Wstarmu}), we associate its predictable
compensator random measure $\nu$ (see (\ref{modelSbis}) for
details). A direct application of  Theorem \ref{representation}
(in Appendix), to the martingale $m$, leads to the existence of a
local martingale $m^\bot$ as well as a $\widetilde{\cal P}(\mathbb
F)$-measurable functionals $f_m$, a process $\beta_m\in
L(S^c,\mathbb F)$ and an ${\widetilde{\cal O}}(\mathbb
F)$-measurable functional $g_m$ such that $f_m\in{\cal
G}^1_{loc}(\mu,\mathbb F )$, $g_m\in{\cal H}^1_{loc}(\mu,\mathbb F
)$ and $\beta_m\in L( {S^c})$ such that
\begin{eqnarray}\label{decompositionofm}
 m=\beta_m\is S^c+f_m\star(\mu -\nu)+ g_m\star \mu +  m^ \bot.
\end{eqnarray}
Due to the quasi-left-continuity of $S$,   ${\cal G}^1_{loc}(\mu,
\mathbb F )$ (respectively ${\cal H}^1_{loc}(\mu,\mathbb F)$) is
the set of all $\widetilde {\cal P}(\mathbb F)$-measurable
functions (respectively all  $\widetilde {\cal O}(\mathbb
F)$-measurable functions)  $W$ such that $$\sqrt{W^2\star\mu} \in
{\cal A}_{loc}^+(\mathbb H).$$

\noindent  we introduce $\mu^{\mathbb G } := I_{\Rbrack
0,\tau\Lbrack}\star \mu$ and its  $\mathbb G$ compensated measure
\begin{eqnarray} \label{canonicaldecompoS}
 \nu^{\mathbb G}(dt,dx):=(1+f_m(x)/Z_{t-})I_{\Rbrack 0,\tau\Lbrack}(t)\nu(dt,dx).\end{eqnarray}

\noindent Below, we state our general result that extend the previous theorem.

\begin{theorem}\label{theoremMgDensity2} Suppose that $S$ is an $\mathbb F$-quasi-left-continuous local martingale. Consider $S^{(0)},\ \psi,\ \mbox{and}\ \Lopt$ defined in (\ref{Szero})
and (\ref{Ntilde}) respectively. If $\left(S, S^{(0)}\right)$ is
an $\mathbb F$-local martingale, then ${\cal
E}\left(\Lopt+{\Lopt}^{(1)}\right)S^{\tau}$ is a $\mathbb G$-local
martingale, where
\begin{equation}\label{mgdensity2}
{\Lopt}^{(1)}:=g_1\star(\mu^{\mathbb G}-\nu^{\mathbb G }),\ \ \mbox{and}\ \ \  g_1:={{1-\psi}\over{1+f_m/Z_{-}}}I_{\{\psi>0\}}.\end{equation}
\end{theorem}

\begin{proof} We start by recalling from (\ref{psifm}) that $\{\psi=0\}=\{Z_{-}+f_m=0\}$ $M^P_{\mu}-a.e.$. Thus the functional $g_1$ is a well defined non-negative $\widetilde{\cal P}(\mathbb F)-$measurable functional. The proof of the theorem will completed in two steps. In the first step we prove that the process $L^{(1)}$ is a well defined local martingale, while in the second step we prove the main statement of the theorem.\\
1) Herein, we prove that the integral $ g_1\star
\left(\mu^{\mathbb{G}} - \nu^\mathbb{G}\right) $ is well-defined.
To this end, it  is enough to prove that $g_1\star
\mu^{\mathbb{G}}\in{\cal A}^+(\mathbb G)$. Therefore, remark that
$$(1-\psi)I_{\{0<Z_{-}\}}=M_\mu^P\left( I_{\{\widetilde{Z} =
0<Z_{-}\}} | \widetilde{\cal P}(\mathbb{F})\right)=M_\mu^P\left(
I_{\Rbrack \widetilde R_0\Lbrack} | \widetilde{\cal
P}(\mathbb{F})\right)I_{\{0<Z_{-}\}},$$and calculate
\begin{eqnarray*}
E\left(g_1\star\mu^{\mathbb G}(\infty)\right)&=&E\left(g_1{\widetilde Z}\star\mu(\infty)\right)\\
\\
&\leq& E\left(I_{\Rbrack \widetilde R_0\Lbrack}\star\mu(\infty)\right)= P\left(\Delta S_{\widetilde R_0}\not=0\ \&\
\widetilde R_0<+\infty\right)\leq 1.\end{eqnarray*}
Thus, the process ${\Lopt}^{(1)}$ is a well defined $\mathbb G$-martingale.\\
 \noindent 2) In this part, we prove that  ${\cal E}\left( \Lopt + \Lopt^{(1)} \right) S^\tau$
 is a $\mathbb{G}$-local martingale. To this end, it is enough to prove that $\langle S^\tau,  \Lopt + \Lopt^{(1)}\rangle^\mathbb{G}$ exists and
 \begin{equation}\label{martingeforStau}
  S^\tau + \left \langle S^\tau,  \Lopt + g_1\star \left(\mu^\mathbb{G} - \nu^\mathbb{G}\right)\right\rangle^\mathbb{G}\ \ \ \mbox{is a $\mathbb{G}$-local martingale.}
 \end{equation}
  Recall that $$ \Lopt= -{{Z_{-}^2}\over{Z_{-}^2+\Delta \langle m\rangle^{\mathbb F}}}\,{1\over{\widetilde Z}} \,I_{\Lbrack 0,\tau\Lbrack} \odot\widehat{m},$$
  and hence $\langle S^\tau,  \Lopt\rangle^\mathbb{G}$ exists due to Proposition \ref{lemmforNtilde}--(b). By stopping, there is no loss of generality in assuming that $S$ is a true martingale. Then, using similar calculation as in the first part 1), we can easily prove that $$E\left[\vert x\vert g_1\star\mu^{\mathbb G}(\infty)\right]\leq E\left(\vert \Delta S_{\widetilde R_0}\vert I_{\{\widetilde R_0<+\infty\}}\right)<+\infty.$$ This proves that $\left \langle S^\tau,  \Lopt+{\Lopt}^{(1)}\right\rangle^\mathbb{G}$ exists. Now, we calculate and simplify the expression in (\ref{martingeforStau}) as follows.
 \begin{eqnarray*}
  S^\tau &+& \left \langle S^\tau,  \Lopt+ g_1\star \left(\mu^\mathbb{G} - \nu^\mathbb{G}\right)\right\rangle^\mathbb{G} = \widehat{S} + \frac{1}{Z_{-}}I_{\Lbrack 0,\tau\Lbrack}\centerdot \langle S, m\rangle^{\mathbb F} + \left \langle S^\tau,  \Lopt \right\rangle^\mathbb{G} + xg_1\star \nu^\mathbb{G}\\
  &=& \widehat{S} + \frac{1}{Z_{-}}I_{\Lbrack 0,\tau\Lbrack}\centerdot\langle S, m\rangle  -\frac{1}{Z_{-}}I_{\Lbrack 0,\tau\Lbrack}\cdot \left( I_{\{\widetilde{Z} >0\}}\cdot \left[S, m\right] \right)^{p,\mathbb{F}} + xg_1\star \nu^\mathbb{G}\\
  &=& \widehat{S} +\frac{1}{Z_{-}}I_{\Lbrack 0,\tau\Lbrack}\centerdot \left( I_{\{\widetilde{Z} =0\}}\centerdot \left[S, m\right] \right)^{p,\mathbb{F}} + x M_\mu^P\left( I_{\{\widetilde{Z} = 0<Z_{-}\}} | \widetilde{\cal P}(\mathbb{F})\right)I_{\left\{Z_{-} + f_m >0\right\}}I_{\Lbrack 0,\tau\Lbrack}\star \nu \\
  &=& \widehat{S}   - x M_\mu^P\left( I_{\{\widetilde{Z} = 0<Z_{-}\}} | \widetilde{\cal P}(\mathbb{F})\right)I_{\left\{\psi=0\right\}}I_{\Lbrack 0,\tau\Lbrack}\star \nu  =\widehat{S}\in {\cal M}_{loc}(\mathbb G).
 \end{eqnarray*}
 The second equality is due to (\ref{equa2december3}), while the last equality follows directly  form the fact that $S^{(0)}$ is an $\mathbb F$-local martingale (which is equivalent to $xI_{\{\psi=0<Z_{-}\}}\star\nu\equiv 0$) and $M_\mu^P\left( I_{\{\widetilde{Z} = 0<Z_{-}\}} | \widetilde{\cal P}(\mathbb{F})\right)=I_{\{0<Z_{-}\}}(1-\psi)$. This ends the proof of the theorem.
\end{proof}


\begin{remark}\label{RemarkForExplcitDefaltors} 1) Both Theorems \ref{individualSbeforetau}-\ref{theoremMgDensity2} provide methods that build-up explicitly $\sigma$-martingale density for $X^{\tau}$, whenever $X$ is an $\mathbb F$-quasi-left-continuous process fulfilling the assumptions of the theorems respectively.\\
2) The extension of Theorem \ref{individualSbeforetau} to the
general case where $S$ is an $\mathbb F$-local martingale (not
necessarily quasi-left-continuous) boils down to find a
predictable process $\Phi$ such that $Y^{(1)}:={\cal E}(\Phi\is
L)$ will be the martingale density for $S^{\tau}$. Finding the
process $\Phi$ will be easy to guess when we will address the case
of thin semimartingale. However the proof of $Y^{(1)}$ is a local
martingale density for $S^{\tau}$ is  very technical. The
extension of Theorem \ref{theoremMgDensity2} to the case of
arbitrary $\mathbb F$-local martingale $S$ requires additional
careful modification of the functional $g_1$ so that $1+\Delta
{\Lopt}+\Delta{\Lopt^{(1)}}$ remains positive. While both
extensions remain very feasible, we opted to not overload the
paper with technicalities.\end{remark}
\section{Proofs of Main Theorems}\label{proofs}
This section is devoted to the proofs of Theorems \ref{main2},
\ref{main3} and \ref{main4}. They are quite long, since some
integrability results have to be proved. For the reader's convenience, we recall the canonical decomposition of $S$, given in Appendix A, by
$$
S=S_0+S^c+h\star(\mu-\nu)+b\cdot A+(x-h)\star\mu,$$ where $h$
defined as $h(x):=xI_{\{ \vert x\vert\leq 1\}}$ is the truncation
function. The canonical decomposition of $S^{\tau}$ under $\mathbb
G$ is given by
$$\begin{array}{lll}
S^{\tau}=S_0+\widehat{S^c}+h\star(\mu^{\mathbb G}-\nu^{\mathbb G})+{{c\beta_m}\over{Z_{-}}}I_{\Lbrack 0,\tau\Lbrack}\cdot A+h{{f_m}\over{Z_{-}}}I_{\Lbrack 0,\tau\Lbrack}\star\nu+b\is A^{\tau}+(x-h)\star\mu^{\mathbb G}\end{array}$$
where $\mu^{\mathbb G}$ and $\nu^{\mathbb G}$ and $(\beta_m, f_m)$ are given in (\ref{canonicaldecompoS}) and (\ref{decompositionofm}) respectively and
$$
\widehat{S^c}:=I_{\Lbrack 0,\tau\Lbrack}\is S^c-{1\over{Z_{-}}}I_{\Lbrack 0,\tau\Lbrack}\is \langle m, S^c\rangle.$$

\subsection{ Proof of Theorem \ref{main2}}\label{main2section}

The proof of Theorem \ref{main2} will be completed in four steps.
The first  step provides an equivalent formulation to assertion
(a) using the filtration $\mathbb F$ instead. In the second  step,
we   prove (a)$\Rightarrow $(b), while the reverse  implication
is  proved in the third step. The proof of
(b)$\Longleftrightarrow$ (c) is given in the last step.
 \\

\noindent {\bf Step 1: Formulation  of assertion (a):}  Thanks to
Proposition \ref{charaterisationofNUPBRloc}, $S^\tau$ satisfies
NUPBR($\mathbb{G}$) if and only if there exist a
$\mathbb{G}$-local martingale $N^{\mathbb G}$ with $1+ \Delta
N^{\mathbb G} >0$ and a $\mathbb G$-predictable process
$\phi^{\mathbb G}$ such that $0<\phi^{\mathbb G}\leq 1$ and ${\cal
E}\left( N^{\mathbb G}\right)\left(\phi^{\mathbb G}\is
S^\tau\right)$ is a $\mathbb{G}$-local martingale.
We can
reduce our attention to  processes $N^{\mathbb G}$ such that (see   Theorem
\ref{theosigmadensityiff} in the Appendix)   $N^{\mathbb G}$ has
the form
$$ N^{\mathbb G}= \beta^{\mathbb G}\centerdot
\widehat{S^c}+( {f}^{\mathbb G}-1)\star(\mu^{\mathbb{G}}
-\nu^{\mathbb G})
$$ where
  $ \beta^{\mathbb G}\in L( \widehat{S^c},\mathbb G)$
and  $ {f}^{\mathbb G}$ is a positive $\widetilde{\cal P}(\mathbb
G)$-measurable functional.

\noindent  Then, one notes that ${\cal E}\left(
N^{\mathbb G} \right)\left(\phi^{\mathbb G}\is S^\tau\right)$  is
a $\mathbb G$-local martingale if and only if $\phi^{\mathbb G}\is
S^\tau + [\phi^{\mathbb G}\is S^\tau, N^{\mathbb G}]$ is a
$\mathbb{G}$-local martingale, which in turn, is equivalent to
\begin{equation}\label{integrabilityG}
 \phi^{\mathbb G}\vert x  {f}^{\mathbb G}(x) -h(x)\vert\left( 1+ \frac{f_m(x)}{Z_{-}} \right)I_{\Rbrack 0,\tau\Lbrack}\star \nu \in {\cal A}^+_{loc}(\mathbb G),
\end{equation}
and $P\otimes A-a.e.$ on $\Rbrack 0,\tau\Lbrack$, introducing the
 kernel $F$ defined in the Appendix \ref{Representationlocalmart}
\begin{equation}\label{mgEquationG}
 b+c(\frac{\beta_m}{Z_{-}} + {\beta}^{\mathbb G})+ \int \left[ (x  f ^{\mathbb G}(x)-h(x))\left( 1+ \frac{f_m(x)}{Z_{-}} \right)-h(x){{f_m(x)}\over{Z_{-}}}\right]   F(dx) = 0.
\end{equation}

\noindent From Lemma \ref{lemma:predsetFG}, there exist $\phi^{\mathbb F}$ and $ \beta^{\mathbb F}$ two $\mathbb F$-predictable processes and a positive $\widetilde{\cal P}(\mathbb F)$-measurable functional, ${ f}^{\mathbb F}$, such that $0<\phi^{\mathbb F}\leq 1$,
\begin{equation}\label{fromG2Felements}
 \beta^{\mathbb F}= \beta^{\mathbb G},\ \phi^{\mathbb F}=\phi^{\mathbb F},\ {  f}^{\mathbb F}={ f}^{\mathbb G}\ \mbox{on}\ \Rbrack 0,\tau\Lbrack.\end{equation}
In virtue of these and taking account integrability conditions given in  Proposition \ref{prop:alocundergf}, we deduce that (\ref{integrabilityG})--(\ref{mgEquationG}) imply that,  on $\{Z_{-}\geq \delta \}$, we have
\begin{equation}\label{integrabilityGbis}
 W^{\mathbb F}:=\int \vert (x  {f}^{\mathbb F}(x) -h(x))\vert\left( 1+ \frac{f_m(x)}{Z_{-}} \right)F(dx)<+\infty\ \ \ P\otimes A-a.e,
\end{equation}
and $P\otimes A$-a.e. on $\{Z_{-}\geq \delta \}$, we have
\begin{equation}\label{mgEquationGbis}
 b+c\left( {\beta}^{\mathbb F}+\frac{\beta_m}{Z_{-}}\right)-\int h(x) I_{\{ \psi=0\}}F(dx) + \int \left[x {f}^{\mathbb F}(x)( 1+ \frac{f_m(x)}{Z_{-}})-h(x) \right]I_{\{\psi>0\}}F(dx) = 0.
\end{equation}
Due to (\ref{integrabilityGbis}), this latter equality follows immediately by taking the $\mathbb F$-predictable projection of (\ref{mgEquationG}) after inserting (\ref{fromG2Felements}).\\

\noindent {\bf Step 2: Proof of (a) $\Rightarrow$ (b).} Suppose
that $S^\tau$ satisfies NUPNR($\mathbb{G}$),  hence
(\ref{integrabilityGbis})--(\ref{mgEquationGbis}) hold. To prove
that $I_{\{ Z\geq \delta\}}\is(S-S^{(0)})$ satisfies
NUPBR$(\mathbb F)$, we consider
\begin{equation}
 \beta:= \left(\frac{\beta_m}{Z_{-}}+  {\beta}^{\mathbb F} \right)I_{\{Z_{-}\geq \delta\}}\ \mbox{and}\ \  f =  {f}^{\mathbb F}  \left( 1+ \frac{f_m}{Z_{-}} \right) I_{\{Z_{-}\geq \delta\ \&\ \psi>0\}} +  I_{\{0\leq Z_{-}< \delta\ \mbox{or}\ \psi=0\}}.
\end{equation}
If  $\beta\in L(S^c,\mathbb F)$ and $(f-1)\in {\cal G}^1_{loc}(\mu,\mathbb F)$,   we conclude that
\begin{equation}
 N := \beta\centerdot S^c + (f-1)\star (\mu - \nu).
\end{equation}
is a well defined $\mathbb F$-local martingale. Therefore, by choosing $\phi=\left(1+W^{\mathbb F}I_{\{Z_{-}\geq \delta\}}\right)^{-1}$, using (\ref{mgEquationGbis}), and applying It\^o's formula for ${\cal E}(N)\left(\phi I_{\{Z_{-}\geq \delta\}}\cdot ( S -S^{(0)}) \right)$, we deduce that this process is a local martingale.  Hence, $ I_{\{Z_{-}\geq \delta\}}\cdot ( S -S^{(0)})$ satisfies   NUPBR$(\mathbb F)$, and the proof of (a)$\Rightarrow$(b) is completed.\\
Now we focus on proving $\beta\in L(S^c)$ and $(f-1)\in {\cal
G}^1_{loc}(\mu,\mathbb F)$ (or equivalently
$\sqrt{(f-1)^2\star\mu}\in {\cal A}^+_{loc}(\mathbb F)$). Since
$\beta_m\in  L(S^c)$, then it is obvious that
$\frac{\beta_m}{Z_{-}}I_{\{Z_{-}\geq \delta\}}\in L(S^c)$ on the
one hand. On the other hand, $( {\beta}^{\mathbb F})^T c
{\beta}^{\mathbb F}I_{\{0\leq Z_{-}< \delta\}}\is A\in {\cal
A}^+_{loc}(\mathbb F)$ due to $( {\beta}^{\mathbb F})^T c
{\beta}^{\mathbb F}\is A^{\tau}=( {\beta}^{\mathbb G})^T c
{\beta}^{\mathbb G}\is A^{\tau}\in {\cal A}^+_{loc}(\mathbb G)$
and
Proposition \ref{prop:alocundergf}--(c). This completes the proof of $\beta\in L(S^c)$.\\

\noindent Now, we focus on proving $(f-1)\in {\cal G}^1_{loc}(\mu,\mathbb F)$. Since $S$ is quasi-left-continuous, this is equivalent to prove $\sqrt{(f-1)^2\star \mu} \in {\cal A}^+_{loc}(\mathbb{F})$. Thanks to Proposition \ref{prop:alocundergf} and $\sqrt{(\  {f}^{\mathbb F} -1)^2\star \mu^\mathbb{G}}= \sqrt{(\  {f}^{\mathbb G} -1)^2\star \mu^\mathbb{G}}\in {\cal A}^+_{loc}(\mathbb{G})$,
we deduce that
\begin{equation}
 (\  {f}^{\mathbb F} -1)^2I_{\{ |  {f}^{\mathbb F} -1|\leq \alpha\}} \widetilde{Z}I_{\{Z_{-}
 \geq \delta\}} \star \mu \ \mbox{and~}\
  |  {f}^{\mathbb F} -1|I_{\{ |  {f}^{\mathbb F} -1|> \alpha\}} \widetilde{Z}I_{\{Z_{-}\geq \delta\}} \star \mu \in {\cal A}^+_{loc}(\mathbb{F}).
\end{equation}
By stopping, there is no loss of generality in assuming that these two processes and $[m,m]$ are integrable.
By putting $\Sigma_0:=\{Z_{-}\geq \delta\ \&\ \psi>0\}$, then we get
\begin{equation}\label{decompositionoff}
 f - 1 = \left(   {f}^{\mathbb F} -1 \right) \left( 1 + \frac{f_m}{Z_{-}} \right)I_{\Sigma_0} +\frac{f_m}{Z_{-}}I_{\Sigma_0} =: h_1 + h_2.
\end{equation}
Therefore, we derive that
\begin{eqnarray*}
 E\left[h_1^2 I_{\{ | {f}^{\mathbb F} -1|\leq \alpha\}}\star \mu_{\infty} \right] &\leq& \delta ^{-2} E\left[\left(  {f}^{\mathbb F} -1 \right)^2 \left( Z_{-} + f_m \right)^2 I_{\{ | {f}^{\mathbb F} -1|\leq \alpha\}}I_{\{Z_{-}\geq \delta\}}\star \mu_{\infty} \right]\\
 &\leq& \delta ^{-2}  E\left[\left(  {f}^{\mathbb F} -1 \right)^2 \widetilde{Z}I_{\{ | {f}^{\mathbb F} -1|\leq \alpha\}}I_{\{Z_{-}\geq \delta\}}\star \mu_{\infty} \right]<+\infty,
\end{eqnarray*}
and
\begin{eqnarray*}
 E\left[|h_1| I_{\{ | {f} -1|> \alpha\}}\star \mu_{\infty} \right] &\leq& \delta^{-1}  E\left[ | {f} -1 | \ | Z_{-} + f_m | I_{\{ |  {f} -1|> \alpha\}}I_{\{Z_{-}\geq \delta\}}\star \mu_{\infty} \right]\\
 &=& \delta^{-1}  E\left[ |  {f} -1 | \widetilde{Z}I_{\{ |  {f} -1|> \alpha\}}I_{\{Z_{-}\geq \delta\}}\star \mu_{\infty} \right]<+\infty.
\end{eqnarray*}
By combining the above two inequalities, we conclude that  $\left(h_1^2\star\mu\right)^{1/2}\in {\cal A}^+_{loc}(\mathbb F).$  It is easy to see that $\left(h_2^2\star\mu\right)^{1/2}\in {\cal A}^+_{loc}(\mathbb F)$ follows from
\begin{eqnarray*}
 E\left[h_2^2\star \mu_\infty \right] &\leq & \delta^{-2}  E\left[f_m^2\star \mu_\infty \right]\leq  \delta^2 E\left[(\Delta m)^2\star \mu_\infty \right] \leq \delta^{-2}  E\left[ m,m\right]_\infty <+\infty.
\end{eqnarray*}

\noindent {\bf Step 3: Proof of (b) $\Rightarrow$ (a)}. Suppose
that for any $\delta>0$, the process $I_{\{Z_{-}\geq
\delta\}}\cdot \left( S -S^{(0)} \right)$ satisfies
NUPBR$(\mathbb{F})$.   Then, there exist an $\mathbb{F}$-local
martingale $N^\mathbb{F}$ and an $\mathbb F$-predictable process
$\phi$ such that $0<\phi\leq 1$ and ${\cal E}\left( N^\mathbb{F}
\right)\left[\phi I_{\{Z_{-}\geq \delta\}}\centerdot\left( S
-S^{(0)} \right)\right] $ is an $\mathbb{F}$-local martingale.
Again, thanks to Theorem \ref{theosigmadensityiff}, we can
restrict our attention to the case
\begin{equation}
 N^\mathbb{F} := \beta^{\mathbb{F}}\centerdot  S^c + (f^\mathbb{F}-1)\star (\mu -
 \nu),
\end{equation} where    $\beta^\mathbb{F}\in L(S^c)$ and $f^{\mathbb{F}}$ is a
positive ${\widetilde{\cal P}}(\mathbb F)$-measurable functional.

\noindent Thanks to It\^o's formula, the fact  that ${\cal
E}\left( N^\mathbb{F} \right)\left[\phi I_{\{Z_{-}\geq
\delta\}}\centerdot \left( S -S^{(0)} \right)\right] $ is an
$\mathbb{F}$-local martingale implies that on
$\{Z_{-}\geq\delta\}$
\begin{equation}\label{integrabilityF}
k^{\mathbb F}:=\int \vert xf^{\mathbb
F}(x)I_{\{\psi(x)>0\}}-h(x)\vert F(dx)<+\infty\ \ \ \ P\otimes
A-a.e.\end{equation} and $P\otimes A$-a.e. on $\{ Z_{-}\geq
\delta\}$, we have
\begin{equation}\label{eq:equiv4}
 b-\int h(x) I_{\{\psi = 0\}}F(dx) + c \beta^\mathbb{F} +\int \left[  x f^\mathbb{F}(x)-h(x)\right]I_{\{\psi > 0\}}F(dx)= 0.
\end{equation}
Consider
\begin{equation}\label{beta/f}
 \beta^\mathbb{G}:= \left( \beta^\mathbb{F} - \frac{\beta_m}{Z_{-}} \right)I_{\Lbrack
0,\tau\Lbrack} \ \ \ \mbox{and } \ \ f^\mathbb{G}:= \frac{f^\mathbb{F}}{1 + f_m/Z_{-}} I_{\{\psi > 0\}}I_{\Lbrack 0,\tau\Lbrack} + I_{\{\psi = 0\}\cup\Lbrack \tau, +\infty\Rbrack}.
\end{equation}
and assume that
\begin{equation}\label{integrability2}
\beta^{\mathbb G}\in L(\widehat {S^c})\ \ \ \mbox{and}\ \ \
(f^{\mathbb G}-1)\in {\cal G}^1_{loc}(\mu^{\mathbb
G}).\end{equation} Then, necessarily $ N^\mathbb{G} :=
\beta^\mathbb{G}\centerdot \widehat{S^c} + (f^\mathbb{G}-1)\star
(\mu^\mathbb{G} - \nu^\mathbb{G})$ is a well defined $\mathbb
G$-local martingale satisfying ${\cal E}( N^\mathbb{G} )>0$.
Furthermore, due to (\ref{eq:equiv4}) and to
$\{\psi=0\}=\{Z_{-}+f_m=0\}$  (see(\ref{psifm})), on $\Lbrack
0,\tau\Lbrack $ we obtain
\begin{equation}\label{eq:equiv3}
 b +c\left(\beta^\mathbb{G} +\frac{\beta_m}{Z_{-}}\right) + \int \left(x f^\mathbb{G} \left( 1+ \frac{f_m}{Z_{-}} \right) - h(x)\right)F(dx)= 0,
\end{equation}
Then, by taking $\phi^{\mathbb G}:=\left(1+k^{\mathbb F}I_{\{
Z_{-}\geq\delta\}}\right)^{-1}$, and applying It\^o's formula for
$(\phi^{\mathbb G}I_{\{ Z_{-}\geq\delta\}}\is S^{\tau}){\cal
E}(N^\mathbb{G})$, we conclude that this process is a
$\mathbb{G}$-local martingale due to (\ref{eq:equiv3}). Thus,
$I_{\{ Z_{-}\geq\delta\}}\is S^{\tau}$ satisfies   NUPBR$(\mathbb
G)$ as long as (\ref{integrability2}) is fulfilled.

 \noindent Since $Z^{-1}I_{\Rbrack 0,\tau\Rbrack}$ is $\mathbb G$-locally bounded, then there exists a family of $\mathbb G$-stopping times $(\tau_{\delta})_{\delta>0}$ such that $ \Rbrack 0,\tau_{\delta}\Lbrack\subset\{Z_{-}\geq\delta\}$ (or equivalently $I_{\{Z_{-}\geq \delta\}}\is S^{\tau\wedge\tau_{\delta}}=S^{\tau\wedge\tau_{\delta}}$) and $\tau_{\delta}$ increases to infinity when $\delta$ goes to zero. Thus, using Proposition \ref{NUPBRLocalization}, we deduce that $S^{\tau}$
 satisfies  NUPBR$(\mathbb G)$.
 This achieves the proof of (b)$\Rightarrow$(a) under (\ref{integrability2}).\\

\noindent To prove that (\ref{integrability2}) holds true, we
remark
  in a first step that $Z^{-1}_{-}I_{\Lbrack
0,\tau\Lbrack}$ is $\mathbb G$-locally bounded and both $\beta_m$ and $\beta^{\mathbb F}$ belong to $L(S^c)$. This, easily, implies that $\beta^{\mathbb G}\in L(\widehat {S^c})$.\\
Now, we prove that $\sqrt{(f^\mathbb{G}-1)^2\star \mu^\mathbb{G}} \in {\cal A}^+_{loc}(\mathbb{G})$. Since $\sqrt{( f^\mathbb{F} -1)^2\star \mu} \in {\cal A}^+_{loc}(\mathbb{F})$, Proposition \ref{prop:alocundergf}  allows us again to deduce that
\begin{equation}
 ( f^{\mathbb F} -1)^2I_{\{ | f^{\mathbb F} -1|\leq \alpha\}}   \star \mu \in {\cal A}^+_{loc}(\mathbb{F}) \  \mbox{and } \
  | f^{\mathbb F} -1|I_{\{ | f^{\mathbb F} -1|> \alpha\}}    \star \mu \in {\cal A}^+_{loc}(\mathbb{F}).
\end{equation}
Without loss of generality, we assume that these two processes and $[m,m]$ are integrable. Put
\begin{equation}
 f^\mathbb{G} - 1 =   I_{\{\psi > 0\}}I_{\Lbrack 0,\tau\Lbrack} \frac{Z_{-}(f^\mathbb{F} - 1)}{f_m + Z_{-}}-   I_{\{\psi > 0\}}I_{\Lbrack 0,\tau\Lbrack} \frac{f_m}{f_m + Z_{-}} := f_1 + f_2.
\end{equation}
Then, we calculate
$$\begin{array}{lll}
E\left(f_1^2 I_{\{ f_m+Z_{-}>\delta/2\}\cap\{\vert f^{\mathbb
F}-1\vert\leq \alpha\}} \star\mu^{\mathbb G}_{\infty}\right)\leq
\left({{2}\over{\delta}}\right)^2 E\bigl[(f^{\mathbb F}-1)^2
I_{\{\vert f^{\mathbb F}-1\vert\leq
\alpha\}}\star\mu_{\infty}\bigr]<+\infty.
\end{array}$$
and
\begin{eqnarray*}
E\sqrt{f_1^2 I_{\{ f_m+Z_{-}\leq \delta/2\}\cap\{\vert f^{\mathbb F}-1\vert\leq \alpha\}}\star\mu^{\mathbb G}_{\infty}}&\leq& \alpha E\bigl(I_{\{ f_m+Z_{-}\leq \delta/2\}}(Z_{-}+f_m)^{-1}\star\mu^{\mathbb G}({\infty})\bigr)\\
&\leq &E\bigl(I_{\{\vert f_m\vert \geq \delta/2\}}\star\mu({\infty})\bigr)\leq {{4\alpha}\over{\delta^2}}E[m,m]_{\infty}<+\infty.
\end{eqnarray*}
This proves that $\sqrt{f_1^2 I_{\{\vert f^{\mathbb F}-1\vert\leq \alpha\}}\star\mu^{\mathbb G}}\in {\cal A}^+_{loc}(\mathbb G).$ Similarly,
we calculate
\begin{eqnarray*}
E\sqrt{f_1^2 I_{\{\vert f^{\mathbb F}-1\vert> \alpha\}}\star\mu^{\mathbb G}_{\infty}}&\leq&
E\bigl(\vert f_1\vert I_{\{\vert f^{\mathbb F}-1\vert> \alpha\}}\star\mu^{\mathbb G}_{\infty}\bigr)
\leq E\bigl({{\vert f^{\mathbb F}-1\vert}\over{1+f_m/Z_{-}}}I_{\{\vert f^{\mathbb F}-1\vert> \alpha\}}\star\mu^{\mathbb G}_{\infty}\bigr)\\
&\leq& E\bigl(\vert f^{\mathbb F}-1\vert I_{\{\vert f^{\mathbb
F}-1\vert> \alpha\}} \star\mu_{\infty}\bigr)<+\infty.
\end{eqnarray*}
Thus, by combining all the remarks obtained above, we conclude that $\sqrt{f_1^2 \star\mu^{\mathbb G}}$ is $\mathbb G$-locally integrable. For the functional $f_2$, we proceed as follows. We calculate
\begin{eqnarray*}
E\bigl(f_2^2 I_{\{ f_m+Z_{-}> \delta/2\}}\star\mu^{\mathbb
G}_{\infty}\bigr) \leq (2/\delta)^2
E\bigl(f_m^2\star\mu_{\infty}\bigr)\leq (2/\delta)^2
E[m,m]_{\infty}<+\infty,
\end{eqnarray*}
and
\begin{eqnarray*}
E\sqrt{f_2^2 I_{\{f_m+Z_{-}\leq \delta/2\}}\star\mu^{\mathbb G}_{\infty}}&\leq &
E\bigl(\vert f_m\vert I_{\{\vert f_m\vert\geq \delta/2\}}\star\mu({\infty})\bigr)\\
&\leq& (2/\delta) E\bigl(f_m^2\star\mu({\infty})\bigr)\leq (2/\delta) E[m,m]_{\infty}<+\infty.
\end{eqnarray*}
This proves that $\sqrt{f_2^2\star\mu^{\mathbb G}}$ is $\mathbb G$-locally integrable. Therefore, we conclude that (\ref{integrability2}) is valid, and the proof of (b)$\Rightarrow$(a) is completed.\\

\noindent {\bf Step 3: Proof of (b) $\Longleftrightarrow$ (c)}.
For any $\delta>0$ and   any $n\in \mathbb N$, we denote
$$
\sigma_{\infty}:=\inf\{t\geq 0:\ \ Z_{t}=0\},\ \ \ \
\tau_{\delta}:=\sup\{t\,: Z_{t-}\geq \delta\}.$$ Then, due to
$\Lbrack\sigma_{\infty},+\infty\Rbrack\subset\{Z_{-}=0\}\subset\{Z_{-}<\delta\}$,
we deduce
$$
\sigma_{1/\delta}\leq \tau_{\delta}\leq \sigma_{\infty}\ \ \ \
\mbox{and}\ \  \ Z_{{\tau_{\delta}}-}\geq \delta>0\ \ \ \ \
P-a.s.\ \ \mbox{on}\ \  \{\tau_\delta <\infty\}\,.
$$
Furthermore, setting $\Sigma:=\bigcap_{n\geq
1}(\sigma_n<\sigma_{\infty})$, we have
$$
\mbox{on}\ \ \Sigma\cap\{\sigma_{\infty}<\infty\}\ \ \
Z_{\sigma_{\infty}-}=0,\ \ \mbox{and}\ \ \  \tau_\delta
<\sigma_\infty\ \ \ P-a.s.$$ We introduce the semimartingale
$X:=S-S^{(0)}$. For any $\delta
>0$, and any $H$ predictable such that $H_{\delta}:=HI_{\{
Z_{-}\geq \delta\}}\in L(X)$  and $H_{\delta}\cdot X\geq -1$ , due
to Theorem 23 of \cite{dm2} (page 346 in the French version),
$$
(H_{\delta}\cdot X)_T=(H_{\delta}\cdot X)_{T\wedge\tau_{\delta}},\
\ \ \mbox{and on}\ \{\theta\geq\tau_{\delta}\}\  \
(H_{\delta}\cdot X)_T=(H_{\delta}\cdot X)_{T\wedge\theta}.$$ Then,
for any $T\in (0,+\infty)$, we calculate the following
\begin{equation}\label{probbailityboundedness}
\begin{array}{lllll}
P((H_{\delta}\cdot X)_T>c)=P((H_{\delta}\cdot X)_T>c\ \&\ \sigma_n\geq\tau_{\delta})+P((H_{\delta}\cdot X)_T>c\ \&\ \sigma_n<\tau_{\delta})\\
\\
\hskip 1cm \leq 2\displaystyle\sup_{\phi \in L(X^{\sigma_n}): \phi\centerdot
X^{\sigma ^n} \geq -1  }P((\phi\cdot X)_{ \sigma_n \wedge T}>c)+P(
\sigma_n<\tau_{\delta}\wedge T).
\end{array}
\end{equation}
It is easy to prove that $P( \sigma_n<\tau_{\delta}\wedge
T)\longrightarrow 0$ as $n$ goes to infinity. This can be seen due
to the fact that on $\Sigma$, we have, on the one hand,
$\tau_{\delta}\wedge T<\sigma_{\infty}$ (by differentiating the
two cases whether $\sigma_{\infty}$ is finite or not). On the
other hand,   the event   $(\sigma_n<\sigma_{\infty})$ increases
to $\Sigma$ with $n$. Thus, by combining these, we obtain the
following
\begin{equation}\label{convergence}
\begin{array}{lll}
P( \sigma_n<\tau_{\delta}\wedge T)=P( (\sigma_n<\tau_{\delta}\wedge T)\cap\Sigma)+P((\sigma_n<\tau_{\delta}\wedge T)\cap\Sigma^c)\\
\\
\hskip 1cm \leq P(\sigma_n<\tau_{\delta}\wedge
T<\sigma_{\infty})+P((\sigma_n<\sigma_{\infty})\cap\Sigma^c)\longrightarrow
0.\end{array}
\end{equation}

\noindent Now suppose that for each $n\geq 1$, the process
$(S-S^{(0)})^{\sigma_n}$ satisfies  NUPBR$(\mathbb F)$. Then a
combination of (\ref{probbailityboundedness}) and
(\ref{convergence}) implies that for any $\delta>0$, the process
$I_{\{ Z_{-}\geq \delta\}}\is X:=I_{\{ Z_{-}\geq \delta\}}\is
(S-S^{(0)})$ satisfies   NUPBR$(\mathbb F)$, and the proof of
(c)$\Rightarrow$ (b) is completed. The proof of the reverse
implcation is obvious due to the fact that
$$
\Rbrack 0,\sigma_n\Lbrack\subset\{Z_{-}\geq 1/n\}\subset\{
Z_{-}\geq \delta\},\ \ \ \mbox{for}\ \ \ n\leq \delta^{-1},$$
which implies that $(I_{\{Z_{-}\geq\delta\}}\cdot
X)^{\sigma_n}=X^{\sigma_n}$. This ends the proof of (b)
$\Longleftrightarrow $(c), and the proof of the theorem is
achieved.

\subsection{Intermediate Result}\label{proofsOfmainTheorems3}

The proofs of Theorems \ref{main3} and \ref{main4} rely on the following intermediatory result about $\mathbb F$-martingales with a single jump, which is interesting in itself.

 \begin{proposition}\label{cruciallemma1} Let $M$ be an $\mathbb F$-martingale given by $M:=\xi I_{\Rbrack T,+\infty\Rbrack}$, where  $T$ is an $\mathbb F$-predictable stopping time, and $\xi$ is an ${\cal F}_T$-measurable random variable. Then the following assertions are equivalent.\\
{\rm{(a)}} $M$ is an $\mathbb F$-martingale under $Q_T$ given by
\begin{equation}\label{QT}
 {{dQ_T}\over{dP}}:={{I_{\{ \widetilde Z_T>0\}\cap\Gamma(T)}}\over{P(\widetilde Z_T>0\big|\ {\cal F}_{T-})}}+I_{\Gamma^c(T)},\ \ \ \ \ \Gamma(T):=\{P(\widetilde Z_T>0|{\cal F}_{T-})>0\}. \end{equation}
{\rm{(b)}} On the set $\{ T<+\infty\}$, we have
  \begin{equation}\label{zeroequationbeforetau}
E\left(M_T I_{\{{\widetilde Z}_T=0<Z_{T-}\}}\big|\ {\cal
F}_{T-}\right)=0,\ \ \ P-a.s.\end{equation} {\rm{(c)}} $M^{\tau}$
is a $\mathbb G$-martingale under $Q^{\mathbb
G}_T:=\left(U^{\mathbb G} (T)/E(U^{\mathbb G} (T)\big|\ {\cal
G}_{T-})\right)\dot P$ where
\begin{equation}\label{QGbeforetau}
U^{\mathbb
G}(T):=I_{\{T>\tau\}}+I_{\{T\leq\tau\}}{{Z_{T-}}\over{\widetilde
Z_{T}}}>0. \end{equation}
\end{proposition}

\begin{proof} The proof will be achieved in two steps.\\
{\bf Step 1.} Here, we prove the equivalence between assertions
(a) and (b). For simplicity we denote by $Q:=Q_T$, where $Q_T$ is
defined in (\ref{QT}), and remark that on $\{ Z_{T-}=0\}$, $Q$
coincides with $P$ and (\ref{zeroequationbeforetau}) holds. Thus,
it is enough to prove the equivalence between (a) and (b) on the
set $\{T<+\infty\ \&\ Z_{T-}>0\}$. On this set, due to $E(X |{\cal
F}_{T-})=0$, we derive

 \begin{eqnarray*} \label{equa30}
 E^{Q}(\xi\big|{\cal F}_{T-})&=&E(\xi I_{\{ \widetilde Z_T>0\}}\big|{\cal F}_{T-})\Bigl(P(\widetilde Z_T>0\big|{\cal F}_{T-})\Bigr)^{-1}\\
 & =&-E(\xi I_{\{ \widetilde Z_T=0\}}\big|{\cal F}_{T-})\left(P(\widetilde Z_T>0\big|{\cal F}_{T-})\Bigr)\right)^{-1}.
\end{eqnarray*}
Therefore, we conclude that assertion (a) (or equivalently $E^Q(\xi|{\cal F}_{T-})=0$) is equivalent to (\ref{zeroequationbeforetau}). This ends the proof of (a) $\Longleftrightarrow$ (b).\\
{\bf Step 2.} To prove (a)$\Longleftrightarrow$(c), we first
notice that due to $(T\leq\tau)\subset(\widetilde Z_{T}>0)\subset
(Z_{T-}>0)$, on $\{T\leq\tau\}$ we have

\begin{eqnarray}
  P\left(\widetilde Z_T>0\big|{\cal F}_{T-}\right)E^{Q^\mathbb{G}_T} \left(\xi | \mathcal{G}_{T-}\right) &=& E \left( \frac{Z_{T-}}{\widetilde{Z}_T} \xi I_{\{ T \leq \tau\}} | \mathcal{G}_{T-}\right)=E\left( \xi I_{\{ \widetilde{Z}_T > 0\}} | \mathcal{F}_{T-}\right)   \nonumber \\
  &=&  E^Q\left(\xi | \mathcal{F}_{T-}\right)  P\left(\widetilde Z_T>0\big|{\cal F}_{T-}\right) \nonumber.
\end{eqnarray}


\noindent This equality proves that $M^{\tau}\in {\cal M}(Q^{\mathbb G},\mathbb G)$
  if and only if $M\in {\cal M}(Q,\mathbb F)$, and the proof of (a)$\Longleftrightarrow$(c) is completed. This ends the proof of the theorem. \end{proof}

\subsection{Proof of Theorem \ref{main3}} \label{proofmain3}

For the reader convenience, in order to prove Theorem \ref{main3},
we   state a more precise version of the theorem, in which we
describe explicitly the choice  for the probability measure $Q_T$.

\begin{theorem}\label{main3bis} Suppose that the assumptions of Theorem \ref{main3} are in force.
Then, the assertions {\rm{(a)}} and {\rm{(b)}}  of Theorem \ref{main3} are equivalent to the following assertions.\\
{\rm{(d)}} $S$ satisfies NUPBR$(\mathbb F,{\widetilde Q}_T)$,
where ${\widetilde Q}_T$ is
\begin{equation}\label{Qtilde(T)}
{\widetilde Q}_T:=\left({{\widetilde Z_T}\over{Z_{T-}}}I_{\{ Z_{T-}>0\}}+I_{\{ Z_{T-}=0\}}\right)\cdot P,\end{equation}
{\rm{(e)}} $S$ satisfies NUPBR$(\mathbb F,Q_T)$, where $Q_T$ is defined in (\ref{QT}).
\end{theorem}

\begin{proof}
The proof of this theorem will be achieved by proving ${\rm{(d)}}\Longleftrightarrow {\rm{(e)}} \Longleftrightarrow {\rm{(b)}}$ and ${\rm{(b)}}\Rightarrow {\rm{(a)}}\Rightarrow {\rm{(d)}}$. These will be carried out in four steps.\\

\noindent{\bf Step 1:} In this step, we  prove
${\rm{(d)}}\Longleftrightarrow {\rm{(e)}}$. Since $S$ is a single
jump process with predictable jump time $T$, then it is easy to
see that $S$ satisfies NUPBR under some probability $R$ is
equivalent to the fact that $I_A S$ and $I_{A^c} S$ satisfies
NUPBR$(R)$ for any ${\cal F}_{T-}$-measurable event $A$. Hence, it
is enough to prove the equivalence between the assertions \rm{(d)}
and \rm{(e)} separately on the events $\{Z_{T-} = 0 \}$ and
$\{Z_{T-} > 0 \}$. Since $\{Z_{T-}=0\}\subset\{\widetilde Z_T=0\}$
and $E(\widetilde Z_T|{\cal F}_{T-})=Z_{T-}$ on $\{T<+\infty\}$,
by putting $\Gamma_0:=\Bigl\{ P(\widetilde{Z}_T > 0 \big| {\cal
F}_{T-}) =0\Bigr\}$, we derive
$$E\left(Z_{T-}I_{\Gamma_0\cap\{T<+\infty\}}\right)=E\left({\widetilde
Z}_{T}I_{\Gamma_0\cap\{T<+\infty\}}\right)=0,$$ and
  $$\begin{array}{lll}
  0=P\left(\{Z_{T-}=0\}\cap\{\widetilde Z_{T}>0\}\cap\{T<+\infty\}\right)\\
  \\
  \hskip 0.35cm=E\left(I_{\{ Z_{T-}=0\}\cap\{T<+\infty\}}P\left(\widetilde Z_T>0|{\cal F}_{T-}\right)\right).\end{array}
  $$
These equalities imply that on $\{T<+\infty\}$, $P-a.s.$, we
have
  \begin{equation}\label{equalitybetweensets}
  \{Z_{T-} = 0 \}=\Gamma_0   \subset \{\widetilde{Z}_T =0\}.
  \end{equation}
  Thus, on the set $\{T<+\infty\}\cap\Gamma_0$, the three probabilities $P$, $Q_T$ and $\widetilde{Q}_T$ coincide, and the equivalence between assertions \rm{(d)}
   and \rm{(e)} is obvious. On the set $\{T<+\infty\ \&\  P[\widetilde{Z}_T >0 | {\cal F}_{T-}] >0\}$, one has $\widetilde{Q}_T \sim Q_T$, and the equivalence between \rm{(d)} and \rm{(e)} is also obvious. This achieves this first step.\\

\noindent{\bf Step 2:} This step proves
\rm{(e)}$\Longleftrightarrow$ \rm{(b)}. Again thanks to
(\ref{equalitybetweensets}), we deduce that on $\{ Z_{T-}=0\}$,
$\widetilde S\equiv S\equiv 0$ and $Q_T$ coincides with $P$ as
well. Hence, the equivalence between assertions \rm{(e)} and
\rm{(b)} is obvious for this case. Thus, it is enough to prove the
equivalence between these assertions on $\{T<+\infty\ \&\
P(\widetilde{Z}_T >0 | {\cal F}_{T-}) >0\}$.\\ Assume that
\rm{(e)} holds. Then, there exists a positive and ${\cal
F}_T$--measurable random variable, $Y$, such that $P-a.s.$ on
$\{T<+\infty\}$, we have
 $$E^{Q_T}(Y| {\cal F}_{T-}) =1,\ \ E^{Q_T}(Y|\xi| | {\cal F}_{T-}) <+\infty,\ \ \&\  \ \ E^{Q_T}(Y\xi I_{\{\widetilde{Z}_T >0\}} | {\cal F}_{T-}) =0 .$$
 Since $Y>0$ on $\{\widetilde Z_T>0\}$, by putting
  $$
  Y_1 := YI_{\{\widetilde{Z}_T >0\}} + I_{\{\widetilde{Z}_T =0\}} \ \ \mbox{and}\ \ \widetilde{Y}_1 := \frac{Y_1}{E[Y_1 | {\cal F}_{T-}]} ,
  $$
  it is easy to check that $Y_1>0$, $\widetilde Y_1>0$,
  \begin{eqnarray*}
    E\left[\widetilde{Y}_1 | {\cal F}_{T-}\right] = 1 \mbox{ and } E\left[\widetilde{Y}_1 \xi I_{\{\widetilde{Z}_T >0\}} | {\cal F}_{T-}\right] = \frac{E\left[Y \xi I_{\{\widetilde{Z}_T >0\}} | {\cal F}_{T-}\right]}{E[Y_1 | {\cal F}_{T-}]} = 0.
  \end{eqnarray*}
   Therefore, $\widetilde{S}$ is  a martingale under $R:=\widetilde Y_1\cdot P\sim P$, and hence $\widetilde S$ satisfies NUPBR$(\mathbb F)$. This proves assertion \rm{(b)}.\\
   To prove the reverse sense, we suppose  that assertion \rm{(d)} holds. Then, there exists $0<Y\in L^0({\cal F}_T)$, such that $E[Y |\xi| I_{\{\widetilde{Z}_T >0\}} | {\cal F}_{T-}] <+\infty$, $E[Y| {\cal F}_{T-}] =1$ and $ E[Y\xi I_{\{\widetilde{Z}_T >0\}} | {\cal F}_{T-}] = 0.$ Then, consider

    $$ Y_2 := \frac{Y I_{\{\widetilde{Z}_T >0\}} P(\widetilde{Z}_T >0 | {\cal F}_{T-})}{E[Y I_{\{\widetilde{Z}_T >0\}} | {\cal F}_{T-}]} >0, \ \ \ \ \ \ \ Q_T-a.s..
  $$
  Then it is easy to verify that $Y_2>0\ \ Q_T-a.s.$,
  $$
    E^{Q_T}\left(Y_2 | {\cal F}_{T-}\right) = 1, \ \ \ \ \ \mbox{and}\ \ \ \ \
    E^{Q_T}\left(Y_2 X| {\cal F}_{T-}\right) = \frac{E\left[Y X I_{\{\widetilde{Z}_T >0\}} | {\cal F}_{T-}\right]}{E[Y I_{\{\widetilde{Z}_T >0\}} | {\cal F}_{T-}]} =0.
  $$
  This proves assertion \rm{(e)}, and the proof of \rm{(e)}$\Longleftrightarrow$\rm{(b)} is achieved.\\

\noindent{\bf Step 3:} Herein, we prove \rm{(a)} $\Rightarrow$
\rm{(d)}. Suppose that $S^{\tau}$ satisfies NUPBR$(\mathbb G)$.
Then there exists a positive ${\cal G}_T$-measurable random
variable $Y^{\mathbb G}$ such that $E[\xi Y^{\mathbb G} I_{\{T\leq
\tau\}} | {\cal G}_{T-}] = 0$ on $\{T<+\infty\}$. Due to Lemma
\ref{lemma:predsetFG}--\rm{(a)}, we deduce the existence of a
positive ${\cal F}_T$-measurable variable $Y^{\mathbb F} $ such
that $Y^{\mathbb G} I_{\{T\leq \tau\}} = Y^{\mathbb F} I_{\{T\leq
\tau\}}$. Then, on $\{T<+\infty\}$ we obtain
  \begin{eqnarray*}
    0&=& E[\xi Y^{\mathbb F} I_{\{T \leq \tau\}} | {\cal G}_{T-}] = E[XY^{\mathbb F} \widetilde{Z}_T | {\cal F}_{T-}]  \frac{I_{\{T\leq \tau\}}}{Z_{T-}}.
  \end{eqnarray*}
  Therefore, by taking conditional expectation in the above equality, we get $$0=E[\xi Y^{\mathbb F} \frac{\widetilde{Z}_T}{ Z_{T-}} I_{\{Z_{T-}>0\}}| {\cal F}_{T-}] = E^{\widetilde{Q}_T}[\xi Y^{\mathbb F}| {\cal F}_{T-}]I_{\{Z_{T-}>0\}}=E^{\widetilde{Q}_T}[S_T Y^{\mathbb F}| {\cal F}_{T-}].$$ This proves that assertion (d) holds and the proof of \rm{(a)}$\Rightarrow $\rm{(d)} is achieved.\\

\noindent{\bf Step 4:} This last step proves \rm{(b)}$\Rightarrow
$\rm{(a)}. Suppose that $\widetilde{S}$ satisfies   NUPBR$(\mathbb
F)$. Then, there exists $Y\in L^1({\cal F}_T)$ such that on
$\{T<+\infty\}$ we have
  \begin{eqnarray*}
    E[Y|{\cal F}_{T-}] = 1, \ \ Y>0,\ \ \  \ E[Y|\xi|I_{\{\widetilde{Z}_T>0\}} | {\cal F}_{T-}] <+\infty,\ \ P-a.s.
  \end{eqnarray*}
  and
  \begin{eqnarray*}
    E[Y \xi I_{\{\widetilde{Z}_T>0\}} | {\cal F}_{T-}] = 0.
  \end{eqnarray*}
  Then by considering $R:=Y\cdot P\sim P$, we get
  $$
  E^R\bigl[{\widetilde S}_T  \big| {\cal F}_{T-}\bigr] =E^R\bigl[\xi I_{\{\widetilde{Z}_T >0\}} \big| {\cal F}_{T-}\bigr]= 0.
  $$
  Therefore, assertion \rm{(a)} follows directly from Proposition \ref{cruciallemma1}  applied to $M=\widetilde S$ under $R\sim P$ (it is easy to see that (\ref{zeroequationbeforetau}) holds for $(\widetilde S, R)$,  i.e.  $E^R({\widetilde S}_T I_{\{ \widetilde Z_T=0\}}|{\cal F}_{T-})=0$). This ends the fourth step and the proof of the theorem is completed.
  \end{proof}


\subsection{Proof of Theorem \ref{main4}} \label{proofmain4}
To highlight the precise difficulty in proving Theorem
\ref{main4}, we remark   that on $\{ T<+\infty\}$,
$$
{{U^{\mathbb G}(T)}\over{E(U^{\mathbb G}(T)\big|\ {\cal
G}_{T-})}}={{1+\Delta{\Lopt}_T-\Delta V^{\mathbb
G}_T}\over{1-\Delta V^{\mathbb G}_T}}\not=1+\Delta{\Lopt}_T=
{{{\cal E}({\Lopt})_T}\over{{\cal E}({\Lopt} )_{T-}}}.$$ where
$U^{\mathbb G}(T)$ is defined in (\ref{QGbeforetau}). This
highlights one of the main difficulties that we will face when we
will formulate the results for possible many predictable jumps
that might not be ordered. Simply, it might not be possible to
piece up
$$U^{\mathbb G}(T_n)=1-{{\Delta m_{T_n}}\over{{\widetilde
Z}_{T_n}}}I_{\{T_n\leq\tau\}},\ \ n\geq 1$$ to form a positive
$\mathbb G$-local martingale density for the process
$(I_{\cup\Rbrack T_n\Lbrack}\centerdot S)^{\tau}$.

\noindent Thus, in virtue of  the above, the key idea behind the
proof of Theorem \ref{main4} lies in connecting the
 NUPBR condition  with the existence of a positive supermartingale (instead) that is a
deflator for the market model under consideration.
\begin{definition}\label{deflator} Consider an $\mathbb H$-semimartingale $X$. Then, $X$ is said to admit an
$\mathbb H$-deflator if there exists a positive $\mathbb
H$-supermartingale $Y$ such that $Y(\theta\is X)$ is a
supermartingale, for any $\theta\in L(X,\mathbb H)$ such that
$\theta\is X\geq -1$.\end{definition}

\noindent  For supermartingale deflators, we reader the reader to
Rokhlin \cite{{rokh}}. Again, the above definition differs from
that of the literature when the horizon is infinite, while it is
the same as the one of the literature when the horizon is finite
(even random). Below, we slightly generalize \cite{rokh} to our
context.

\begin{lemma}\label{DeflatorNUPBR}
Let $X$ be an $\mathbb H$-semimartingale. Then, the following assertions are equivalent.\\
{\rm{(a)}} $X$ admits an $\mathbb H$-deflator.\\
{\rm{(b)}} $X$ satisfies  NUPBR($\mathbb H$).
\end{lemma}

\begin{proof} The proof of this lemma is straightforward, and is omitted.
 \end{proof}

\noindent Now, we start giving the proof of Theorem \ref{main4}.

\begin{proof}{\it of Theorem \ref{main4}}
 The proof of the theorem will given in two steps,   where we prove (b)$\Rightarrow$(a)
  and  the reverse  implication respectively. For the sake of simplifying the overall proof of the theorem, we remark that
\begin{equation}\label{Zzerosetbeforetau}
\{\widetilde{Z}^Q_T=0\} = \{\widetilde{Z}_T = 0\},\ \ \ \mbox{for any}\  Q\sim P\ \ \mbox{and any ${\mathbb F}$-stopping time}\ T,\end{equation}
 where $\widetilde{Z}^Q_t := Q[\tau \geq t| {\cal F}_t]$. This equality follows from \begin{eqnarray*}
    E\left[\widetilde{Z}_T I_{\{\widetilde{Z}^Q_T=0\}}\right] = E\left[ I_{\{\tau \geq T\}} I_{\{\widetilde{Z}^Q_T=0\}}\right] = 0,
  \end{eqnarray*}
(which implies $\{\widetilde{Z}^Q=0\} \subset \{\widetilde{Z} = 0\}$) and the symmetric role of $Q$ and $P$.

\noindent {\bf Step 1: } Here, we prove (b)$\Rightarrow$ (a).
Suppose that assertion (b) holds, and consider a sequence of
$\mathbb F$-stopping times $(\tau_n)_n$ that increases to infinity
such that $Y^{\tau_n}$ is an $\mathbb F$-martingale. Then, setting
$Q_n:=Y_{\tau_n}/Y_0\cdot P$, and using (\ref{Zzerosetbeforetau})
and Proposition \ref{NUPBRLocalization}, we deduce that there is
no loss of generality in assuming $Y\equiv 1$. Condition
(\ref{zeroequationbeforetau}) in Theorem \ref{cruciallemma1} holds
for
   $\Delta S_{T_n} I_{\{\widetilde{Z}_{T_n} >0\}}$ and
 $\Delta S_{T_n} I_{\{\widetilde{Z}_{T_n} >0\}}I_{\Rbrack T_n,+\infty \Rbrack},$. Therefore, using the notation $ V^{\mathbb G}$ and $L$ defined
   in (\ref{V(b)process}) and (\ref{Ntilde}),
   for each $n$, $(1 + \Delta {\Lopt}_{T_n} - \Delta V^{\mathbb G}_{T_n})\Delta S_{T_n} I_{\{ T_n\leq\tau\}}I_{\Rbrack T_n,+\infty\Rbrack}$
  is a $\mathbb G$-martingale. Then, a direct application of Yor's exponential formula, we
  get that, for any $\theta \in L(S^\tau,\mathbb G)$
$$
    {\cal E}\left(I_{\Gamma}\is {\Lopt} - I_{\Gamma}\is V^{\mathbb G}\right) {\cal
    E}
    \left(\theta I_{\Gamma}\is S^\tau\right) = {\cal
    E}\left(X  \right)$$
     where
     $$ X:=I_{\Gamma}\is {\Lopt} - I_{\Gamma}\is V^{\mathbb G} + \sum_{n\geq 1}\theta_{T_n}
     \left(1 + \Delta {\Lopt}_{T_n} - \Delta V^{\mathbb G}_{T_n}\right)\Delta S_{T_n} I_{\{ T_n\leq\tau\}}I_{\Rbrack T_n, +\infty\Rbrack}.
 $$
  Consider now  the  $\mathbb G$-predictable process
  \begin{eqnarray*}
     \phi &= &\displaystyle\sum_{n\geq 1}
     \xi_n I_{\Rbrack T_n\Lbrack\cap \Rbrack 0,\tau\Lbrack} + I_{\Gamma^c\cup\Lbrack\tau
     +\infty\Rbrack},\ \ \ \ \ \ \mbox{where}\\
     \xi_n&:=&\displaystyle\frac{2^{-n} \left(1+{\cal E}(X)_{T_n-}\right)^{-1}}{\left(1 + E\left[|\Delta {\Lopt}_{T_n}|\ \Big|{\cal G}_{T_n-}\right] + \Delta V^{\mathbb G}_{T_n-} + E\left[|\theta_{T_n} \frac{Z_{T_n-}}{\widetilde{Z}_{T_n}}I_{\{T_n\leq \tau\}}\Delta S_{T_n} | \Big| {\cal G}_{T_n-}\right]\right)}.
  \end{eqnarray*}
  Then, it is easy to verify that    $0<\phi\leq 1$ and
  $E\left(|\phi\centerdot {\cal E}(X)|_{var}(+\infty)\right) \leq \sum_{n\geq 1}2^{-n} =1$.
  Hence, $\phi\is {\cal E}(X)\in {\cal A}(\mathbb G)$.
Since, $\Delta {\Lopt}_{T_n} I_{\Rbrack T_n, +\infty\Rbrack}$ and
$(1 + \Delta {\Lopt}_{T_n} - \Delta V^{\mathbb G}_{T_n})\Delta
S_{T_n} I_{\{ T_n\leq\tau\}}I_{\Rbrack T_n,+\infty\Rbrack}$ are
$\mathbb G$-martingales, we derive $$\left(\phi\is {\cal
E}(X)\right)^{p,\mathbb G} =\sum_n \phi_{T_n}{\cal
E}_{T_n-}(X)E(\Delta X_{T_n}|{\cal G}_{T_n-})I_{\Rbrack T_n,
+\infty\Rbrack}= -\phi{\cal E}_{-}(X) \is V^{\mathbb G}\leq 0.$$
  This proves that ${\cal E}(X)$ is a positive $\sigma$-supermartingale\footnote{Recall
  that a process $X$ is said to be a $\sigma$-supermartingale if it is a semimartingale and
  there exists a predictable process $\phi$ such that $0<\phi\leq 1$ and $\phi\is X$ is a
  supermartingale}. Thus, thanks to Kallsen \cite{kallsen04}, we conclude that it is a
  supermartingale and $\left(I_{\{ Z_{-}\geq \delta \}}\is S\right)^{\tau}$  admits a
  $\mathbb G$-deflator. Then, thanks to Lemma \ref{DeflatorNUPBR}, we deduce that $\left(I_{\{ Z_{-}\geq \delta \}}\is S\right)^{\tau}$ satisfies NUPBR$(\mathbb G)$. Remark that, due to the $\mathbb G$-local boundedness of $(Z_{-})^{-1}I_{ \Rbrack 0,\tau\Lbrack}$, there exists a family of $\mathbb G$-stopping times $\tau_{\delta},\ \delta >0$ such that $\tau_{\delta}$ converges almost surely to infinity when $\delta$ goes zero and
  $$
  \Rbrack 0,\tau\wedge\tau_{\delta}\Lbrack\subset\{Z_{-}\geq \delta\}.$$
  This implies that $S^{\tau\wedge\tau_{\delta}}$ satisfies NUPBR$(\mathbb G)$, and the assertion (a) follows from Proposition \ref{NUPBRLocalization} (by taking $Q_n=P$ for all $n\geq 1$). This ends the proof of  (b)$\Rightarrow $(a).  \\

  \noindent {\bf Step 2:}

In this step, we focus on (a)$\Rightarrow $(b).  Suppose that
$S^\tau$ satisfies NUPBR($\mathbb G$). Then, there exists a
$\sigma$-martingale density under $\mathbb G$, for $I_{\{Z_{-}\geq
\delta\}}\is S^{\tau},$ ($\delta>0$), that we denote by $
D^\mathbb G$. Then, from a direct application of Theorem
\ref{representation} and Theorem \ref{theosigmadensityiff},
  we deduce the existence of a positive $\widetilde{\cal P}(\mathbb G)$-measurable functional, $f^{\mathbb G}$, such that $D^\mathbb G :={\cal E}(N^{\mathbb G})>0$, with
$$
  N^{\mathbb G}:=W^{\mathbb G}\star (\mu^{\mathbb G} - \nu^{\mathbb G}),  \ W^{\mathbb G}:= f^{\mathbb G}-1 + \frac{\widehat{f}^{\mathbb G} - a^{\mathbb G}}{1 - a^{\mathbb G}}I_{\{a^{\mathbb G} <1\}},
 $$
where $\nu^{\mathbb G}$ was defined in (\ref{canonicaldecompoS}),
and, introducing $f_m$ defined in (\ref{decompositionofm})
  \begin{eqnarray}\label{mgequationbeforetau}
    xf^{\mathbb G}I_{\{Z_{-}\geq\delta\}}\star \nu^{\mathbb G}  = xf^{\mathbb G}\left(1 + \frac{f_m}{Z_{-}}\right)I_{\Lbrack 0,\tau\Lbrack}I_{\{Z_{-}\geq\delta\}}\star  \nu \equiv 0.
  \end{eqnarray}

  Thanks to Lemma \ref{lemma:predsetFG}, we conclude to the existence of a positive
  $\widetilde{\cal P}(\mathbb F)$-measurable functional, $f$,
  such that $f^{\mathbb G}I_{\Lbrack 0,\tau\Lbrack} = fI_{\Lbrack 0,\tau\Lbrack}$.
  Thus (\ref{mgequationbeforetau}) becomes
 $$ xf\left(1 + \frac{f_m}{Z_{-}}\right)I_{\Lbrack 0,\tau\Lbrack}I_{\{ Z_{-}>0\}}\star  \nu \equiv 0.
$$
 Introduce the following notations
  \begin{eqnarray}\label{eq:cruYzero}
   \mu_0 &:=& I_{\{\widetilde{Z}>0\ \&\ Z_{-} \geq \delta\}} \cdot \mu , \ \ \nu_0 := h_0I_{\{ Z_{-} \geq \delta\}}\cdot \nu, \   h_0:= M^P_{\mu}\left(I_{\{\widetilde{Z}>0\}} | \widetilde{\cal P}\right), \nonumber \\
    g &:=& \frac{f(1 + \frac{f_m}{Z_{-}})}{h_0} I_{\left\{h_0>0\right\}} + I_{\left\{h_0=0\right\}}, \ \ a_0(t):= \nu_0(\{t\}, {\mathbb{R}^d}),
  \end{eqnarray}
  and assume that
  \begin{equation}\label{mainassumtpionbeforetaubis}
  \sqrt{(g-1)^2\star\mu_0}\in {\cal A}^+_{loc}(\mathbb F).\end{equation}
  Then, thanks to Lemma \ref{boundednessofuhat}, we deduce that $W:=(g-1)/(1- a^0 + \widehat{g})\in {\cal G}^1_{loc}(\mu_0,\mathbb F)$, and the local martingales
  \begin{equation}\label{Nzerobeforetau}
  N^0:= \frac{g-1}{1 - a^0 + \widehat{g}}\star(\mu_0 - \nu_0), \ \ Y^0 := {\cal E}(N^0),\end{equation}
  are well defined satisfying $1 + \Delta N^0 > 0$, $[N^0,S]\in{\cal A}(\mathbb F)$,  and on $\{ Z_{-}>0\}$ we have
  \begin{eqnarray*}
    {{^{p,\mathbb F}\left(Y^0\Delta S I_{\{\widetilde{Z}>0\}}\right)}\over{Y^0_{-}}} &=& \ ^{p,\mathbb F}\left((1+\Delta N^0)\Delta S I_{\{\widetilde{Z}>0\}}\right) =\  ^{p,\mathbb F}\left(\frac{g}{1 - a^0 + \widehat{g}}\Delta S I_{\{\widetilde{Z}>0\}}\right)\\
     &=& \Delta \frac{gxh_0}{1 - a^0 + \widehat{g}}\star \nu = \Delta\frac{ xf(1 + f_m/{Z_{-}})}{1 - a^0 + \widehat{g}}\star \nu=Z_{-}^{-1}\frac{ ^{p,\mathbb F}\left(\Delta U\right)}{1 - a^0 + \widehat{g}} \equiv 0.
  \end{eqnarray*}
  This proves that assertion (b) holds under the assumption (\ref{mainassumtpionbeforetaubis}). \\

\noindent   The remaining part of the proof will show that this assumption always holds. To this end, we start by noticing that on the set $ \left\{h_0>0\right\}$,
  \begin{eqnarray*}
    g-1 &=& \frac{f(1 + \frac{f_m}{Z_{-}})}{h_0} - 1 =
     \frac{(f-1)(1 + \frac{f_m}{Z_{-}})}{h_0} + \frac{f_m}{Z_{-}h_0} + \frac{M^P_{\mu}\left(I_{\{\widetilde{Z}=0\}} | \widetilde{\cal P}\right)}{h_0}\nonumber \\
    &:=& g_1 + g_2 + g_3.
  \end{eqnarray*}
  Since $\left((f-1)^2 I_{\Lbrack 0,\tau \Lbrack}\star \mu \right)^{1/2}  \in {\cal A}^+_{loc}(\mathbb{G})$,
  then due to Proposition \ref{prop:alocundergf}--(e)
  \begin{eqnarray*}
    \sqrt{(f-1)^2 I_{\{Z_{-}\geq \delta\}}\star (\widetilde{Z}\cdot \mu)}\in {\cal A}^+_{loc}(\mathbb F),\ \ \ \mbox{for any}\ \ \delta>0.\end{eqnarray*}
    Then, a direct application of Proposition \ref{prop:alocundergf}--(a), for any $\delta>0$, we have
    \begin{eqnarray*} (f-1)^2 I_{\{\vert f-1\vert\leq \alpha\ \&\ Z_{-}\geq \delta\}}\star (\widetilde{Z}\cdot \mu),\ \  |f-1|I_{\{|f-1|> \alpha\ \&\ Z_{-}\geq\delta\}}\star (\widetilde{Z}
    \cdot \mu) \in {\cal A}^+_{loc}(\mathbb F).
  \end{eqnarray*}
  By stopping, without loss of generality, we assume these two processes and $[m,m]$ belong to ${\cal A}^+(\mathbb F)$. Remark that $Z_{-} + f_m = M^P_{\mu}\left(\widetilde{Z}| \widetilde{\cal P}\right) \leq M^P_{\mu}\left(I_{\{\widetilde{Z}>0\}} | \widetilde{\cal P}\right) = h_0$ that follows from $\widetilde{Z} \leq I_{\{\widetilde{Z}>0\}}$. Therefore, we derive
  \begin{eqnarray*}
    E\left[g_1^2I_{\{|f-1|\leq \alpha\}}\star \mu_0(\infty)\right] &=& E\left[\frac{(f-1)^2(1 + \frac{f_m}{Z_{-}})^2}{h_0^2}I_{\{|f-1|\leq \alpha\}}\star \mu_0(\infty)\right] \nonumber \\
    &=& E\left[\frac{(f-1)^2(1 + \frac{f_m}{Z_{-}})^2}{h_0^2}I_{\{|f-1|\leq \alpha\}} \star \nu_0(\infty)\right] \nonumber \\
    &\leq& \delta^{-2}E\left[(f-1)^2( Z_{-} + f_m) I_{\{|f-1|\leq \alpha\ \&\ Z_{-}\geq \delta\}} \star \nu(\infty)\right] \nonumber \\
    &=& \delta^{-2}E\left[(f-1)^2 I_{\{|f-1|\leq \alpha\}} \star (\widetilde{Z}I_{\{Z_{-}\geq \delta \}}
    \cdot \mu)(\infty)\right]<+\infty, \nonumber
  \end{eqnarray*}
 and
  \begin{eqnarray*}
    E\left[g_1I_{\{|f-1|> \alpha\}}\star \mu_0(\infty)\right] &=& E\left[\frac{|f-1|(1 + \frac{f_m}{Z_{-}})}{h_0}I_{\{|f-1|> \alpha\}} \star \mu_0(\infty)\right] \nonumber \\
    &=& E\left[|f-1|(1 + \frac{f_m}{Z_{-}})I_{\{|f-1|> \alpha\}} I_{\{Z_{-}\geq\delta\}}\star \nu_0(\infty)\right] \nonumber \\
    &\leq& \delta^{-1}E\left[|f-1| I_{\{|f-1|> \alpha\}} \star (\widetilde{Z}I_{\{Z_{-}\geq \delta \}}
    \cdot \mu)(\infty)\right]<+\infty.\nonumber
  \end{eqnarray*}
  Here $\mu_0$ and $\nu_0$ are defined in (\ref{eq:cruYzero}). Therefore, again by Proposition \ref{prop:alocundergf}--(a), we conclude that $\sqrt{g_1^2\star\mu_0}\in {\cal A}_{loc}^+(\mathbb F)$.\\

  \noindent Notice that $g_2 + g_3 = \frac{M^P_{\mu}\left(\Delta mI_{\{\widetilde{Z}>0\}} | \widetilde{\cal P}\right)}{Z_{-}h_0}$, and due to Lemma \ref{conditionalexpectationMpmu}, we derive
  \begin{eqnarray*}
  E\left[(g_2+g_3)^2\star \mu_0(\infty)\right] &=&E\left[\frac{M^P_{\mu}\left(\Delta mI_{\{\widetilde{Z}>0\}} | \widetilde{\cal P}\right)^2}{Z_{-}^2h_0^2}\star \mu_0(\infty)\right] \nonumber \\
  &\leq&E\left[\frac{M^P_{\mu}\left((\Delta m)^2 | \widetilde{\cal P}\right) M^P_{\mu}\left(I_{\{\widetilde{Z}>0\}} | \widetilde{\cal P}\right)}{Z_{-}^2h_0^2}\star \mu_0(\infty)\right] \nonumber \\
  &=& E\left[\frac{M^P_{\mu}\left((\Delta m)^2 | \widetilde{\cal P}\right) }{Z_{-}^2}I_{\{Z_{-}\geq \delta \}}\star \mu(\infty)\right]\\
  & \leq& \delta^{-2} E\left[ [m,m]_{\infty}\right]<+\infty. \nonumber
  \end{eqnarray*}
  Hence, we conclude that $\sqrt{(g-1)^2\star \mu_0} \in {\cal A}^+_{loc}(\mathbb{F}).$ This ends the proof of (\ref{mainassumtpionbeforetaubis}), and the proof of the theorem is completed. \end{proof}
\vspace*{1cm}

\centerline{\textbf{APPENDIX}}
\appendix
 \normalsize

\section{Representation of  Local
Martingales}\label{Representationlocalmart}
This section  recall an important result on representation of
local martingales. This result relies on the continuous local
martingale part and the jump random measure of a given
semimartingale. Thus, throughout this section, we suppose given a
$d$-dimensional semimartingale, $S=(S_t)_{0\leq t\leq T}$. To this
semimartingale, we associate its { predictable characteristics}
that we will present below (for more details about these and other
related issues, we refer the reader to Section II.2 of \cite{JS}).
The random measure $\mu$ associated to the jumps of $S$ is defined
by
 $$
 \mu(d t,\ d x)=\sum I_{\{\Delta S_s\not=0\}}\delta_{(s,\ \Delta S_s)}(d t,\ d x),
 $$
 with $\delta_a$ the Dirac measure at point $a$. The continuous local martingale part of $S$ is denoted by $S^{c}$.
 This leads to the following decomposition, called ``{\it the
 canonical representation}'' (see Theorem 2.34, Section II.2 of \cite{JS}), namely,
\begin{equation}\label{modelS}
 S=S_0+S^{c}+h(x)\star (\mu-\nu)+ (x-h(x))\star\mu+B,
 \end{equation}
 where the
random measure $\nu$ is the compensator of the random measure
$\mu$, the function $h(x)$ is the truncation function given
by $h(x)=xI_{\{\vert x\vert \leq 1\}}$, and $B$ is a
predictable  process with bounded variation. For the matrix $C$
with entries $C^{ij}:=\langle S^{c,i}, S^{c,j}\rangle $, the
triple $(B,\ C,\ \nu)$ is called {\it predictable characteristics}
of $S$.
  Furthermore, we can find a version of the characteristics triple satisfying
\begin{equation}\label{modelSbis} B=b\is A,\ \ C=c\is A\ \ \mbox{and}\ \
\nu(\omega,\ d t,\ d x)=d A_t(\omega)F_t(\omega,\ d x).
\end{equation} Here $A$ is an increasing and predictable process
which is continuous if and only if $S$ is quasi-left continuous, $b$
and $c$ are predictable processes,
 $F_t(\omega,\ d x)$ is a predictable kernel, $b_t(\omega)$ is a vector in $\hbox{I\kern-.18em\hbox{R}}^d$ and
$c_t(\omega)$ is a symmetric $d\times d$-matrix , for all $(\omega,\
t)\in\Omega\times [0,\ T]$. In the sequel we will often drop
$\omega$ and $t$ and write, for instance, $F(d x)$ as a shorthand
for $F_t(\omega , d x)$.

\noindent The characteristics, $B,\ C,$ and $\nu$,  satisfy
$$
\begin{array}{l}
F_t(\omega,\ \{0\})=0,\hskip 1cm \displaystyle{\int} (\vert
x\vert^2\wedge 1)F_t(\omega,\
d x)\leq 1, \\
\\
 \Delta B_t=b\Delta A=\displaystyle{\int} h(x)\nu(\{t\}, d x),\hskip 1cm \mbox{and} \hskip 1cm  c=0\ \ \mbox{ on }\ \{\Delta A\neq 0\} .
 \end{array}
$$
We set
$$
\nu_t(d x):=\nu (\{t\}, d x),\ \  a_t:=\nu_t
(\hbox{I\kern-.18em\hbox{R}}^d)=\Delta A_t
F_t(\hbox{I\kern-.18em\hbox{R}}^d)\leq 1.
$$

\noindent For the  following representation theorem, we refer to
\cite[Theorem 3.75, page 103]{Jacod} and to \cite[Lemma 4.24, Chap
III]{JS}.
\begin{theorem}\label{representation}
Let $N\in {\cal M}_{0,loc}$. Then, there exist a predictable $S^c$-integrable process $\beta$, $N^{\perp}\in {\cal M}_{0,loc}$ with
$N^{\perp}$ and $S$ orthogonal and functionals $f\in {\widetilde{{\cal P}}}$ and $g\in
{\widetilde{{\cal O}}}$ such that\\
{\rm{(a)}}\hskip 0.25cm $
 \Bigl (\sum_{s\leq t}  f_s(  \Delta S_s )^2
I_{\{\Delta S_s\not = 0\}}\Bigr )^{1/2}$ and $\Bigl (\sum_{s\leq
 t}  g_s(  \Delta S_s )^2
I_{\{\Delta S_s\not = 0\}}\Bigr )^{1/2}$ belong to ${\cal A}^+_{loc}$.\\
{\rm{(b)}}\hskip 0.25cm   $
M^P_{\mu}(g\ |\ {\widetilde {{\cal P}}})=0,\ \ \ M^P_{\mu}-a.e.$, where $M^P_{\mu}:=P\otimes\mu$.\\
{\rm{(c)}}\hskip 0.2cm The process $N$ satisfies \begin{equation}
 N=\beta\is
S^c+W\star(\mu-\nu)+g\star\mu+N^{\perp},\quad \mbox{where~}
W=f+\frac{{\widehat f}}{1-a}I_{\{a<1\}}.\label{Ndecomposition}
\end{equation} Here $\widehat f_t=\int f_t(x)\nu(\{t\},d x)$ and $f$ has a version
such that $\{a=1\}\subset \{\widehat f=0\}$.\\ Moreover
 \begin{equation}
 \label{jumps}
 \Delta N_t= \Bigl(f_t(\Delta S_t)+g_t(\Delta S_t)\Bigr)I_{\{\Delta
 S_t\not=0\}}-{{\widehat
f_t}\over{1-a_t}}I_{\{\Delta
 S_t=0\}}+\Delta N_t^{\perp}.
 \end{equation}
 The quadruplet $\left(\beta, f, g, N^{\perp}\right)$ are called the Jacod's parameters of the local martingale $N$ with respect to $S$.
\end{theorem}

\noindent The following is a simple but useful result on the conditional expectation with respect to $M^P_{\mu}$.

\begin{lemma}\label{conditionalexpectationMpmu}
Consider a filtration $\mathbb H$ satisfying the usual conditions. Let $f$ and $g$ two nonnegative $\widetilde{\cal O}(\mathbb H)$-measurable functionals. Then we have
\begin{equation}\label{CauchySchwarz}
M_{\mu}^P\left(fg\big|\ \widetilde{\cal P}\right)^2\leq M_{\mu}^P\left(f^2\big|\ \widetilde{\cal P}\right)M_{\mu}^P\left(g^2\big|\ \widetilde{\cal P}\right),\ \ \ \ M_{\mu}^P\mbox{--a.e.}
\end{equation}
\end{lemma}
\begin{proof}
The proof is the same as the one of the regular Cauchy-Schwarz formula, by putting $\bar f:=f/M_{\mu}^P\left(f^2\big|\ \widetilde{\cal P}\right)$ and $\bar g:=g/M_{\mu}^P\left(g^2\big|\ \widetilde{\cal P}\right)$ and using the simple inequality $xy\leq (x^2+y^2)/2$.
This ends the proof of the lemma.
\end{proof}

\noindent The following lemma is borrowed from Jacod's Theorem 3.75 in \cite{Jacod} (see also Proposition 2.2 in \cite{Choulli2012}).

 \begin{lemma} Let ${\cal E}(N)$ be a positive local martingale and $\left(\beta, f, g, N'\right)$ be the Jacod's parameters of $N$. Then ${\cal E}(N)>0$ (or equivalently $1+\Delta N>0$) implies that
 $$
 f>0,\ \ \ M^P_{\mu}-a.e.$$
 \end{lemma}


\begin{theorem}\label{theosigmadensityiff}
  Let $S$ be a semi-martingale with predictable characteristic triplet $(b,c,\nu=A\otimes F)$, $N$ be a local martingale such that ${\cal E}(N)>0$, and $(\beta,f,g,N')$ be its Jacod's parameters.  Then the following assertions hold.\\
  1) ${\cal E}(N)$ is a $\sigma$-martingale density of $S$  if and only if the following two properties hold:
  \begin{equation}\label{integrabilitycondition}
  \displaystyle\int \vert x - h(x) + xf(x)\vert F(dx)<+\infty,\ \ \ P\otimes A-a.e.\end{equation}
  and
  \begin{equation}\label{martingalerequation}
  b + c\beta + \displaystyle\int\Bigl(x - h(x) + x f(x)\Bigr)F(dx) = 0,\ \ \ \ \ P\otimes A-a.e.\end{equation}
2) In particular, we have
 \begin{equation}\label{martingalerequationJumps}
 \int x(1+f_t(x))\nu(\{t\},dx) = \int x(1+f_t(x))F_t(dx)\Delta A_t=0,\ \ \ \ \ P-a.e.\end{equation}
 \end{theorem}
\begin{proof}
  The proof can be found in Choulli et al. \cite[Lemma 2.4]{choullistricker07}, and also Choulli and Schweizer \cite{Choulli2012}.
\end{proof}

\begin{lemma}\label{boundednessofuhat} {\bf (see Choulli and Schweizer \cite{Choulli2012}):} Consider a filtration $\mathbb H$ satisfying the usual conditions. Let $f$ be a $\widetilde{\cal P}(\mathbb H)$-measurable functional such that $f>0$  and
\begin{equation}\label{(f-1)G1loc}\Bigl[(f-1)^2\star\mu\Bigr]^{1/2}\in {\cal A}^+_{loc}(\mathbb  H).\end{equation}
Then, the $\mathbb H$-predictable process $\left(1-a^{\mathbb H}+{\widehat f}^{\mathbb H}\right)^{-1}$ is locally bounded, and hence
\begin{equation}\label{WinG1loc}
W_t(x):={{f_t(x)-1}\over{1-a^{\mathbb H}_t+{\widehat f}^{\mathbb H}_t}}\in {\cal G}^1_{loc}(\mu,\mathbb H).\end{equation}
Here, $a^{\mathbb H}_t:=\nu^{\mathbb H}(\{t\},\mathbb R^d)$, ${\widehat f}^{\mathbb H}_t:=\int f_t(x)\nu^{\mathbb H}(\{t\},dx)$ and $\nu^{\mathbb H}$ is the $\mathbb H$-predictable random measure compensator of $\mu$ under $\mathbb H$.
\end{lemma}


\section{Proof of  $K\in {^{o}L^1_{loc}(\widehat m, \mathbb G)}$}\label{proofk}
  We
start by calculating on $\Lbrack 0, \tau \Lbrack$, making use of
Lemma \ref{lemmecrucial}.
We recall that $\kappa:= Z_{-}^2+ \Delta
\langle m \rangle^{\mathbb F}$. \\

\begin{equation}\label{differencede
Sauts}\begin{array}{llll} K\Delta{\widehat m}-{^{p,\mathbb
G\!}\left(K\Delta{{\widehat m}}\right)}=\displaystyle\frac{I_{\Lbrack 0, \tau \Lbrack} Z_{-}^2
\Delta \widehat{m}}{\kappa  {\widetilde Z}}
- \, ^{p, \mathbb G \!}\left({\frac{I_{\Lbrack 0, \tau \Lbrack}\,Z_{-}^2 }{\kappa  {\widetilde Z}}} \,\Delta \widehat{m}\right) \\
\\
= \displaystyle \frac{(Z_{-}^2\Delta m-Z_{-}\Delta \langle m \rangle^{\mathbb F}) }{\kappa \, \widetilde Z}
+\frac{^{p,\mathbb F\!}(I_{\{ \widetilde Z>0\}}\,\Delta \langle m \rangle^{\mathbb F})}{\kappa} -  \frac{^{p,\mathbb F\!}(\Delta m I_{\{ \widetilde Z>0\}} )    \,Z_{-} } {\kappa }\\
 \ = \displaystyle \frac{\Delta m}{\widetilde{Z}} I_{\Lbrack 0,\tau\Lbrack} - \ ^{p,\mathbb{F}} \left(I_{\{ \widetilde{Z} = 0\}}  \right)I_{\Lbrack 0,\tau\Lbrack} =: \Delta V - \Delta V^{\mathbb G}.
\end{array}\end{equation}
 Here, $V^{\mathbb G}$, defined in (\ref{V(b)process}), is nondecreasing, c\`adl\`ag and $\mathbb{G}$-locally bounded (see Proposition \ref{lem:vbfinite}). Hence, we immediately deduce that $\sum (\Delta V^{\mathbb G})^2=\Delta V^{\mathbb G}\is V^{\mathbb G}$ is locally bounded, and in the rest of this part we focus on proving $\sqrt{\sum (\Delta V)^2}\in {\cal A}^+_{loc}(\mathbb G)$. To this end, we consider $\delta \in (0,1)$, and define $C:=\{\Delta m < -\delta Z_{-}\}$ and $C^c$ its complement in $\Omega\otimes[0,+\infty[$. Then we obtain
 \begin{eqnarray*}\label{eq:crulemade11-1}
 \sqrt{\sum (\Delta V)^2}&\leq& \displaystyle\left(\sum{\frac{(\Delta m)^2}{\widetilde Z^2} I_{C}}I_{ \Lbrack 0, \tau\Lbrack}\right)^{1/2} + \left(\sum{\frac{(\Delta m)^2}{\widetilde Z^2}I_{C^c}I_{ \Lbrack 0, \tau\Lbrack}}\right)^{1/2}  \nonumber \\
 &\leq & \sum{\frac{\vert\Delta m\vert}{\widetilde{Z}} I_{ C }}I_{\Lbrack 0, \tau\Lbrack}+ {1\over{1-\delta}}\left(I_{\Lbrack 0, \tau\Lbrack}\frac{1}{ Z_{-}^{2}}\is [m]\right)^{1/2}    =:   V_1 + V_2.
 \end{eqnarray*}
 The last inequality above is due to $\sqrt{\sum (\Delta X)^2}\leq \sum \vert\Delta X\vert$ and $\widetilde Z \geq Z_{-}(1-\delta)$ on $C^c$. Using the fact that  $(Z_{-})^{-1}I_{\Lbrack 0,\tau \Lbrack}$ is $\mathbb G$-locally bounded   and  that $m$ is an $\mathbb F$-locally bounded martingale, it follows that  $V_2$ is $\mathbb G$-locally bounded. Hence, we focus on proving the $\mathbb G$-local integrability of $V_1$.\\

\noindent  Consider a sequence of $\mathbb G$-stopping times $(\vartheta_n)_n$ that increases to $+\infty$ and
 $$
 \Bigl((Z_{-})^{-1}I_{\Lbrack 0, \tau \Lbrack}\Bigr)^{\vartheta_n}\leq n.$$

\noindent Also consider an $\mathbb F$-localizing sequence of stopping times, $(\tau_n)_n $, for the process $V_3:=\sum {{(\Delta m)^2}\over{1+\vert\Delta m\vert}}$. Then, it is easy to prove $$U_n:=\sum \vert\Delta m\vert I_{\{\Delta m <-\delta/n\}}\leq {{n+\delta}\over{\delta}} V_3,$$
and conclude that $\left(U_n\right)^{\tau_n}\in {\cal A}^+(\mathbb F)$. Therefore, due to
$$\begin{array}{lll}
C\cap{\Lbrack 0, \tau\Lbrack}\cap{\Rbrack 0,\vartheta_n\Lbrack} =\{ \Delta m < -\delta Z_{-}\}\cap{\Lbrack 0,\vartheta_n\Lbrack}\cap{\Lbrack 0, \tau \Lbrack}\\
\\
\hskip 3cm \subset{\Lbrack 0, \tau \Lbrack}\cap{\Lbrack 0,\vartheta_n\Lbrack}\cap\{\Delta m < -{{\delta}\over{n}}\},
\end{array}$$ we derive
  $$
 ( V_1)^{\vartheta_n\wedge\tau_n}\leq \left(\widetilde{Z}\right)^{-1} I_{\Lbrack 0, \tau \Lbrack}\is (U_n)^{\tau_n}.$$
  Since $(U_n)^{\tau_n}$ is $\mathbb F$-adapted, nondecreasing and integrable, then due to Lemma \ref{lem:beforetauforoptional1}, we deduce that the process $V_1^{\vartheta_n\wedge\tau_n}$ is nondecreasing, $\mathbb G$-adapted and integrable. Since $\vartheta_n\wedge\tau_n$ increases to $+\infty$, we conclude that the process $V_1$ is $\mathbb G$-locally integrable. This completes the proof of $K\in{^{o}L^1_{loc}({\widehat m}, \mathbb G)}$, and the process $\Lopt$ (given via (\ref{Ntilde}) and Definition \ref{SIOptional}) is a $\mathbb G$-local martingale.\\
\section{$\mathbb G$-Localization versus $\mathbb F$-Localization}
\begin{lemma}\label{lemma:predsetFG}
Let $H^{\mathbb G}$ be a $\widetilde{\cal P}(\mathbb G)$-measurable functional. The the following hold.\\
(a) There exist an $\widetilde{\cal
P}(\mathbb F)$-measurable functional $H^{\mathbb F}$ and a ${\cal
B}(\mathbb R_+)\otimes \widetilde{{\cal P}}(\mathbb F)$-measurable functionals
$K^{\mathbb F}: {\mathbb R}_+\times{\mathbb R}_+ \times \Omega
\times {\mathbb R}^d \rightarrow \mathbb R$ such that
\begin{eqnarray}\label{eq:widePGandwidePF}
H^{\mathbb G}(\omega,t,x) = H^{\mathbb F}(\omega,t,x) I_{\Lbrack 0,\tau\Lbrack}+ K^{\mathbb F}(\tau(\omega),t,\omega,x)I_{\Lbrack \tau,+\infty\Lbrack}.
\end{eqnarray}
(b) If furthermore $H^{\mathbb G}>0$ (respectively $H^{\mathbb G}\leq 1$), then we can choose $H^{\mathbb F}>0$ (respectively $H^{\mathbb F}\leq 1$) such that $$H^{\mathbb G}(\omega,t,x)I_{\Lbrack 0,\tau\Lbrack} = H^{\mathbb F}(\omega,t,x) I_{\Lbrack 0,\tau\Lbrack}.$$
\end{lemma}

\begin{proof}
  The proof of assertion (a) mimics exactly the approach of Jeulin\cite{Jeu}, and will be omitted.\\
  To prove positivity of $H^{\mathbb F}$ when $H^{\mathbb G}>0$ holds, we consider
   $$
  {\overline H}^{\mathbb F}:=(H^{\mathbb F})^++I_{\{ H^{\mathbb F}=0\}}>0,$$
   and we remark that due to (\ref{eq:widePGandwidePF}), we have $\Lbrack 0,\tau\Lbrack\subset\{  H^{\mathbb G}=H^{\mathbb F}\}\subset\{H^{\mathbb F}>0\}$. Thus, we get
   $$
   H^{\mathbb G}I_{\Lbrack 0,\tau\Lbrack}=\overline{H}^{\mathbb F}I_{\Lbrack 0,\tau\Lbrack}.$$
   Similarly, we consider $H^{\mathbb F}\wedge 1$, and we deduce that if $H^{\mathbb G}$ is upper-bounded by one, the process $H^{\mathbb F}$ can also be chosen to not exceed one. This ends the proof of the proposition.
\end{proof}

\noindent In the following, we state and prove our main results of this subsection.
\begin{proposition}\label{prop:alocundergf}
  For any $\alpha>0$, the following assertions hold:\\
  {\rm{(a)}} Let $h$ be a ${\widetilde{\cal P}(\mathbb H)}$-measurable functional. Then, $\sqrt{(h-1)^2\star \mu} \in {\cal A}^+_{loc}(\mathbb H)$ iff
  \begin{eqnarray*}
    (h-1)^2I_{\{|h-1|\leq \alpha\}}\star \mu \mbox{ and } \ |h-1|I_{\{|h-1|>\alpha\}}\star \mu\ \ \mbox{belong to}\ \ {\cal A}^+_{loc}(\mathbb H).
  \end{eqnarray*}
 {\rm{(b)}} Let $(\sigma^{\mathbb G}_n)_n$ be a sequence of $\mathbb G$-stopping times that increases to infinity.
 Then, there exists a nondecreasing sequence of $\mathbb F$-stopping times, $(\sigma^{\mathbb F}_n)_{n\geq 1}$,
 satisfying the following properties
  \begin{eqnarray}
  \sigma^{\mathbb G}_n\wedge\tau= \sigma^{\mathbb F}_n\wedge\tau,\ \ \ \sigma_{\infty}:=\sup_{n} \sigma^{\mathbb F}_n\geq\widehat R\ \ P-a.s.,\label{G/Fstoppingbeforetau1} \\
  \ \ \mbox{and }\ \ \ \ \ \ \ \ Z_{\sigma_{\infty}-}=0\ \ \ P-a.s.\ \ \ \ \mbox{on}\ \ \ \ \Sigma\cap(\sigma_{\infty}<+\infty),\label{G/Fstoppingbeforetau2}
  \end{eqnarray}
  where $\Sigma:=\displaystyle\bigcap_{n\geq 1}(\sigma_n^{\mathbb F}<\sigma_{\infty})$.\\
   {\rm{(c)}}Let $V$ be an $\mathbb F$-predictable and non-decreasing process. Then, $V^{\tau}\in{\cal A}_{loc}^+(\mathbb G)$ if and only if $I_{\{ Z_{-}\geq\delta\}}\is V\in {\cal A}_{loc}^+(\mathbb F)$ for any $\delta>0$.\\
 {\rm{(d)}} Let $h$ be a nonnegative and $\widetilde{\cal P}(\mathbb F)$-measurable functional. Then, $hI_{\Lbrack 0,\tau \Lbrack}\star \mu \in {\cal A}^+_{loc}(\mathbb G)$ if and only if for all $\delta >0$,
  $hI_{\{Z_{-}\geq \delta \}}\star \mu^1 \in {\cal A}^+_{loc}(\mathbb  F)$, where $\mu^1:=\widetilde{Z}\\centerdot\mu.$\\
  {\rm{(e)}} Let $f$ be positive and $\widetilde{\cal P}(\mathbb F)$-measurable, and $\mu^1:=\widetilde{Z}\centerdot\mu.$ Then $\sqrt{(f-1)^2I_{\Lbrack 0,\tau \Lbrack}\star \mu} \in {\cal A}^+_{loc}(\mathbb G)$ iff $\sqrt{(f-1)^2I_{\{Z_{-}\geq \delta\}}\star \mu^1} \in {\cal A}^+_{loc}(\mathbb F)$, for all $\delta >0.$
\end{proposition}

\begin{proof}
  (a) Put $W:=(h-1)^2\star \mu = W_1 + W_2,$ where $\ W_1 := (h-1)^2I_{\{|h-1|\leq \alpha\}}\star \mu, \ W_2:= (h-1)^2I_{\{|h-1|> \alpha\}}\star \mu$ and $W_2' := |h-1|I_{\{|h-1|>\alpha\}}\star \mu$. Note that
  \begin{eqnarray*}
    \sqrt{W} \leq \sqrt{W_1} + \sqrt{W_2} \leq \sqrt{W_1} + W_2'.
  \end{eqnarray*}
  Therefore $\sqrt{W_1}, W_2' \in {\cal A}^+_{loc}$ imply $\sqrt{W}$ is locally integrable. \\
  Conversely, if $\sqrt{W}\in {\cal A}^+_{loc}$, $\sqrt{W_1} $ and $\sqrt{W_2}$ are both locally integrable. Since $W_1$ is locally bounded and has finite variation, $W_1$ is locally integrable.
   In the following, we focus on the proof of the local integrability of $W_2'.$ Denote
  $$
  \tau_n := \inf\{t\geq 0: \ V_t >n\}, \ \ V:=W_2.
  $$
  It is easy to see that $\tau_n$ increases to infinity and $V_{-}\leq n$ on the set $\Lbrack 0,\tau_n \Lbrack$. On the set $\{\Delta V >0\}$, we have $\Delta V \geq \alpha^2.$ By using the elementary inequality
  $\sqrt{1 + \frac{n}{\alpha^2}} - \sqrt{\frac{n}{\alpha^2}} \leq \sqrt{1 + x} - \sqrt{x} \leq 1$, when $0\leq x \leq \frac{n}{\alpha^2}$, we have
  \begin{eqnarray*}
    \sqrt{V_{-} + \Delta V} - \sqrt{V_{-}} \geq \beta_n \sqrt{\Delta V} \ \ \mbox{on}\ \Lbrack 0,\tau_n\Lbrack,\ \mbox{where } \   \beta_n:= \sqrt{1 + \frac{n}{\alpha^2}} - \sqrt{\frac{n}{\alpha^2}},
  \end{eqnarray*}
  and
  \begin{eqnarray*}
    \left(W_2'\right)^{\tau_n} &=& \left(\sum \sqrt{\Delta V}\right)^{\tau_n}
    \leq  \frac{1}{\beta_n}\left(\sum \Delta\sqrt{ V}\right)^{\tau_n}
    =  \frac{1}{\beta_n}\left(\sqrt{W_2}\right)^{\tau_n}\in {\cal A}^+_{loc}(\mathbb H)
  \end{eqnarray*}
  Therefore $W_2' \in ({\cal A}^+_{loc}(\mathbb H))_{loc}={\cal A}^+_{loc}(\mathbb H).$\\

 \noindent (b) Due to Jeulin \cite{Jeu}, there exists a sequence of $\mathbb{F}$-stopping times $(\sigma_n^{\mathbb{F} })_n$ such that
\begin{eqnarray}\label{eq:gfstopinproof1}
  \sigma_n^{\mathbb{G}}\wedge \tau =  \sigma_n^{\mathbb{F}}\wedge \tau.
\end{eqnarray}
By putting $\sigma_n:= \sup_{k\leq n} \sigma_k^{\mathbb{F}},$ we shall prove that
\begin{eqnarray}\label{increasingfact}
  \sigma_n^{\mathbb{G}}\wedge \tau =  \sigma_n\wedge \tau,
\end{eqnarray}
or equivalently $\{\sigma_n^{\mathbb{F}}\wedge \tau <  \sigma_n\wedge \tau\}$ is negligible. Due to (\ref{eq:gfstopinproof1}) and $\sigma_n^{\mathbb G}$ is nondecreasing, we derive
\begin{eqnarray*}
  \{\sigma_n^{\mathbb{F}} < \tau \}= \{\sigma_n^{\mathbb{G}} < \tau \}  \subset\bigcap_{i=1}^n  \{\sigma_i^{\mathbb{G}} =\sigma_i^{\mathbb{F}}\} \subset \{\sigma_n^{\mathbb{F}}= \sigma_n\}.
\end{eqnarray*}
This implies that,
\begin{eqnarray*}
  \{\sigma_n^{\mathbb{F}}\wedge \tau <  \sigma_n\wedge \tau\} = \{ \sigma_n^{\mathbb{F}} < \tau, \ \& \ \sigma_n^{\mathbb{F}}< \sigma_n\} =\emptyset,
\end{eqnarray*}
and the proof of (\ref{increasingfact}) is completed. Without loss of generality we assume that the sequence $\sigma_n^{\mathbb F}$ is nondecreasing. By taking limit in (\ref{eq:gfstopinproof1}), we obtain $\tau = \sigma_\infty \wedge \tau, P-$a.s. which is equivalent to $\sigma_\infty \geq \tau, P-$a.s. Since $\widehat{R}$ is the smallest $\mathbb{F}$-stopping time greater or equal than $\tau$ almost surely, we obtain, $\sigma_\infty \geq \widehat{R} \geq \tau$ $P-a.s.$. This achieves the proof of (\ref{G/Fstoppingbeforetau1}). \\
On the set $\Sigma$, it is easy to show that
\begin{eqnarray*}
  I_{{\Rbrack 0,\sigma^{\mathbb{F}}_n \Lbrack} } \longrightarrow I_{{\Rbrack 0,\sigma^{\mathbb{F}}_\infty \Rbrack} }, \ \mbox{ when $n$ goes to } +\infty.
\end{eqnarray*}
Then, thanks again to (\ref{eq:gfstopinproof1}) (by taking $\mathbb F$-predictable projection and let $n$ go to infinity afterwards), we obtain
\begin{eqnarray}\label{eq:gfstopinproof2}
  Z_{-} = Z_{-}I_{{\Rbrack 0,\sigma_\infty^{\mathbb{F}} \Rbrack} },\ \ \ \ \ \ \mbox{on}\ \ \Sigma.
\end{eqnarray}
Hence, (\ref{G/Fstoppingbeforetau2}) follows immediately, and the proof of assertion (b) is completed.\\

\noindent (c)  Suppose that $hI_{{\Rbrack 0,\tau\Lbrack} }\star
\mu\in {\mathcal{A}^+_{loc}}(\mathbb G)$. Then, there exists a
sequence of $\mathbb{G}$-stopping times $(\sigma^\mathbb{G}_n)$
increasing to infinity such that $hI_{{\Rbrack 0,\tau\Lbrack}
}\star \mu^{\sigma^\mathbb{G}_n} $ is integrable. Consider
$(\sigma_n)$ a sequence of $\mathbb{F}$-stopping times satisfying
(\ref{G/Fstoppingbeforetau1})--(\ref{G/Fstoppingbeforetau2}) (its
existence is guaranteed by assertion (b)). Therefore, for any
fixed $\delta>0$
\begin{eqnarray}\label{Wn}
  W^n:=M_\mu^P\left(\widetilde{Z} | \widetilde{\mathcal{P}}\right) I_{\{Z_{-} \geq \delta\} }h\star \nu^{\sigma_n} \in \mathcal{A}^+(\mathbb{F}),
\end{eqnarray}
or equivalently, this process is c\`adl\`ag predictable with finite values. Thus, it is obvious that the proof of assertion (iii) will follow immediately if we prove that the $\mathbb F$-predictable and nondecreasing process
\begin{eqnarray}\label{eq:muhnufinite1}
  W:=M_\mu^P\left(\widetilde{Z} | \widetilde{\mathcal{P}}\right) I_{\{Z_{-} \geq \delta\} }h\star \nu\ \ \mbox{ is c\`{a}dl\`{a}g with finite values}.
\end{eqnarray}
To prove this last fact, we consider the random time $\tau^\delta$ defined by
$$
  \tau^\delta:= \sup \{t\geq 0 : Z_{t-} \geq \delta\}.
$$
Then, it is clear that $I_{\Lbrack\tau^{\delta},+\infty\Rbrack}\is W\equiv 0$ and
$$\tau^\delta\leq \widehat{R} \leq \sigma_\infty\ \ \ \mbox{and}\ \ Z_{\tau^\delta-} \geq \delta\ \ \ \ P\mbox{--}a.s.\ \ \ \ \ \ \mbox{on}\ \{\tau^\delta<+\infty\}.$$
The proof of (\ref{eq:muhnufinite1}) will be achieved by considering
three sets, namely $\{\sigma_\infty = \infty\}$,  $\Sigma\cap \{\sigma_\infty<+\infty\}$,
and  $\Sigma^c\cap \{\sigma_\infty<+\infty\}$. It is obvious that (\ref{eq:muhnufinite1})
holds on $\{\sigma_\infty = \infty\}$. Due to (\ref{G/Fstoppingbeforetau2}), we deduce
that $\tau^\delta< \sigma_\infty, P-$a.s. on $\Sigma\cap \{\sigma_\infty<+\infty\}$.
 Since $W$ is supported on $\Rbrack 0, \tau^\delta\Lbrack$, then (\ref{eq:muhnufinite1})
 follows immediately on the set $\Sigma\cap \{\sigma_\infty<+\infty\}$. Finally, on the set
$$
\Sigma^c \cap \{\sigma_\infty<+\infty\} = \left(\bigcup_{n\geq 1} \{\sigma_n = \sigma_\infty\}\right) \cap \{\sigma_\infty<+\infty\},
$$
the sequence $\sigma_n$ increases stationarily to $\sigma_\infty$, and thus (\ref{eq:muhnufinite1}) holds on this set. This completes the proof of (\ref{eq:muhnufinite1}), and hence  $h I_{\{Z_{-} \geq \delta\} }\star ({\widetilde Z}\centerdot \mu)$ is locally integrable, for any $\delta>0$.\\

\noindent Conversely, if $hI_{\{Z_{-}\geq \delta\}}{\widetilde Z} \star \mu\in {\cal A}^+_{loc}(\mathbb F) $, there exists a sequence of $\mathbb F$-stopping times $(\tau_n)_{n\geq 1}$ that increases to infinity and $\left(hI_{\{Z_{-}\geq \delta\}}{\widetilde Z}  \star \mu\right)^{\tau_n}\in {\cal A}^+(\mathbb F) $. Then, we have
  \begin{eqnarray}\label{eq:stGF2}
     E\left[hI_{\{Z_{-}\geq \delta\}} I_{\Rbrack 0,\tau\Lbrack}\star \mu(\tau_n)\right]=E\left[hI_{\{Z_{-}\geq \delta\}} \widetilde{Z}\star \mu(\tau_n)\right]<+\infty .
  \end{eqnarray}
  This proves that $hI_{\{Z_{-}\geq \delta\}} I_{\Rbrack 0,\tau\Lbrack}\star \mu$ is $\mathbb G$-locally integrable, for any $\delta>0$. Since $(Z_{-})^{-1}I_{\Rbrack 0,\tau\Lbrack}$ is $\mathbb G$-locally bounded, then there exists a family of $\mathbb G$-stopping times $(\tau_{\delta})_{\delta>0}$ that increases to infinity when $\delta$ decreases to zero, and
  $$
  \Rbrack 0,\tau\wedge\tau_{\delta}\Lbrack\subset\{Z_{-}\geq \delta\}.$$
  This implies that the process $\left(hI_{\Rbrack 0,\tau\Lbrack}\star \mu\right)^{\tau_{\delta}}$ is $\mathbb G$-locally integrable,
  and hence the assertion (c) follows immediately.\\

  \noindent (d) The proof of assertion (d) follows from combining assertions (a)  and (b). This ends the proof of the proposition.
\end{proof}
\bigskip 
{\textbf{Acknowledgements:}} The research of Tahir Choulli  and Jun Deng is supported financially by the
Natural Sciences and Engineering Research Council of Canada,
through Grant G121210818. The research of Anna Aksamit and Monique Jeanblanc is supported
by Chaire Markets in transition, French Banking Federation.





\begin{thebibliography}{1}

\bibitem{afk} Acciaio B., Fontana C.,  Kardaras C.(2014)  Arbitrage of the First Kind and Filtration Enlargements
In Semimartingale Financial Models, http://arxiv.org/pdf/1401.7198.pdf
\bibitem{aksamit/choulli/deng/jeanblanc} Aksamit, A., Choulli, T., Deng, J., and Jeanblanc, M.:\newblock{Arbitrages in a progressive enlargement setting },  Arbitrage, Credit and Informational Risks,  Peking University Series in Mathematics  Vol. 6,  55-88,  World Scientific (2014)

\bibitem{aksamit/choulli/jeanblanc} Aksamit, A., Choulli, T., and Jeanblanc, M.:\newblock{Decomposition of Random Times and their Classification}, \newblock{Preprint of University of Alberta and Evry-Val d'Essone University}, 2013.


\bibitem{arrow/debreu} Arrow, K. J.,  Debreu, G.:  Existence of an equilibrium for a competitive economy. Econometrica, 265-290 (1954).

\bibitem{amendingerimkellerschweizer98} Amendinger, J., Imkeller, P.,  Schweizer, M.: Additional logarithmic utility of an insider. Stochastic processes and their applications, 75(2), 263-286 (1998).

\bibitem{blackscholes1973} Black, F., Scholes, M. (1973). The pricing of options and corporate liabilities. The journal of political economy, 637-654.





\bibitem{choullistricker07} Choulli, T., Stricker, C., and Li J.: \newblock{Minimal Hellinger martingale measures of order $q$.} Finance and Stochastics 11.3 (2007): 399-427.


\bibitem{Choulli2012}Choulli  T. and Schweizer M., \newblock{LlogL Stability of equivalent $\sigma$-martingale density under equivalent change of measures,
} to appear in \newblock{Stochastics}, 2013.

\bibitem{choulli/deng/ma}Choulli  T., Deng, J. and Ma, J. \newblock{How non-arbitrage, viability and num\'eraire portfolio are related.} Accepted in  Finance and Stochastics, 2014. 

\bibitem{coxross76} Cox, J. C.,  Ross S.A.:  The valuation of options for alternative stochastic processes. Journal of financial economics,  145-166  (1976).

\bibitem{duffiebook2010} Duffie, D.:  Dynamic asset pricing theory. Princeton University Press (2010).

\bibitem{duffiehuang86} Duffie, D.,  Huang, C.F.:  Multiperiod security markets with differential information: martingales and resolution times.  Journal of Mathematical Economics 15.3 (1986): 283-303.


\bibitem{dm2} Dellacherie, C. and Meyer, P-A.,
\newblock{Probabilit{\'e}s et Potentiel, chapitres V-VIII},Hermann, Paris, 1980,
English translation : Probabilities and  Potentiel, chapters
V-VIII, North-Holland, (1982).

\bibitem{DMM}  Dellacherie, M., Maisonneuve, B. and  Meyer, P-A.  (1992),
{Probabilit\'es et Potentiel, chapitres XVII-XXIV: Processus de
Markov (fin), Compl\'ements de calcul stochastique}, Hermann, Paris.

\bibitem{debreu1959} Debreu, G. (1959). Theory of value: An axiomatic analysis of economic equilibrium (No. 17). Yale University Press.


\bibitem{delbaenschameyer1994} Delbaen, F.,  Schachermayer, W. (1994). A general version of the fundamental theorem of asset pricing. Mathematische annalen, 300(1), 463-520.







\bibitem{delbaenschameyer1998} Delbaen, F.,  Schachermayer, W. (1998). A simple counterexample to several problems in the theory of asset pricing. Mathematical Finance, 8(1), 1-11.


\bibitem{dybvigross1989}    Dybvig, P. H.,  Ross, S. A. (1989). Arbitrage in The New Palgrave: Finance, ed. J. Eatwell, M. Milgate and P. Newman.

\bibitem{DelbaenSchachermayer1994} Delbaen F. and Schachermayer W., \newblock{A general version of the fundamental theorem of asset pricing}, \newblock{Mathematische Annalen}, 1994.

\bibitem{fjs}Fontana, C. and Jeanblanc, M. and   Song, S. (2013)\newblock{On arbitrages arising with honest times}. To appear in Finance and Stochastics

\bibitem{harrisonkreps1979} Harrison, J. M.,  Kreps, D. M. (1979). Martingales and arbitrage in multiperiod securities markets. Journal of Economic theory, 20(3), 381-408.

\bibitem{harrisonpliska1981}    Harrison, J. M.,  Pliska, S. R. (1981). Martingales and stochastic integrals in the theory of continuous trading. Stochastic processes and their applications, 11(3), 215-260.

\bibitem{Yanbook}He, S. W., Wang, C. K.,  Yan, J. A.: Semimartingale theory and stochastic calculus. CRC Press (1992).


\bibitem{Jacod} Jacod, J., \newblock{Calcul Stochastique et Probl\`emes de Martingales}, Lecture Notes in Mathematics, Vol. 714,(1979).

\bibitem{JS}  Jacod, J. and   Shiryaev, A.N.,
\newblock{Limit theorems for stochastic Processes},
Springer Verlag,
2003


\bibitem{Jeu}  Jeulin, T.  (1980), {Semi-martingales et Grossissement d'une Filtration}, Lecture Notes in Mathematics, vol. 833, Springer, Berlin - Heidelberg - New York.
\bibitem{Jeuyor} Jeulin, Th. and Yor, M., (1978)
Grossissement d'une filtration et semi-martingales~: formules
explicites,  S{\'e}minaire de Probabilit\'es XII, Dellacherie, C.
and Meyer, P-A. and Weil, M., Lecture Notes in Mathematics 649,
78-97,  Springer-Verlag.

\bibitem{kabanov} Kabanov, Y.:  On the FTAP of Kreps-Delbaen-Schachermayer.  In: Kabanov, Y. et al. (eds.):
Statistics and Control of Stochastic Processes: The Liptser
Festschrift, pp. 191--203. World Scientific, Singapore (1997)

\bibitem{kallsen04} Kallsen, J.: \newblock{$\sigma$-localization and $\sigma$-martingales.} Theory of Probability \& Its Applications 48.1 (2004): 152-163.

\bibitem{KaratzasKardaras2007} Karatzas I. and Kardaras C., \newblock{The numeraire portfolio in semimartingale financial models}, \newblock{Finance and Stochastics}, 447-493, 2007.

\bibitem{karadarasHonest2010} Kardaras, C.: A time before whihc insiders would not undertake riks.arXiv:2010.1961v2.

\bibitem{kardaras12} Kardaras, C.:  Market viability via absence of arbitrage of the first kind. Finance and stochastics, 16(4), 651-667  (2012).

\bibitem{Kardaras} Kardaras, K.(2014), On the characterisation of honest times that avoid all stopping
times. Stochastic Processes and their 124, 373--384.

\bibitem{kreps1981} Kreps, D. M. (1981). Arbitrage and equilibrium in economies with infinitely many commodities. Journal of Mathematical Economics, 8(1), 15-35.

\bibitem{kohatsusulem06} Kohatsu--Higa, A.,  Sulem, A.: Utility maximization in an insider influenced market. Mathematical Finance, 16(1), 153-179  (2006).


\bibitem{loewensteinwillard00} Loewenstein, M.,  Willard, G. A.:  Local martingales, arbitrage, and viability Free snacks and cheap thrills. Economic Theory, 16(1), 135-161  (2000).


\bibitem{leventalskorohod1995} Levental, S.,  Skorohod, A. V. (1995). A necessary and sufficient condition for absence of arbitrage with tame portfolios. The Annals of Applied Probability, 906-925.

\bibitem{merton1973} Merton, R. C. (1973). Theory of rational option pricing. The Bell Journal of Economics and Management Science, 141-183.


\bibitem{pikovskykaratzas96} Pikovsky, I.,  Karatzas, I.:  Anticipative portfolio optimization. Advances in Applied Probability, 1095-1122  (1996).
\bibitem{rokh}   Rokhlin, D.B., On the existence of an equivalent supermartingale density for a fork-convex family
of random processes, Math. Notes 87 (2010), no. 34, 556563.













\bibitem{Takaoka} Takaoka, K., \newblock{ A Note on the Condition of No Unbounded Profit with Bounded Risk}, to appear in:  Finance and Stochastics , 2012.

\bibitem{Yor} Yor,  M., \newblock{Grossissement d'une filtration et semi-martingales: th\'eor\`emes g\'en\'eraux} S\'eminaire de probabilit\'es, tome XII,
 Lecture Notes in Mathematics, Vol. 649, (1978), p.61-69.




\end{thebibliography}
\end{document}